\documentclass[aps,pra,11pt,amsmath,amssymb,tightenlines,superscriptaddress,notitlepage,nofootinbib,onecolumn]{revtex4-2}
\pdfoutput=1

\usepackage{graphicx}
\usepackage[dvipsnames,svgnames]{xcolor}
\usepackage{bm,bbm}%
\usepackage{braket}
\usepackage{amsthm,amsmath,amssymb,amsbsy}
\usepackage{enumerate}
\usepackage[shortlabels]{enumitem}
\usepackage{booktabs,multirow,moresize}
\usepackage{mathtools,mathrsfs,bbding}
\usepackage{microtype}

\let\counterwithin\relax
\usepackage{chngcntr}

\usepackage[T1]{fontenc}
\usepackage[utf8]{inputenc}

\definecolor{darkred}{rgb}{0.57,0,0.12}
\usepackage[colorlinks=true, linkcolor=MidnightBlue, urlcolor=MidnightBlue, citecolor=MidnightBlue]{hyperref}

\usepackage{cleveref}

\let\nc\newcommand

\let\mathscr\relax
\usepackage{newpxtext,newpxmath}
\usepackage[cal=cm]{mathalfa}

\hypersetup{
    bookmarksnumbered=true, 
    breaklinks=true,
    unicode=false, 
    pdfstartview={FitH}, 
    pdftitle={}, 
    pdfnewwindow=true, 
    colorlinks=true, 
    allcolors=MidnightBlue!70!black!70!TealBlue
}

\nc{\note}[1]{{\color{blue!90!black} #1}}
\nc{\noteb}[1]{{\color{red!80!black}\textbf{#1}}}

\DeclareMathOperator{\Tr}{Tr}

\DeclareMathOperator{\NN}{NN}
\DeclareMathOperator{\supp}{supp}

\DeclareMathOperator{\ran}{ran}

\let\Re\relax
\DeclareMathOperator{\Re}{Re}

\nc{\psic}{\psi^{c}}
\nc{\two}[1]{\underline{2^{d-#1}}}

\nc{\Hanc}{\mathcal{H}_{\text{anc}}}
\nc{\psianc}{\psi_{\text{anc}}}
\nc{\lampow}{\lambda^{1/d}}

\nc{\norm}[2]{\left\lVert#1\right\rVert_{#2}}
\nc{\proj}[1]{\ket{#1}\!\bra{#1}}
\nc{\pro}[1]{#1 #1^\dagger}
\nc{\lnorm}[2]{\left\lVert#1\right\rVert_{\ell_{#2}}}

\nc{\VV}{V^{\,\C}_\Sp}

\let\sect\S
\renewcommand{\S}{\mathcal{S}}
\nc{\R}{\mathcal{R}}
\nc{\W}{\mathcal{W}}
\nc{\T}{\mathcal{T}}
\nc{\B}{\mathcal{B}}
\nc{\C}{\mathcal{C}}
\nc{\U}{\mathcal{U}}
\nc{\E}{\mathcal{E}}
\DeclareMathOperator*{\EE}{\mathbb{E}}
\nc{\K}{\mathcal{K}}
\nc{\V}{\mathcal{V}}
\nc{\X}{\mathcal{X}}
\nc{\F}{\mathcal{F}}
\nc{\Z}{\mathcal{Z}}
\nc{\G}{\mathcal{G}}
\nc{\D}{\mathcal{D}}
\newcommand{\LL}{\mathcal{L}}
\nc{\Y}{\mathcal{Y}}

\nc{\M}{\mathcal{M}}
\nc{\N}{\mathcal{N}}
\nc{\I}{\mathcal{I}}
\nc{\Q}{\mathbb{Q}}
\nc{\RR}{\mathbb{R}}
\nc{\CC}{\mathbb{C}}

\nc{\HH}{\mathbb{H}}
\nc{\MM}{\smash{\mathbb{M}}}
\nc{\DD}{\mathbb{D}}
\renewcommand{\NN}{\mathbb{N}}

\nc{\J}{\mathcal{J}}

\nc{\Bures}{\mspace{1mu}}

\nc{\FF}{\mathbb{F}}

\nc{\CPTP}{{\mathrm{CPTP}}}
\nc{\CP}{{\mathrm{CP}}}
\nc{\CPTN}{{\mathrm{CPTN}}}

\nc{\ext}[1]{\operatorname{ext}\left(#1\right)}

\renewcommand{\*}{\textup{*}}
\nc{\<}{\left\langle}
\renewcommand{\>}{\right\rangle}

\let\bar\relax
\nc{\lset}{\left\{\vphantom{\big|}}
\nc{\bar}{\;\middle|\;}
\nc{\rset}{\vphantom{\big|}\right\}}

\nc{\ve}{\varepsilon}

\nc{\id}{\mathbbm{1}}
\nc{\idc}{\mathrm{id}}

\nc{\RHS}{(\mathrm{RHS})}

\DeclareMathOperator{\primal}{Primal}
\DeclareMathOperator{\dual}{Dual}
\DeclareMathOperator{\primald}{Primal\*}

\theoremstyle{plain}

\newtheorem{theorem}{Theorem}

\newtheorem{proposition}[theorem]{Proposition}
\newtheorem{corollary}[theorem]{Corollary}
\newtheorem{lemma}[theorem]{Lemma}

\theoremstyle{definition}

\newtheorem*{remark}{Remark}
\newtheorem{definition}[theorem]{Definition}

\let\oldproofname\proofname
\renewcommand{\proofname}{\rm\bf{\oldproofname}}

\makeatletter
\renewenvironment{proof}[1][\proofname]{\par
\pushQED{\qed}%
\normalfont \topsep6\p@\@plus6\p@\relax
\trivlist
\item\relax
{\bfseries  
#1\@addpunct{.}}\hspace\labelsep\ignorespaces 
}{%
\popQED\endtrivlist\@endpefalse
}
\makeatother

\nc{\ketbra}[2]{\ket{#1}\!\bra{#2}}

\nc{\prob}{\mathrm{prob}}

\makeatletter
\let\save@mathaccent\mathaccent
\newcommand*\if@single[3]{%
  \setbox0\hbox{${\mathaccent"0362{#1}}^H$}%
  \setbox2\hbox{${\mathaccent"0362{\kern0pt#1}}^H$}%
  \ifdim\ht0=\ht2 #3\else #2\fi
  }
\newcommand*\rel@kern[1]{\kern#1\dimexpr\macc@kerna}
\newcommand*\widebar[1]{\wide@bar{#1}{0}}
\newcommand*\wide@bar[2]{\if@single{#1}{\wide@bar@{#1}{#2}{1}}{\wide@bar@{#1}{#2}{2}}}
\newcommand*\wide@bar@[3]{%
  \begingroup
  \def\mathaccent##1##2{%
    \let\mathaccent\save@mathaccent
    \if#32 \let\macc@nucleus\first@char \fi
    \setbox\z@\hbox{$\macc@style{\macc@nucleus}_{}$}%
    \setbox\tw@\hbox{$\macc@style{\macc@nucleus}{}_{}$}%
    \dimen@\wd\tw@
    \advance\dimen@-\wd\z@
    \divide\dimen@ 3
    \@tempdima\wd\tw@
    \advance\@tempdima-\scriptspace
    \divide\@tempdima 10
    \advance\dimen@-\@tempdima
    \ifdim\dimen@>\z@ \dimen@0pt\fi
    \rel@kern{0.6}\kern-\dimen@
    \if#31
      \overline{\rel@kern{-0.6}\kern\dimen@\macc@nucleus\rel@kern{0.4}\kern\dimen@}%
      \advance\dimen@0.4\dimexpr\macc@kerna
      \let\final@kern#2%
      \ifdim\dimen@<\z@ \let\final@kern1\fi
      \if\final@kern1 \kern-\dimen@\fi
    \else
      \overline{\rel@kern{-0.6}\kern\dimen@#1}%
    \fi
  }%
  \macc@depth\@ne
  \let\math@bgroup\@empty \let\math@egroup\macc@set@skewchar
  \mathsurround\z@ \frozen@everymath{\mathgroup\macc@group\relax}%
  \macc@set@skewchar\relax
  \let\mathaccentV\macc@nested@a
  \if#31
    \macc@nested@a\relax111{#1}%
  \else
    \def\gobble@till@marker##1\endmarker{}%
    \futurelet\first@char\gobble@till@marker#1\endmarker
    \ifcat\noexpand\first@char A\else
      \def\first@char{}%
    \fi
    \macc@nested@a\relax111{\first@char}%
  \fi
  \endgroup
}
\makeatother

\nc{\wt}{\widetilde}
\renewcommand{\ol}{\widebar}

\nc{\rhos}{\rho'}

\let\epsilon\varepsilon

\usepackage{scalerel,stackengine}
\newcommand\pig[1]{\scalerel*[5pt]{\big#1}{%
  \ensurestackMath{\addstackgap[1.5pt]{\big#1}}}}

\newcommand{\deff}[1]{\textbf{\emph{#1}}}

\nc{\ba}{\begin{equation}\begin{aligned}}
\nc{\ea}{\end{aligned}\end{equation}}

\let\textleq\relax
\let\textgeq\relax
\let\texteq\relax

\newcommand{\texteq}[1]{\stackrel{\mathclap{\mbox{\scriptsize #1}}}{=}}
\newcommand{\textleq}[1]{\stackrel{\mathclap{\mbox{\scriptsize #1}}}{\leq}}

\newcommand{\textgeq}[1]{\stackrel{\mathclap{\mbox{\scriptsize #1}}}{\geq}}

\usepackage[most,breakable]{tcolorbox}
\renewenvironment{boxed}[1][white]%
  {\expandafter\ifstrequal\expandafter{#1}{filled}{\begin{tcolorbox}[colback=MidnightBlue!70!black!70!TealBlue!5!white,colframe=MidnightBlue!70!black!70!TealBlue!30!white!80!black,breakable=false,enhanced,left=5.75pt,right=5.75pt,grow sidewards by=10pt]}{\begin{tcolorbox}[colback=white,colframe=gray!15,breakable,enhanced,left=5.75pt,right=5.75pt,grow sidewards by=10pt]}}%
  {\end{tcolorbox}}

\nc{\rrho}{\widetilde{\rho}}

\nc{\abs}[1]{\left|#1\right|}

\usepackage{stackengine}
\nc{\univ}{2\*}

\let\ccc\c
\renewcommand{\c}[1]{ \textrm{\scriptsize \textcolor{white!50!black}{ #1 }}}

\renewcommand{\thesubsection}{\arabic{subsection}}

\newcommand{\redd}[1]{{\color{red!70!black} #1}}

\usepackage{tikz-cd}

\makeatletter
\def\fnum@figure{\textbf{Figure}\nobreakspace\textbf{\thefigure}}
\renewcommand{\p@subsection}{\thesection.}
\renewcommand{\p@subsubsection}{\thesection.\thesubsection.}
\makeatother

\DeclareMathOperator{\opt}{opt}

\usepackage{xstring}
\usepackage{xparse}

\NewDocumentCommand{\s}{m m o}{%
 \lowercase{\def\sdist{#2}}%
 ^{#1,\mspace{2mu}%
 \IfEqCase{\sdist}{%
  {t}{%
    T%
  }%
  {p}{%
    P%
  }%
  {m}{%
    \MM%
  }%
 }%
 \IfNoValueTF{#3}{}{%
    \lowercase{\def\sarr{#3}}%
    ,%
    \IfEqCase{\sarr}{%
      {up}{ \mspace{.5mu}\smash\uparrow }%
      {uparrow}{ \mspace{.5mu}\smash\uparrow }%
      {down}{ \mspace{1mu}\smash\downarrow }%
      {downarrow}{ \mspace{1mu}\smash\downarrow }%
    }%
  }%
 }
 \IfNoValueTF{#3}{}{%
  \mathchoice{\vphantom{^{\uparrow}}}{}{}{}%
 }%
}%
\NewDocumentCommand{\Ds}{m m m}{%
  D_{#1}\s{#2}{#3}%
}
\NewDocumentCommand{\Dmax}{m m}{%
  D_{\max}\s{#1}{#2}%
}
\NewDocumentCommand{\Qs}{m m m}{%
  Q_{#1}\s{#2}{#3}%
}
\NewDocumentCommand{\Qmax}{m m}{%
  Q_{\max}\s{#1}{#2}%
}
\NewDocumentCommand{\Hs}{m m m O{up}}{%
  H_{#1}\s{#2}{#3}[#4]%
}
\NewDocumentCommand{\Hmin}{m m O{up}}{%
  H_{\min}\s{#1}{#2}[#3]%
}

\hypersetup{
    pdftitle={Rethinking quantum smooth entropies: Tight one-shot analysis of quantum privacy amplification}
}

\begin{document}

\title{Rethinking quantum smooth entropies:\texorpdfstring{\\}{ }Tight one-shot analysis of quantum privacy amplification}

\author{Bartosz Regula}
\email{bartosz.regula@gmail.com}
\affiliation{Mathematical Quantum Information RIKEN Hakubi Research Team, RIKEN Pioneering Research Institute (PRI) and RIKEN Center for Quantum Computing (RQC), Wako, Saitama, Japan}

\author{Marco Tomamichel}
\email{marco.tomamichel@nus.edu.sg}
\affiliation{Department of Electrical and Computer Engineering, National University of Singapore, Singapore}
\affiliation{Centre for Quantum Technologies, National University of Singapore, Singapore}

\begin{abstract}
We introduce an improved one-shot characterisation of randomness extraction against quantum side information (privacy amplification), strengthening known one-shot bounds and providing a unified derivation of the tightest known asymptotic constraints.  
Our main tool is a new class of smooth conditional entropies defined by lifting classical smooth divergences through measurements.
A key role is played by the measured smooth R\'enyi relative entropy of order 2, which we show to admit an equivalent variational form: it can be understood as allowing for smoothing over not only states, but also non-positive Hermitian operators.
Building on this, we establish a tightened leftover hash lemma, significantly improving over all known smooth min-entropy bounds on extractable randomness and recovering the sharpest classical achievability results. We extend these methods to decoupling, the coherent analogue of privacy amplification, obtaining a corresponding improved one-shot bound.
Relaxing our smooth entropy bounds leads to one-shot achievability results in terms of measured R\'enyi divergences, tightening the bounds of~[Dupuis, \href{https://ieeexplore.ieee.org/document/10232924}{IEEE T-IT 69, 7784 (2023)}] and recovering the state-of-the-art asymptotic i.i.d.\ error exponents.
We show an approximate optimality of our results by giving a matching one-shot converse bound up to additive logarithmic terms.
This yields an optimal second-order asymptotic expansion of privacy amplification under trace distance, establishing a tighter one-shot achievability result than previously shown in~[Shen, Gao and Cheng, \href{https://ieeexplore.ieee.org/document/10387493}{IEEE T-IT 70, 5077 (2024)}] and proving its optimality for all hash functions.
\end{abstract}

\maketitle


\tableofcontents

\section{Introduction}

The task of privacy amplification, concerned with extracting uniform and secret randomness in the presence of an adversary with side information, is a fundamental step in establishing the security of quantum key distribution~\cite{renner_2005,portmann_2022,pirandola_2020}. 
In cryptographic contexts, privacy amplification is fundamentally connected with the notion of min-entropy $H_{\min}(X|E)$ between a random variable $X$ and the adversary system $E$, which quantifies the probability that the adversary can guess the true value of $X$. Importantly, the knowledge of the source distribution from which randomness is to be extracted is typically limited: it is precisely the min-entropy of the source that is known or can be estimated, and randomness extractors are expected to function optimally while knowing only this single value. Designing universal randomness extractors that connect the min-entropy of any source with its extractable randomness is thus crucial for cryptographic applications~\cite{nisan_1996,bennett_1995}.

A prominent result in privacy amplification is the leftover hash lemma~\cite{bennett_1988,impagliazzo_1989,bennett_1995}, which showed that privacy amplification against classical adversaries can be accomplished through universal hash functions, connecting the achievable deviation from uniform (measured in total variation distance) with the collision entropy $H_2(X|E)$ of the source, which can then be easily bounded by its min-entropy $H_{\min}(X|E)$. 
A seemingly simple but extremely consequential realisation is that one can employ \emph{smoothing}, that is, optimise over all distributions which approximate the source, for improved performance. As formalised by Renner and Wolf~\cite{renner_2004}, this smooth min-entropy $H_{\min}^\ve(X|E)$ tightly characterises privacy amplification at the one-shot level. 
In the asymptotic i.i.d.\ limit, smooth entropies converge to the conditional entropy $H(X|E)$, the asymptotic rate of extractable randomness~\cite{renner_2004}. 
Studying higher-order refinements of such i.i.d.\ bounds, Hayashi then observed that, while $H_{\min}^\ve(X|E)$ does lead to a tight second-order asymptotic expansion of the achievable rates~\cite{hayashi_2016-2,hayashi_2016-3}, smoothing the collision entropy itself as $H_2^\ve(X|E)$ can lead to further improvements in large-deviation analysis~\cite{hayashi_2013,hayashi_2016-2}.

The leftover hash lemma admits a natural generalisation to randomness extraction against quantum side information, where the adversary is not assumed to be classical~\cite{renner_2005,renner_2005-1,tomamichel_2011}. This extension through the quantum collision entropy $\wt{H}_2(X|E)$, later understood to be part of a broader family known as the sandwiched R\'enyi entropies~\cite{muller-lennert_2013},  now underpins the security proofs of quantum key distribution~\cite{portmann_2022,metger_2022,arqand_2025}. 
The idea of smoothing also found fruitful applications in the quantum setting~\cite{tomamichel_2011,tomamichel_2013}. Indeed, generalisations of the leftover hash lemma such as decoupling~\cite{horodecki_2005-1,dupuis_2014} 
led to smooth entropies finding widespread use in quantum information theory more broadly~\cite{dupuis_2014,berta_2011}. 

However, in the analysis of quantum privacy amplification, smoothing encountered some limitations. The one-shot bounds obtained using the standard toolkit of quantum smooth entropies~\cite{renner_2005,tomamichel_2011} did not match the classical results.  Hayashi~\cite{hayashi_2014} also realised that the methods used in his large-deviation analysis do not readily extend to quantum contexts, leading to suboptimal bounds for quantum adversaries. Due to the difficulties in analysing trace distance between quantum states, improved asymptotic error estimates~\cite{hayashi_2015,li_2023} and a tight second-order analysis of privacy amplification~\cite{tomamichel_2013} were only possible under a modified distance criterion, namely the purified distance, departing from the standard choice of trace distance as the operational security measure in cryptography. 
On the one hand, these issues motivated the development of completely different approaches. Two notable examples are the achievability bound of Dupuis~\cite{dupuis_2023} based on sandwiched R\'enyi divergences, which extended the leftover hash lemma through norm interpolation techniques in Schatten spaces, as well as an improved second-order achievability result of Shen, Gao, and Cheng~\cite{shen_2024} that refined Hayashi's techniques based on spectral pinching~\cite{hayashi_2016}. On the other hand, however, the discrepancy with classical results and the lack of a tight one-shot characterisation under trace distance suggested that many aspects of the analysis of quantum privacy amplification could be improved. The work~\cite{dupuis_2023} made the argument that perhaps one-shot results should be approached directly with R\'enyi divergences rather than smooth entropies. Here we instead argue that smoothing does provide perfectly accurate constraints --- we had simply not been using the right quantities for this~task.

\subsection{Summary of results}

We introduce a new approach to the study of quantum privacy amplification through a modified notion of smoothing, leading to the tightest known formulation of the leftover hash lemma and in particular its connection with smooth entropies. 
At the heart of our construction is the observation that the operation of smoothing and the procedure of lifting classical divergences to quantum ones by measurements do not commute (see Figure~\ref{fig:diagram}). 
Reversing the conventional order and defining a class of smooth divergences based on classical smoothing before measuring, we then obtain a new family of measured smooth entropies that we show to improve on prior approaches to randomness extraction with quantum side information.  
The resulting class of smooth quantum divergences exhibits a number of useful properties, including  significantly improved 
error scaling and a capability to directly generalise classical properties and inequalities to quantum states. It unifies also other one-shot divergences that found applications in the study of quantum information, such as the hypothesis testing relative entropy~\cite{wang_2012,buscemi_2010} and an information-spectrum variant of smooth max-relative entropy~\cite{datta_2015}.
Our approach also identifies not the conventional choice of sandwiched R\'enyi divergences, but rather the measured (minimal) R\'enyi divergences $D_\alpha^{\MM}$ as being the optimal choice for characterising privacy amplification.

\begin{figure}
\[\begin{tikzcd}[row sep=huge]
D_{\max}(p\|q) \arrow{rr}{\text{smoothing}} 
\arrow[swap,"{ \text{ measuring}}"']{d} 
&[5em]  &[-.5em] \Dmax{\ve}{T}(p\|q) \arrow["{ \begin{subarray}{l}\text{ measuring}\\[5pt]\end{subarray}}"']{d}\\[2.1em]
\begin{subarray}{l}\displaystyle D_{\max}^{\MM}(\rho\|\sigma)\\[1pt]\displaystyle\hspace*{-9pt} = D_{\max}(\rho\|\sigma)\end{subarray} \arrow{r}{\text{smoothing}} & \Dmax{\ve}{T}(\rho\|\sigma) \hspace{5pt} \redd{=\joinrel\neq\joinrel=} \;\; \hspace*{-30pt} & \Dmax{\ve\vphantom{T}}{M} (\rho\|\sigma)
\end{tikzcd}
\]
\caption{Smooth entropies, defined via R\'enyi divergences such as the collision relative entropy $D_2$ or the max-relative entropy $D_{\max}$, underlie the precise description of privacy amplification. When extending their definitions to quantum states, the resulting quantities differ depending on the order in which the operations of smoothing (optimising a divergence over distributions in an $\ve$-ball of trace distance) and measuring (taking the minimal quantum extension of a classical divergence by maximising it over all measurement channels) are applied. In this work we show that it is the class of divergences defined by first smoothing the classical divergence and only then lifting it to quantum states, represented by $\Dmax{\ve}{M}$ in the diagram, that tightly characterises quantum randomness extraction.}
\label{fig:diagram}
\end{figure}
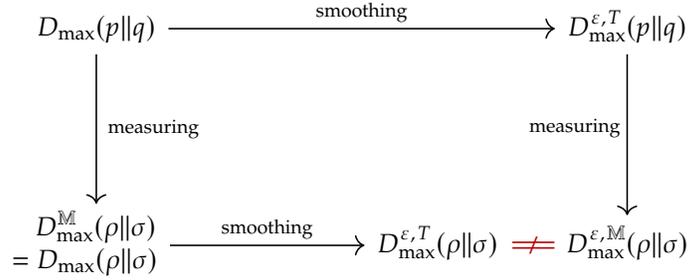

A representative quantity in this class of divergences, and a central element of our technical developments, is the measured smooth collision divergence
\begin{equation}\begin{aligned}
  \Ds{2}{\ve}{M}(\rho\|\sigma) \coloneqq \sup_{\M \in \MM} \Ds{2}{\ve}{T} (\M(\rho)\|\M(\sigma))
\end{aligned}\end{equation}
(see Section~\ref{sec:smollision} for details). 
Here, the optimisation is over all measurement channels $\M$, and $\Ds{2}{\ve}{T} (\M(\rho)\|\M(\sigma))$ denotes the fully classical smooth R\'enyi divergence of order 2 between the probability distributions resulting from the measurement, using the conventional classical smoothing with total variation (trace) distance. We stress that the smoothing here is effectively with respect to the post-measurement statistics, which may seem not directly connected to the distance to the original state $\rho$. However, as one of our main technical contributions, in Theorem~\ref{thm:D2_smoothed_forms} we show that by employing a suitable matrix weighted inner product known as the Bures product~\cite{braunstein_1994,lesniewski_1999,petz_2011} combined with convex analytic techniques, this quantity can be expressed as a quantum smooth divergence --- one that, instead of quantum states, allows for smoothing over more general Hermitian operators:
\begin{equation}\begin{aligned}\label{eq:intro_variational}
  \Ds{2}{\ve}{M}(\rho\|\sigma) &= \inf \lset D_2^{\MM} (R\|\sigma) \bar R = R^\dagger, \;  R \leq \rho, \; \|\rho - R\|_+ \leq \ve \rset,
\end{aligned}\end{equation}
where $\|\cdot\|_+$ denotes the generalised trace distance.
In Section~\ref{sec:smollision} we argue that structural properties of quantum states point towards the necessity to use such decompositions that involve non-positive operators instead of only positive ones, contrasting with the classical case. 
This immediately gives the idea that will underlie our approach to leftover hashing: 
to optimally approximate a quantum state $\rho$, one should do so by a Hermitian operator $R$, rather than a quantum state like in conventional frameworks.

This characterisation of $\Ds{2}{\ve}{M}$ as both a measured-smoothed and a Hermitian-smoothed divergence mirrors a recent finding of~\cite{regula_2025}, where the same property was shown for a variant of the smooth max-relative entropy $D_{\max}(\rho\|\sigma)$, i.e.\ R\'enyi divergence of order~$\infty$. This connection forms an important part of several of our results, and motivates the definition of a more appropriate notion of smooth min-entropy that differs from standard constructions. 
Specifically, for a bipartite state $\rho_{XE}$ we define the measured smooth  collision entropy  and min-entropy  as
\begin{equation}\begin{aligned}
  H_{2}\s{\ve}{M} (X|E)_{\rho} = - \inf_{\sigma_E} \Ds{2}{\ve}{M}(\rho_{XE} \| \id_X \otimes \sigma_E), \quad\; H_{\min}\s{\ve}{M} (X|E)_{\rho} = - \inf_{\sigma_E} \Dmax{\ve}{M}(\rho_{XE} \| \id_X \otimes \sigma_E).
\end{aligned}\end{equation}
We use these quantities to establish both improved achievability (leftover hash) and converse results. Our Theorem~\ref{thm:privacy_amp} and Proposition~\ref{prop:converse} together show that the maximal number of bits of randomness that can be extracted from the given state up to trace distance error $\ve$, denoted $\ell_\ve(\rho_{XE})$, satisfies
\begin{equation}\begin{aligned}
 H_{\min}\s{\ve-\mu}{M} (X|E)_{\rho} - \log \frac{1}{4 \mu^2} &\leq H_{2}\s{\ve-\mu}{M} (X|E)_{\rho} - \log \frac{1}{4 \mu^2}\\
 &\leq \ell_\ve(\rho_{XE})\\
&\leq H_{\min}\s{\ve+\delta}{M} (X|E)_{\rho} + \log\frac{\ve+\delta}{\delta}.
\end{aligned}\end{equation}
This improves on every achievability result in the literature that relied on smooth entropies, including conventional variants of the quantum leftover hash lemma~\cite{renner_2005,tomamichel_2011,anshu_2020}, attempts to extend classical smoothing approaches through pinching~\cite{hayashi_2014}, or the second-order achievability result of~\cite{shen_2024}. 
For classical side information, our result recovers the tightest formulation of the classical leftover hash lemma~\cite{bennett_1995,renner_2005-2,hayashi_2013}.
The converse shows this to be approximately tight, in particular establishing that $H_{\min}\s{\ve}{M} (X|E)_{\rho}$ tightly determines the second-order expansion of quantum privacy amplification under trace distance (Corollary~\ref{cor:second_order_tight}), proving the optimality of the asymptotic achievability results shown in~\cite{shen_2024}.

The quantitative improvement provided by this result can be realised by comparing it with previous iterations of the leftover hash lemma with quantum side information, which led to bounds of the form~\cite{tomamichel_2011}
\begin{equation}\begin{aligned}
    H_{\min}\s{(\ve-\mu)/2}{P} (X|E)_{\rho} - \log \frac{1}{4\mu^2}  \,\leq\, \ell_\ve(\rho_{XE}) \,\leq\, H_{\min}\s{2\sqrt\ve}{P} (X|E)_{\rho}\,,
\end{aligned}\end{equation}
where $H_{\min}^{\ve,\,P}$ denotes the conventional min-entropy smoothed with purified distance~\cite{tomamichel_2011}. A more recent modified variant of `partial' smoothing can improve the achievability result from $H_{\min}^{(\ve-\mu)/2}$ to $H_{\min}^{\ve-\mu}$~\cite{anshu_2020}. 
Our results, stated in terms of the same purified distance smooth entropy, give tighter constraints as
\begin{equation}\begin{aligned}
   H_{\min}^{\sqrt{\vphantom{b}\ve-\mu},\,P} (X|E)_{\rho} - \log \frac{1}{4 \mu^3} \,\leq\, \ell_\ve(\rho_{XE})\, \leq\, H_{\min}^{\sqrt{\vphantom{b}\ve},\,P} (X|E)_{\rho} + \log\frac{1}{1-\ve}\,.
\end{aligned}\end{equation}
This shows that the prior achievability results indicated an inaccurate scaling --- the size of randomness that can be extracted scales with $\smash{H_{\min}^{\sqrt\ve,\,P}}$, and not with $H_{\min}^{\ve,\,P}$, and hence much larger amounts of randomness can be proven to be extractable in our approach.  In an asymptotic sense, such scaling follows already from the second-order bounds of~\cite{shen_2024}, but their reliance on spectral pinching means that those bounds are looser in the one-shot and finite-copy regimes, and our result overcomes the need for such relaxations.

By developing one-shot bounds that closely connect the measured smooth collision entropy $H_{2}\s{\ve}{M}$ with measured R\'enyi entropies $H_\alpha^{\MM}$, in Theorem~\ref{thm:renyi_achiev} we show that our result implies a bound on the achievable error of privacy amplification as
\begin{equation}\begin{aligned}
  \ve \leq \exp \left( - \sup_{\alpha \in (1,2]} \frac{\alpha-1}\alpha \Big( H_\alpha^{\MM}(X|E)_\rho - \ell_\ve(\rho_{XE}) \Big)\right).
\end{aligned}\end{equation}
This strengthens an achievability result established in~\cite{dupuis_2023} through complex interpolation techniques, stated there in terms of the looser sandwiched R\'enyi divergences.
Our approach also naturally gives converse bounds on the achievable error exponent (Corollary~\ref{cor:error_exp_general}), recovering classical converse bounds studied in~\cite{hayashi_2013,hayashi_2016-3}.

We further show that our proof approach naturally extends to the more general setting of decoupling~\cite{horodecki_2005-1}, which is a form of coherent and fully quantum randomness extraction that finds use in a wide number of achievability results in quantum information. In Theorem~\ref{thm:decoupling_smollision} we use our techniques to show an analogous achievability bound in terms of the measured smooth R\'enyi divergence of order 2, improving on prior one-shot formulations of decoupling~\cite{dupuis_2014,dupuis_2023,shen_2023}.

In addition to the above, we establish a number of properties of the newly defined quantity $\Ds{2}{\ve}{M}(\rho\|\sigma)$ as well as its connections with other smooth divergences. 
The one-shot inequalities that relate it with the measured smooth max-relative entropy $\Dmax{\ve}{M}$ and the hypothesis testing relative entropy $D^\ve_H$ (Lemmas~\ref{lem:secondorder_equiv} and~\ref{lem:D2_upper_exponent_bound}) find use in many of our derivations, leading to new results also in the study of strong converse exponents of randomness extraction under trace distance (Proposition~\ref{prop:sc_bounds}). 

Overall, our results give new insights into both the one-shot performance and the asymptotic limits of quantum privacy amplification, showing the importance of the measured smooth R\'enyi divergences $\Ds{2}{\ve}{M}$ and $\Dmax{\ve}{M}$ in the analysis of this task. The fact that these previously unconsidered quantities lead to significantly tighter bounds compared to all prior approaches 
points to the family of measured smooth entropies as being the most appropriate generalisation of classical conditional entropies in the study of randomness extraction. 
We believe that this indicates the need to reconsider approaches to smoothing in quantum information more broadly, and we expect the ideas and techniques developed in this work to lead to improved analysis of many operational problems beyond the settings explicitly considered here.

\section{Preliminaries}

\subsection{Relative entropies and conditional entropies}

Throughout this paper, all operators are assumed to be acting on finite-dimensional Hilbert spaces; for simplicity, we often omit the underlying spaces from our notation. Unless otherwise stated, we will restrict our discussion to self-adjoint operators. We use Greek letters ($\rho$, $\sigma$\ldots) to refer to quantum states (positive semidefinite operators of trace one) or, more generally, positive operators that may not necessarily be normalised. We use lowercase Latin letters ($p$, $q$\ldots) to refer to classical probability distributions on some finite alphabet $\X$; with a slight abuse of notation, we often treat classical distributions as diagonal quantum states in some orthonormal basis $\{\ket{x}\}_x$, i.e.\ $p = \sum_x p(x) \proj{x}$.
The functions $\log$ and $\exp$ are taken to the same, but otherwise arbitrary, base. 

For classical distributions $p$ and $q$, the R\'enyi relative entropies of order $\alpha$ are given by
\begin{equation}\begin{aligned}
  D_\alpha(p\|q) \coloneqq \frac{1}{\alpha-1} \log \sum_x p(x)^{\alpha} \, q(x)^{1-\alpha}.
\end{aligned}\end{equation}
Here, $\alpha \in [0,\infty]$, with the cases of $\alpha \in \{0, 1, \infty\}$ understood as the respective limits in $\alpha$. 
The most commonly encountered quantum generalisations of this family are the Petz--R\'enyi relative entropies $\ol{D}_\alpha$~\cite{petz_1986} and the sandwiched R\'enyi relative entropies $\wt{D}_\alpha$~\cite{muller-lennert_2013,wilde_2014}, defined for any positive semidefinite operators $\rho, \sigma$ as\footnote{Negative powers here are taken on the support of an operator, and in particular for rank-deficient $\sigma$ the values can be understood as the limits as $\ve \to 0$ of the divergences evaluated at $\sigma + \ve \id$.  For $\alpha > 1$, the quantities diverge to $+\infty$ whenever $\supp(\rho) \not\subseteq \supp(\sigma)$.}
\begin{equation}\begin{aligned}\label{eq:quantum_Renyi_def}
\ol{D}_\alpha (\rho \| \sigma) &\coloneqq \frac{1}{\alpha-1} \log \Tr \left( \rho^\alpha \sigma^{1-\alpha}\right),  \\
\wt{D}_\alpha (\rho \| \sigma) &\coloneqq \frac{1}{\alpha-1} \log \Tr\! \left[ \left(\sigma^{\frac{1-\alpha}{2\alpha}} \rho \sigma^{\frac{1-\alpha}{2\alpha}}\right)^\alpha \right]. 
\end{aligned}\end{equation}
Another natural way to extend classical quantities is through measurements~\cite{donald_1986,hiai_1991}, with the resulting quantum divergences being the minimal extensions that satisfy the data processing inequality~\cite{hiai_2017}.\footnote{We remark that some works refer to the sandwiched R\'enyi divergences $\wt{D}_\alpha$ as minimal, since for $\alpha \geq 1/2$ they are the smallest \emph{additive} extensions of classical R\'enyi divergences. They are however larger than $D^{\MM}_\alpha$ in general.}  The \emph{measured R\'enyi divergences}~\cite{fuchs_1996,matsumoto_2014,berta_2017-1} are then defined as
\begin{align}\label{eq:measured_Renyi_def}
D_\alpha^{\MM}(\rho\|\sigma) \coloneqq \sup_{\M \in \MM} D_{\alpha}\! \left( \M(\rho) \middle\| \M(\sigma) \right),
\end{align}
where we understand each $\M$ to be a quantum-to-classical channel, so that $\M(\rho)$ is a probability distribution whenever $\Tr \rho = 1$. Formally, we can understand this as
\begin{equation}\begin{aligned}
  \MM \coloneqq \bigcup_{n \in \NN} \lset \M : X \mapsto \sum_{i=1}^n \Tr (M_i X)\, \proj{i} \bar M_i \geq 0 \; \forall i,\; \sum_{i=1}^n M_i = \id \rset,
\end{aligned}\end{equation}
with $(M_i)_i$ being the POVM elements of the given measurement. 
The various quantum R\'enyi divergences are in general strictly different: $D_\alpha^{\MM}(\rho\|\sigma) \leq \wt{D}_\alpha(\rho\|\sigma)$ for all $\alpha \in [\frac12,\infty]$, with equality if and only if $\alpha \in \big\{\frac12,\infty\big\}$ or $\rho$ and $\sigma$ commute~\cite{berta_2017-1}, and $\wt{D}_\alpha(\rho\|\sigma) \leq \ol{D}_{\alpha}(\rho\|\sigma)$ for $\alpha \in [0,2]$, with equality if and only if $\alpha=1$ or $\rho$ and $\sigma$ commute~\cite{wilde_2014,datta_2014,berta_2017-1}. We note in particular that $\wt{D}_1(\rho\|\sigma) = \ol{D}_1(\rho\|\sigma) = D(\rho\|\sigma)$ is the standard  quantum relative entropy, while $D_{1}^{\MM}(\rho\|\sigma)$ is strictly smaller for non-commuting states.

We will use the notation $\ol{Q}_\alpha$, $\wt{Q}_\alpha$, and $Q_{\alpha}^{\MM}$ to refer to the trace terms in the definitions of the different R\'enyi divergences, namely
\begin{equation}\begin{aligned}
  \mathbb{Q}_\alpha (\rho \| \sigma) \coloneqq \exp\big[ (\alpha-1)\, \mathbb{D}_\alpha (\rho\|\sigma) \big],
\end{aligned}\end{equation}
with $\DD_\alpha$ standing for one of $\ol{D}_\alpha$, $\wt{D}_\alpha$, or $D_{\alpha}^{\MM}$, and analogously for $\mathbb{Q}_\alpha$.

In the limit $\alpha \to \infty$, the sandwiched R\'enyi divergences give the so-called max-relative entropy~\cite{muller-lennert_2013}, which we will define for any self-adjoint operators $A$ and $B$ as~\cite{datta_2009}
\begin{equation}\begin{aligned}
  D_{\max} (A \| B) \coloneqq \log \inf \lset \lambda \in \RR_+ \bar A \leq \lambda B \rset.
\end{aligned}\end{equation}
It holds in fact that $D_{\max} (\rho \| \sigma) = D_{\max}^{\MM}(\rho\|\sigma)$~\cite{mosonyi_2015}. We will take $Q_{\max}(\rho\|\sigma) = \exp(D_{\max}(\rho\|\sigma))$.

An important notion in this work is that of \emph{smoothing}. This is most often based on either (generalised) \emph{trace distance} $\|\rho-\rho'\|_+$ or the \emph{purified distance} $P(\rho,\rho')$~\cite{tomamichel_2016}. Here, the norm $\|\cdot\|_+$ is defined through
\begin{equation}\begin{aligned}
  \norm{X}{+} &\coloneqq \max_{0 \leq M \leq \id} \big|\! \Tr M X \big| = \frac12 \big|\!\Tr X\big| + \frac12 \|X\|_1 ,
\end{aligned}\end{equation}
with $\norm{X}{1} = \Tr \sqrt{X^\dagger X}$ standing for the Schatten 1-norm. We make note of the fact that $\norm{X}{+} = \frac12 \norm{X}{1}$ whenever $\Tr X = 0$, and in particular $\norm{\rho-\rho'}{+} = \frac12 \norm{\rho-\rho'}{1}$ when $\rho$ and $\rho'$ are both states, reducing to the conventional notion of trace distance between states (a.k.a.\ statistical or total variation distance for classical distributions). More generally, $\norm{X}{+} = \Tr X_+$ whenever $\Tr X \geq 0$, with $X_+$ standing for the positive part of a Hermitian operator $X$. 
The purified distance is defined for any $\rho,\rho' \geq 0$ with $\Tr \rho' \leq \Tr \rho = 1$~as
\begin{equation}\begin{aligned}
  P(\rho, \rho') &\coloneqq \sqrt{1 - F(\rho,\rho')}, \qquad   F(\rho,\rho') \coloneqq \norm{\sqrt{\rho\vphantom{T}}\sqrt{\rho'}}{1}^{\,2} = \left( \Tr\sqrt{\sqrt{\rho}\rho'\sqrt{\rho}} \right)^2.
\end{aligned}\end{equation}

Given any quantum R\'enyi divergence $\DD_\alpha$ (including $D_{\max}$), its smoothed variants based on either trace (T) or purified (P) distance are conventionally defined~as~\cite{renner_2004,renner_2005}
\begin{equation}\begin{aligned}\label{eq:smooth_def}
  \DD_\alpha^{\ve,\,T} (\rho \| \sigma) &\coloneqq \opt \lset \vphantom{\big|} \DD_\alpha(\rho' \| \sigma) \bar \norm{\rho - \rho'}{+} \leq \ve, \; \rho' \geq 0,\; \Tr \rho' \leq 1 \rset\\
  \DD_\alpha^{\ve,\, P} (\rho \| \sigma) &\coloneqq \opt \lset \vphantom{\big|} \DD_\alpha(\rho' \| \sigma) \bar P(\rho, \rho') \leq \ve, \; \rho' \geq 0,\; \Tr \rho' \leq 1 \rset,
\end{aligned}\end{equation}
where
\begin{equation}\begin{aligned}
  \opt \coloneqq \begin{cases} \sup & \text{ for } \alpha \in [0,1) \\ \inf & \text{ for }  \alpha \in (1,\infty]. \end{cases}
\end{aligned}\end{equation}
Note that all definitions used here optimise over subnormalised states ($\Tr \rho' \leq 1$). 

Consider now a bipartite quantum state $\rho_{AE}$. If the system $A$ is classical, we will denote it as $X$ or $Z$. We say that the state is classical--quantum (CQ) if $\rho_{XE} = \sum_x p_X(x) \proj{x} \otimes \rho_{E,x}$ for some probability distribution $p_X$ and a collection of quantum states $\{\rho_{E,x}\}_x$ on the system $E$.

For any R\'enyi divergence $\DD_\alpha$, we define two variants of a conditional entropy:
\begin{equation}\begin{aligned}
  \HH_\alpha^{\uparrow}(X|E)_\rho &\coloneqq - \!\!\inf_{\substack{\sigma_E \geq 0\\\Tr \sigma_E = 1}} \DD_\alpha(\rho_{XE} \| \id_X \otimes \sigma_E),\\
  \HH_\alpha^{\downarrow}(X|E)_\rho &\coloneqq -\, \DD_\alpha(\rho_{XE} \| \id_X \otimes \rho_E).
\end{aligned}\end{equation}
Replacing $\DD_\alpha$ with any of the smooth R\'enyi divergences leads also to the quantities $\HH_\alpha^{\ve,\,\Delta,\,\smash\uparrow}(X|E)_\rho$ and $\HH_\alpha^{\ve,\,\Delta,\,\smash\downarrow}(X|E)_\rho$, with $\Delta$ standing for one of the smoothing variants, $P$ or $T$.

An important special case of a conditional entropy is the \emph{min-entropy}, which corresponds to the choice $\DD_\alpha = D_{\max}$. Its two variants are the average min-entropy $H_{\min}^{\smash\uparrow}(X|E)_\rho$ and the worst-case min-entropy $H_{\min}^{\smash\downarrow}(X|E)_\rho$. 
The \emph{smooth min-entropy} is then defined using one of the smooth variants of $D_{\max}$; in classical literature, smooth min-entropy is typically taken to mean $\Hmin{\ve}{T}(X|E)_\rho$, while in quantum contexts the quantity $\Hmin{\ve}{P}(X|E)_\rho$ is often taken as the definition. Another case that we will make significant use of is the \emph{collision entropy}, which corresponds to R\'enyi divergence with~$\alpha = 2$.

\subsection{Randomness extraction and privacy amplification}

Randomness extraction is concerned with converting imperfect randomness into outputs that are statistically indistinguishable from uniformly random~\cite{nisan_1996}. This task is naturally described in terms of benchmarks that measure the probability of failure: first, the strength of the randomness of a source distribution is quantified in terms of its unpredictability, that is, the probability that its value can be guessed, which is connected to the min-entropy as $p_{\rm guess}(X) = \exp(-H_{\min}^{\smash{\uparrow}}(X))$; second, the performance of an extractor is similarly naturally measured by the probability that its output can be distinguished from uniform, which corresponds to the trace distance.

Formally, a function $h\colon \X \times \S \to \Z$ is called a strong, seeded ${(k,\epsilon)}$ randomness extractor with seed distribution $\mu_S$ if, for any distribution $p_X$ with $H^{\smash{\uparrow}}_{\min}(X)_p \geq k$, the induced distribution $p^h_{ZS}(z,s) \coloneqq \sum_{x : h(x,s) = z} p_X(x) \mu_S(s)$ satisfies $\big\|p_{ZS}^h - \frac{\id_Z}{|\Z|} \otimes \mu_S\big\|_+ \leq \ve$. The role of the seed $\mu_S$ here (which is often taken to be uniformly random) is to enable one to sample from a family of random hash functions $f_s \colon \X \to \Z$, corresponding to $f_s(x) = h(x,s)$. However, this seed cannot be assumed to be secret and may be revealed publicly, which necessitates the requirement that the output distribution be (approximately) independent from the seed. Understanding the seed as a family of hash functions $\mathcal{F} \coloneqq \{f_s\}_{s\in \S}$ together with a distribution $\mu_S$, we can equivalently write the definition as the requirement that
\begin{equation}\begin{aligned}
   \ve &\geq \norm{\,p_{ZS}^h - \frac{\id_Z}{|\Z|} \otimes \mu_S}{+} = \frac12 \sum_{s} \sum_z \sum_{x : f_s(x) = z} \left|  p_X(x) \,\mu_S(s) - \frac{1}{|\Z|} \mu_S(s) \right|
   \\&= \!\EE_{S \sim \mu_S} \norm{\, p^{f_S}_{Z} - \frac{\id_Z}{|\Z|}}{+} 
   =\! \EE_{f \sim \mu_\mathcal{F}} \,\norm{ \,p^{f}_{Z} - \frac{\id_Z}{|\Z|}}{+}.
\end{aligned}\end{equation}
Here we defined $p^f_{Z}(z) \coloneqq \sum_{x : f(x) = z} p_X(x)$ and re-interpreted $f$ as a random variable taking values in $\F$ and distributed according to the distribution $\mu_{\F}$ induced by $\mu_S$. We will hereafter make $\mu_\F$ implicit and simply write $\EE_f$ for the expectation over some family of random functions.

The above definition is further generalised by accounting for side information, that is, an adversary whose system $E$ may be correlated with the initial source distribution. We refer to randomness extraction against side information as \emph{privacy amplification}~\cite{bennett_1988,bennett_1995,renner_2005}, the most general formulation of which does not restrict the adversary to be classical, but allows them to be quantum. Formally, modelling the initial distribution as a classical--quantum state $\rho_{XE} = \sum_x p_X(x) \proj{x} \otimes \rho_{E,x}$, the aim of a $(k,\ve)$~randomness extractor is then to ensure that the output randomness is approximately independent of the side information, in the sense that for any $\rho_{XE}$ with $H^{\smash{\uparrow}}_{\min}(X|E)_\rho \geq k$ we require
\begin{equation}\begin{aligned}
  \EE_{f} \norm{ \rho^f_{ZE} - \frac{\id_Z}{|\Z|} \otimes \rho_E}{+} \leq \ve,
\end{aligned}\end{equation}
where
\begin{equation}\begin{aligned}
  \rho^f_{ZE} \coloneqq \sum_z \proj{z} \otimes \!\!\sum_{x : f(x) = z} p_X(x) \,\rho_{E,x}.
\end{aligned}\end{equation}
We note that in the literature one sometimes encounters formulations that consider the distance $\min_{\sigma_E} \big\| \rho^f_{ZE} - \frac{\id_Z}{|\Z|} \otimes \sigma_E\big\|_+$ rather than fixing $\sigma_E = \rho_E$; however, such security criteria do not satisfy an essential property known as composability~\cite{canetti_2001,portmann_2022}, preventing their use in the study of subroutines of quantum key distribution schemes --- a key application of privacy amplification. Distance criteria based on measures such as accessible information are also not suitable for the analysis of composable security~\cite{konig_2007}.

Both conceptually and operationally, an important aspect of randomness extraction is for it to work universally, i.e.\ to not depend on the source distribution but only on its randomness as measured by the min-entropy. However, to exactly characterise the properties of privacy amplification protocols that we study here, it will be useful to follow standard notation in quantum information theory and define the state-dependent notion of \emph{extractable randomness with trace distance error ${\ve}$} as
\begin{equation}\begin{aligned}\label{eq:ell_def}
    \ell_\ve(\rho_{XE}) &\coloneqq \max \lset \log \left|\Z\right| \bar \EE_f\, \norm{ \rho^f_{ZE} - \frac{\id_Z}{\left|\Z\right|} \otimes \rho_E }{+} \leq \ve \rset\\
    &\hphantom{:}= \max \lset \log \left|\Z\right| \bar \min_{f \colon \X \to \Z}\, \norm{ \rho^f_{ZE} - \frac{\id_Z}{\left|\Z\right|} \otimes \rho_E }{+} \leq \ve \rset,
\end{aligned}\end{equation}
where the maximisation in the first line is over output alphabet sizes $|\Z|$ as well as all families $\F$ of functions $f\colon \X \to \Z$ and all distributions $\mu_\F$ thereon. The second line follows since on the one hand the minimum cannot be larger than the expected value, and on the other hand one can always take a singleton family of functions. 

The merit of smooth entropies stems precisely from their connection to extractable randomness. In the case of \emph{classical} side information $Y$, it is known that~\cite{renner_2005-2,renner_2005}
\begin{equation}\begin{aligned}\label{eq:classical_side_info}
 \Hmin{\ve-\mu}{T}(X|Y)_p - \log \frac{1}{4\mu^2} \leq \ell_\ve(p_{XY}) \leq \Hmin{\ve}{T}[down] (X|Y)_p.
\end{aligned}\end{equation}
This shows precisely that knowing only the smooth min-entropy of the source is enough to tightly characterise its extractable randomness. 
In fact, although smooth min-entropy is perhaps the most appealing for such formulations due to its direct connection to guessing probability, other smooth R\'enyi entropies of order $\alpha > 1$ are equivalent to it up to suitable error terms~\cite{renner_2005-2}.
The achievability direction here, based on the leftover hash lemma~\cite{bennett_1988,impagliazzo_1989,bennett_1995} that we will return to in Section~\ref{sec:achievability}, is even more tightly described in terms of the collision entropy $\Hs{2}{\ve-\mu}{T}(X|Y)_p$. 

Extensions of one-shot bounds as in~\eqref{eq:classical_side_info} to quantum side information have been studied since the early works of~\cite{renner_2005,renner_2005-1,tomamichel_2011}. However, as we discussed in the Introduction, their known formulations are not tight. It is the aim of this work to rectify this.

\section{Measured smooth collision divergence}\label{sec:smollision}

\subsection{Motivation: the subtle question of smoothing}\label{subsec:subtle}

In the classical case, the total variation (trace) distance $\norm{p-q}{+}$ naturally features in the definitions of security criteria for cryptographic tasks, as it exactly quantifies the best probability of successfully distinguishing any two distributions.
This also led to the use of smooth divergences $\Ds{\alpha}{\ve}{T}$, in particular the smooth min-entropy $\Hmin{\ve}{T}$, in such applications.

In the quantum case, depending on the precise setting, the `correct' definition of a statistical distance is somewhat more debatable.
The operational motivation for the trace distance as the optimal average probability of distinguishing any two distributions extends to quantum states~\cite{helstrom_1969,holevo_1973}, making it a perfectly well-motivated choice, and indeed one that underlay the first definitions of quantum smooth divergences. 
In some operational tasks in quantum information, however, an important role is played by purifications, making the purified distance --- which can be understood as the least trace distance between the purifications of the states~\cite{rastegin_2006,tomamichel_2010} --- a justified and useful alternative. Thanks also to the many desirable properties satisfied by this distance~\cite{tomamichel_2010,tomamichel_2016}, it has found widespread use in quantum information theory. 
We will in fact see in Section~\ref{sec:converse} that also in privacy amplification, there are reasons to give consideration to this distance, at least as a technical tool.

Nevertheless, especially in quantum cryptography, a key role is played by distinguishability through measurements, and the security of cryptographic protocols is intrinsically tied to minimising the probability that they can be distinguished from perfectly secure ones~\cite{portmann_2022,ferradini_2025}. Hence, bounding the trace distance is of primary importance and is conventionally used to define security criteria. 
This has led to difficulties in obtaining tight bounds for tasks such as quantum privacy amplification, as the trace distance is unfortunately not as well behaved for general quantum states as it is for classical distributions~\cite{hayashi_2014,renes_2018,shen_2024}, and thus a tight asymptotic analysis of smooth entropies in quantum information theory has often been limited to the purified distance.

Our investigation begins with a reconsideration of a seemingly basic question: how to properly generalise the notion of trace distance smoothing to quantum states?

One important point to note is that some of the earliest definitions of classical smooth divergences were not actually defined using trace distance per se. The underlying idea can instead be understood as discarding a small mass from a classical probability distribution $p$, that is, decomposing it into two `sub-distributions' as $p = p_1 + p_2$, and then using one of these approximate distributions as a surrogate for $p$. In the particular case of privacy amplification, this idea draws on the approach of~\cite{impagliazzo_1989} and was already present in the early work of Renner and Wolf~\cite{renner_2005-2}, with a lucid exposition given in~\cite[Appendix~I]{yang_2019}: by taking $p_1$ as a sub-distribution with a lower variance than $p$, one can bound the amount of extractable randomness much more tightly than if one were to work with $p$ directly. The reason why this definition is typically not formally distinguished from trace distance smoothing in classical literature is that, when computing the smooth entropies $H^{\ve,\,T}_{\min}$ and $H^{\ve,\,T}_2$, the two notions are actually completely equivalent:  we have for instance that~\cite{yang_2019,renes_2018,abdelhadi_2020}
\begin{equation}\begin{aligned}\label{eq:two_smoothings}
  \Dmax{\ve}{T} (p \| q) &= \inf \lset \,D_{\max}(p'\|q) \bar  p' \geq 0,\; \Tr p' \leq 1,\; \norm{p-p'}{+} \leq \ve \rset\\
  &=  \inf \lset \, D_{\max}(p'\|q) \bar  0 \leq p' \leq p,\; \norm{p - p'}{+} \leq \ve \rset.
\end{aligned}\end{equation}
This follows because the optimal distribution $p'$ can always be taken of the form $p' = p - (p - \gamma q)_+$, i.e.\ the pointwise minimum of $p$ and $\gamma q$, for a suitable choice of a positive number $\gamma$.

A key problem now is that attempting to do something similar for quantum states leads to a \emph{different} notion of smoothing than trace distance. It is not difficult to find counterexamples to a relation like~\eqref{eq:two_smoothings} in the quantum case. The technical reason for this is that the Loewner order, i.e.\ the partial order induced by the positive semidefinite cone, does not form a lattice; due to this, a non-commutative minimum of two positive operators generally cannot be a positive operator, and in particular $\rho - (\rho - \gamma \sigma)_+ \not\geq 0$ (we refer to~\cite{moreland_1999,cheng_2023-1} for discussions). This leads us to an intuitive, if somewhat controversial, idea: instead of 
sub-dis\-tri\-bu\-tions, let us consider smoothing over sub-operators $R \leq \rho$, even if they are not necessarily positive. 
Such a smoothing notion was implicitly considered for the max-relative entropy in~\cite[Appendix~A]{regula_2025}, from which one can deduce that indeed, with an allowance for non-positive operators, the trace-distance smoothing and sub-operator smoothing become equivalent. The quantity defined in this way corresponds exactly to an inverse function of the quantum hockey-stick divergence $E_\gamma(\rho\|\sigma) \coloneqq \Tr ( \rho - \gamma \sigma )_+$~\cite{sharma_2013,hirche_2023}; precisely,
\begin{equation}\begin{aligned}
  \wt{D}^\ve_{\max}(\rho\|\sigma) &\coloneqq \inf \lset \log \gamma \bar \Tr ( \rho - \gamma \sigma )_+ \leq \ve \rset\\
  &\hphantom{:}= \inf \lset \,D_{\max}(R\|\sigma) \bar R = R^\dagger,\; \Tr R \leq 1,\; \norm{\rho - R}{+} \leq \ve \rset\\
  &\hphantom{:}= \inf \lset \,D_{\max}(R\|\sigma) \bar R = R^\dagger,\; R \leq \rho,\; \norm{\rho - R}{+} \leq \ve \rset,
\end{aligned}\end{equation}
with the last two lines shown in~\cite[Lemma~A.1]{regula_2025}. 
The quantity $\wt{D}^\ve_{\max}(\rho\|\sigma)$ here is a divergence that made appearances in quantum information theory under various guises~\cite{nagaoka_2007,datta_2009-1,datta_2015,hirche_2023,nuradha_2024,regula_2025}; it was formalised by Datta and Leditzky~\cite{datta_2015} as a variant of an information spectrum divergence, and many of its properties in connection with the smooth max-relative entropy were later studied  in~\cite{nuradha_2024,regula_2025}. 
This could be generalised to a definition of `Hermitian-smoothed' R\'enyi divergences as
\begin{equation}\begin{aligned}
  \DD_{\alpha}^{\ve,\,\mathrm{Herm}} (\rho \| \sigma) \coloneqq \inf \lset \DD_\alpha(R \| \sigma) \bar R = R^\dagger,\; R \leq \rho,\; \norm{\rho - R}{+} \leq \ve \rset,
\end{aligned}\end{equation}
assuming that the definition of the given quantum R\'enyi divergence $\DD_\alpha$ can be extended to non-positive first arguments, and changing the $\inf$ to a $\sup$ if $\alpha < 1$. 
Classically, due to the relation~\eqref{eq:two_smoothings}, one has in particular $D^{\ve, \, \mathrm{Herm}}_{\max}(p\|q) = \wt{D}^\ve_{\max}(p\|q) = D^{\ve,\,T}_{\max}(p\|q)$. However, the need to consider non-positive operators $R$ in the quantum definition could make such quantities appear somewhat unphysical, putting into question the appropriateness of the resulting notion of smoothing.

Let us then step back and consider yet another, natural way to extend smoothing notions from classical distributions to quantum states: inspired by the definition of measured R\'enyi divergences, we define the \deff{measured smooth R\'enyi divergences} as
\begin{equation}\begin{aligned}\label{eq:measured_smooth_def}
  \Ds{\alpha}{\ve}{M} (\rho \| \sigma) \coloneqq \sup_{\M\in\MM} \Ds{\alpha}{\ve}{T}(\M(\rho)\|\M(\sigma)).
\end{aligned}\end{equation}
To the best of our knowledge, this smoothing notion first appeared in~\cite{regula_2025} where it was applied to the max-relative entropy. That work revealed an especially curious connection: for all quantum states it holds that~\cite[Proposition~3]{regula_2025}
\begin{equation}\begin{aligned}
   \Dmax{\ve}{M} (\rho \| \sigma) = \wt{D}^\ve_{\max}(\rho\|\sigma) = D^{\ve, \, \mathrm{Herm}}_{\max}(\rho\|\sigma),
\end{aligned}\end{equation}
and hence the notions of smoothing through measurements and smoothing through Hermitian sub-operators are actually equivalent for $D_{\max}$. Not only does this mitigate the ostensibly unphysical character of `Hermitian smoothing', it actually points towards this smoothing being a very appropriate choice when focusing on distinguishability through measurements. 
We see in particular that two different ways to extend the classical smoothing with total variation distance lead precisely to this form of quantum smoothing, distinct from conventional trace distance smoothing over quantum states.

Motivated by this insight, we will show that it is precisely this modified smoothing notion --- whether understood as lifting classical smoothing through measurements, or as a generalisation of smoothing through sub-distributions that allows optimisation over Hermitian operators --- that is the most appropriate choice for the analysis of quantum privacy amplification. To establish this, however, we need to extend the above properties beyond the special case of max-relative entropy.

\subsection{Measured collision divergence}

We begin by recalling some known facts about the measured R\'enyi divergence of order 2. 

For any $\sigma > 0$, it can be expressed as
\begin{align}\label{eq:D2_unsmoothed_def}
  D_{2}^{\MM} (\rho\|\sigma) &\coloneqq \sup_{\M \in \MM} D_2(\M(\rho)\|\M(\sigma))\\[-5pt]
  &\hphantom{:}= \log \sup_{0 \leq W \leq \id}  \frac{\left(\Tr W \rho\right)^2}{\Tr W^2 \sigma} \label{eq:D2_unsmoothed_def_var}\\[3pt]
  &\hphantom{:}= \log \,\Tr \rho \, \J_\sigma^{-1}(\rho), \label{eq:D2_unsmoothed_def_bures}
\end{align}
where the second line is a well-known variational form of~\cite{berta_2017-1}, and where $\J_\sigma^{-1}$ in the last line is the inverse of the matrix multiplication superoperator defined as
\begin{equation}\begin{aligned}
  \J_\sigma(X) = \frac12 (\sigma X + X \sigma)
\end{aligned}\end{equation}
for all Hermitian $X$. This last expression in~\eqref{eq:D2_unsmoothed_def_bures} in terms of $\J_\sigma^{-1}$ is a form often used in the study of the quantum (Bures) $\chi^2$ divergence~\cite{braunstein_1994,petz_2011,temme_2015}, which is indeed closely related to the collision divergence: we have $D_2^{\MM}(\rho\|\sigma) = \log ( \chi^2(\rho\|\sigma) + 1)$ whenever $\Tr \rho = \Tr \sigma = 1$. 

To obtain an explicit expression for the superoperator $\J_\sigma^{-1}$ itself, we can diagonalise $\sigma$ as $\sigma = \sum_i \sigma_{i} \proj{i}$ and write
\begin{equation}\begin{aligned}
  \J_\sigma^{-1}(X) &= \sum_{i,j} \frac{2}{\sigma_{i} + \sigma_{j}} X_{ij} \ketbra{i}{j}, \quad X_{ij} = \braket{i|X|j}\\
  &= 2 \int_0^\infty \! e^{- t \sigma} X e^{- t \sigma} \mathrm{d}t,
\end{aligned}\end{equation}
where the second form is common in the study of the Lyapunov matrix equation, the solutions to which are expressed precisely through $\J_\sigma^{-1}$.

Observe that $\< X, Y \>_{\sigma} \coloneqq \Tr X \J_\sigma^{-1}(Y)$ defines a Hermitian weighted inner product. This is often called the Bures inner product~\cite{lesniewski_1999,hiai_2012}. It is the very same inner product as the one used to define the commonly encountered variant of quantum Fisher information~\cite{braunstein_1994,sidhu_2020}, and the metric induced by this inner product (the minimal monotone Riemannian metric~\cite{petz_1996,lesniewski_1999,hiai_2012}) is sometimes also called the symmetric logarithmic derivative (SLD) metric. One can then define the corresponding Bures norm $\norm{X}{\Bures\sigma} \coloneqq \sqrt{\Tr X \J_\sigma^{-1}(X)}$, yielding $D_2^{\MM}(\rho\|\sigma) = \log  \norm{\rho}{\Bures\sigma}^2$.

The above definitions can be extended to general $\sigma \geq 0$ by excluding operators that do not lie in the range of $\J_{\sigma}$. More precisely, the equation $\J_{\sigma}(Z) = Y$ has a solution $Z$ if and only if $\Pi_\sigma^\perp Y \Pi_\sigma^\perp = 0$, where $\Pi_\sigma^\perp$ stands for the projection onto $\ker(\sigma)$. 
In full generality, we then define the \deff{Bures seminorm}
\begin{equation}\begin{aligned}
  \norm{X}{\Bures\sigma} &\coloneqq \begin{cases} \displaystyle \sqrt{\Tr X \J_\sigma^{-1}(X)} & \text{ if } X \in \ran(\J_\sigma)\\ \infty & \text{ otherwise}\end{cases}\\
  &\hphantom{:}= \lim_{\ve\to0^+} \norm{X}{\Bures\sigma + \ve \id}
\end{aligned}\end{equation}
with $\J_{\sigma}^{-1}$ denoting the inverse on the range of $\J_\sigma$, the latter being the subspace
\begin{equation}\begin{aligned}
\ran(\J_\sigma) = \lset X \bar X = X^\dagger,\; \Pi_\sigma^\perp X \Pi_\sigma^\perp = 0 \rset.
\end{aligned}\end{equation}
The quantity $\|\cdot\|_{\Bures\sigma}$ takes values in the extended halfline $\RR_+ \!\cup\! \{+\infty\}$ and is a norm whenever $\sigma > 0$. 
This recovers the standard definition of the measured R\'enyi divergence for general quantum states $\rho$ and $\sigma$: we have $D_2^{\MM}(\rho\|\sigma) = \log \|\rho\|_{\Bures\sigma}^2$, which evaluates to infinity if $\supp(\rho) \not\subseteq \supp(\sigma)$. 

Let us then define our main quantity of interest. First, recall that for classical probability distributions, we consider
\begin{equation}\begin{aligned}
  \Ds{2}{\ve}{T} (p\|q) = \inf \lset \log \sum_x p'(x)^2\, q(x)^{-1} \bar p' \geq 0,\; \Tr p' \leq 1,\; \norm{p - p'}{+} \leq \ve \rset.
\end{aligned}\end{equation}
Our protagonist is then the measured smooth extension of this quantity, whose definition we restate for clarity.
\begin{definition}
The \deff{measured smooth collision divergence} $\boldsymbol{D_2^{\ve,\,\MM}}$ 
is given by the supremum of the classical smooth collision divergence $D^{\ve,T\!}_2$ optimised over all quantum measurements. Precisely,
\begin{equation}\begin{aligned}
  \Ds2\ve M (\rho\|\sigma) &\coloneqq \sup_{\M \in \MM} \Ds{2}{\ve}{T} (\M(\rho)\|\M(\sigma)).
\end{aligned}\end{equation}
\end{definition}
We recall also the notation $\Qs2\ve M(\rho\|\sigma) = \exp(\Ds2\ve M(\rho\|\sigma))$.

\subsection{Variational form through Hermitian smoothing}

Our results will rely on a variational characterisation of the measured smooth collision divergence, showing that it can be understood exactly as the measured R\'enyi divergence smoothed over Hermitian operators. This establishes an equivalence between the smoothing notions discussed earlier in Section~\ref{subsec:subtle}, in the sense that $\Ds{2}{\ve}{M}(\rho\|\sigma) = D_2^{\MM,\,\ve,\,\mathrm{Herm}} (\rho\|\sigma)$.

\begin{boxed}[filled]
\begin{theorem}\label{thm:D2_smoothed_forms}
For all quantum states $\rho$, all operators $\sigma \geq 0$, and all $\ve \in [0,1)$, it holds that
\begin{align}
  \Ds{2}{\ve}{M}(\rho\|\sigma) &= \log \inf \lset \norm{R}{\Bures\sigma}^2 \bar R = R^\dagger,\; \norm{\rho - R}{+} \leq \ve,\; R \leq \rho \rset \label{eq:D2_variational_forms_primal}\\
  &= \log \sup_{0 \leq W \leq \id}  \frac{\left(\Tr (W \rho) - \epsilon\right)_+^2}{\Tr W^2 \sigma}, \label{eq:D2_variational_forms_dual}
\end{align}
where $(x)_+ = \max \{ 0, x \}$, and for consistency we understand $a / 0 = \infty \; \forall a > 0$ and $0 / 0 = 0$.\\[-8pt]

Furthermore, $\Ds{2}{\ve}{M}(\rho\|\sigma) < \infty$ if and only if  $\Tr \Pi_\sigma^\perp \rho \leq \ve$, where $\Pi_\sigma^\perp$ denotes the projection onto the kernel of~$\sigma$. Whenever this is the case, there exists an optimal solution $R$ achieving the infimum in~\eqref{eq:D2_variational_forms_primal}.
\end{theorem}
\end{boxed}

\begin{remark}
One can compare the variational forms of $\Ds{2}{\ve}{M}$ in Theorem~\ref{thm:D2_smoothed_forms} with those of $\Dmax{\ve}{M}$, which can be expressed as~\cite{nuradha_2024,regula_2025}
\begin{equation}\begin{aligned}
  \Dmax{\ve}{M} (\rho \| \sigma) &= \log \inf \lset \lambda \bar R \leq \lambda \sigma,\; R = R^\dagger,\; \norm{\rho - R}{+} \leq \ve,\; R \leq \rho \rset\\
  &= \log \sup_{0 \leq W \leq \id}  \frac{\left(\Tr (W \rho) - \epsilon\right)_+}{\Tr W \sigma}.
\end{aligned}\end{equation}
\end{remark}

We refer to the two optimisation problems appearing in the statement of Theorem~\ref{thm:D2_smoothed_forms} as the primal and the dual, respectively, and denote their optimal values as
\begin{equation}\begin{aligned}
  \primal(\rho,\sigma,\ve) &\coloneqq \inf \lset \norm{R}{\Bures\sigma}^2 \bar R = R^\dagger,\; \norm{\rho - R}{+} \leq \ve,\; R \leq \rho \rset,\\
    \dual(\rho,\sigma,\ve) &\coloneqq \sup_{0 \leq W \leq \id}  \frac{\left(\Tr (W \rho) - \epsilon\right)_+^2}{\Tr W^2 \sigma}.
\end{aligned}\end{equation}
The proof of this result will proceed step by step through a series of lemmas: first to show the equivalence of the two problems, then their equality with $\Ds{2}{\ve}{M}$ in the classical (commuting) case, and finally the extension to all quantum states. Before that, we separately handle degenerate cases.

\begin{lemma}[Diverging case]\label{lem:D2_diverging_case}
If $\Tr \Pi_\sigma^\perp \rho > \ve$, then $\Ds{2}{\ve}{M}(\rho\|\sigma) = \primal(\rho,\sigma,\ve) = \dual(\rho,\sigma,\ve) = \infty$. Otherwise, the quantities are all finite.
\end{lemma}
\begin{proof}
Assume first that $\Tr \Pi_\sigma^\perp \rho > \ve$. Consider that the operator $W = \Pi_\sigma^\perp$ is feasible for $\dual(\rho,\sigma,\ve)$ with a diverging feasible optimal value, so $\dual(\rho,\sigma,\ve) = \infty$. 
For the primal problem, assume that there exists a feasible $R$ such that $\norm{R}{\Bures\sigma} < \infty$, i.e.\ $\Pi_\sigma^\perp R \Pi_\sigma^\perp = 0$. As $\norm{\rho-R}{+} \leq \ve$, we must have $\Tr R \geq 1-\ve$, and hence $\Tr \Pi_\sigma R \geq 1 - \ve$ where $\Pi_\sigma = \id - \Pi_\sigma^\perp$. But since $R \leq \rho$, this would mean that $\Tr \Pi_\sigma \rho \geq 1-\ve$, which contradicts the assumption that $\Tr \Pi_\sigma^\perp\rho > \ve$. Therefore, no such feasible $R$ can exist, and hence $\primal(\rho,\sigma,\ve) = \infty$. 
Now, for the measurement channel $\M(\cdot) = \big(\!\Tr( \Pi_\sigma^\perp\,\cdot),\, \Tr(\Pi_\sigma \,\cdot) \big)$, we have $\Ds{2}{\ve}{M}(\rho\|\sigma) \geq \Ds{2}{\ve}{T} (\M(\rho)\|\M(\sigma))$; an analogous argument as we just made for $R$ tells us that there cannot exist a feasible $p'$ with $\norm{\M(\rho)-p'}{+} \leq \ve$ and a finite objective value, yielding $\Ds{2}{\ve}{M}(\rho\|\sigma) = \infty$.

Assume now that $\Tr \Pi_\sigma^\perp \rho \leq \ve$. Then $R = \rho - \Pi_\sigma^\perp \rho \Pi_\sigma^\perp$ is feasible for the primal optimisation problem, and since $\Pi_\sigma^\perp R \Pi_\sigma^\perp = 0$, we have $\primal(\rho,\sigma,\ve) < \infty$. 
The problem $\dual(\rho,\sigma,\ve)$ can only diverge when there exists a feasible $W$ satisfying $\Tr W \sigma = 0$ but $\Tr W \rho > \ve$ (or a feasible sequence converging to such a $W$). However, the constraint $\Tr W \sigma = 0$ means that any such $W \in [0,\id]$ would need to satisfy $W \leq \Pi_\sigma^\perp$, and hence $\Tr W \rho > \ve$ is impossible.
To argue the finiteness of $\Ds{2}{\ve}{M}(\rho\|\sigma)$ itself, consider any measurement channel $\M$ with POVM elements $(M_x)_x$, let $p(x) = \Tr M_x \rho $, $q(x) = \Tr M_x \sigma$, and define $p'$ as the  restriction of $p$ to the support of $q$, that is, to the symbols in the set $S \coloneqq \{ x : \Tr M_x \sigma > 0 \}$. Then $\|p - p'\|_+ = \sum_{x \notin S} \Tr M_x \rho \leq \Tr \Pi_{\sigma}^\perp \rho \leq \ve$, and by construction $p'$ is supported on the support of $q$, so $D_2(p'\|q)$ is finite.
\end{proof}

\begin{lemma}[Strong duality]\label{lem:D2_primaldual}
Assume that $\primal(\rho,\sigma,\ve) < \infty$. Then $\displaystyle \primal(\rho,\sigma,\ve) = \dual(\rho,\sigma,\ve)$, and the infimum in the definition of $\primal(\rho,\sigma,\ve)$ is achieved.
\end{lemma}
\begin{proof}
Notice first that the constraint $R \leq \rho$ means that $\norm{\rho - R}{+} = \Tr(\rho-R)$. Thus
\begin{equation}\begin{aligned}
\primal(\rho,\sigma,\ve) =  \inf \lset \norm{R}{\Bures\sigma}^2 \bar R \leq \rho,\; \Tr R \geq \Tr \rho-\ve \rset.
\end{aligned}\end{equation}
The above is readily observed to be a convex optimisation problem. Due to the assumption of finite optimal value, we can without loss of generality restrict ourselves to $R$ in the range of $\J_\sigma$. Working in the real vector space of Hermitian matrices with the Hilbert--Schmidt inner product, the corresponding Lagrangian is then
\begin{equation}\begin{aligned}
  \LL(R,A,y) &= \Tr R \J_\sigma^{-1}(R) - \Tr A(\rho - R) - y ( \Tr R - (1-\ve))\\
  &= \Tr R \!\left(\J_\sigma^{-1}(R) + A - y \id\right) - \Tr A\rho + y (1-\ve)
\end{aligned}\end{equation}
with Lagrange multipliers $A \geq 0$, $y \geq 0$. 
To evaluate the Lagrange dual function, we now need to minimise the Lagrangian over all Hermitian $R$. Writing $R = \J_\sigma(Z)$ for some Hermitian $Z$, we have
\begin{equation}\begin{aligned}
  \LL(R,A,y) =  \Tr \J_\sigma(Z) Z + \Tr \J_\sigma(Z)\left( A - y \id\right)  - \Tr A\rho + y (1-\ve).
\end{aligned}\end{equation}
The $Z$-dependent part of this expression can be recognised as a quadratic form with respect to the inner product $\<X,Y\>_{\J_\sigma} \coloneqq \Tr \J_\sigma(X) Y$, noting this to be different from (dual to) the Bures inner product $\<X,Y\>_\sigma$ that defines $\|\cdot\|_{\Bures\sigma}$. Using now the identity $\|a + b\|^2 = \|a\|^2 + \|b\|^2 + 2 \< a ,b \>$ we can write the above as
\begin{equation}\begin{aligned}
  \LL(R,A,y) =  \norm{ Z + \frac12 \left( A - y \id\right) }{\J_\sigma}^2 - \frac14 \norm{ A - y \id }{\J_\sigma}^2  - \Tr A\rho + y (1-\ve),
\end{aligned}\end{equation}
whose minimum in $Z$ is clearly achieved at $Z = -\frac12 (A - y \id)$. Reparametrising as $B \coloneqq -\frac12 (A - y \id)$, we thus have
\begin{equation}\begin{aligned}
  \inf_{R = R^\dagger} \LL(R,A,y) &= \Tr \J_\sigma(B) B  - 2 \Tr \J_\sigma(B) B  + 2 \Tr B \rho - y \Tr \rho + y (1-\ve)\\
  &= - \Tr B^2 \sigma + 2 \Tr B \rho - y \ve
\end{aligned}\end{equation}
where we used the definition of $\J_\sigma$ and the cyclicity of the trace. Making another reparametrisation as $t \coloneqq \frac12 y$, by definition of Lagrange duality we can write the dual of the optimisation problem as
\begin{equation}\begin{aligned}\label{eq:D2smooth_dual}
 \primald(\rho,\sigma,\ve) \coloneqq \sup_{\substack{A \geq 0,\\y \geq 0}} \inf_{R = R^\dagger} \LL(R,A,y) &= \sup \lset 2 \Tr B \rho - 2 t \ve - \Tr B^2 \sigma \bar  B \leq t \id,\, t \geq 0\rset.
\end{aligned}\end{equation}
Since $t = 1$ and $B = \frac{1}{2} \id$ are strictly feasible for this dual problem, it follows by Slater's condition (see e.g.~\cite[Sec.~5.9]{boyd_2004}) that strong duality holds. This immediately gives that $\primald(\rho,\sigma,\ve) = \primal(\rho,\sigma,\ve)$ and that an optimal solution $R$ for $\primal(\rho,\sigma,\ve)$ exists.

Now, observe that for any feasible $B$ for~\eqref{eq:D2smooth_dual}, $B_+$ is also feasible, since $B \leq t \id \iff B_+ \leq t \id$. But choosing $B_+$ can only improve the objective function: due to the positivity of $\rho$ and $\sigma$, we have $\Tr B_+ \rho \geq \Tr B \rho$ and $\Tr B_+^2 \sigma \leq \Tr (B_+^2 + B_-^2) \sigma = \Tr B^2 \sigma$. We can thus restrict the optimisation to $B_+$ without loss of generality. Denoting $W \coloneqq \frac{1}{t} B_+$, we thus obtain
\begin{equation}\begin{aligned}\label{eq:D2smooth_dual2}
  \primal(\rho,\sigma,\ve) = \sup \lset 2 t \Tr W \rho - 2 t \ve - t^2 \Tr W^2 \sigma \bar  0 \leq W \leq \id,\, t \geq 0\rset.
\end{aligned}\end{equation}
Consider now that, for any $a \in \RR$ and $b \in \RR_+$, it holds that
\begin{equation}\begin{aligned}
  \sup_{t\geq 0} \, (2 t a - t^2 b) = \begin{cases} \frac{a^2}{b} & \text{ if }\, a \geq 0, b > 0\\ 0 & \text{ if }\, a < 0 \,\text{ or }\, a = b = 0\\ \infty & \text{ if }\, a > 0, b = 0.\end{cases}
\end{aligned}\end{equation}
We can thus equivalently write the above as
\begin{equation}\begin{aligned}\label{eq:D2smooth_dual3}
  \primal(\rho,\sigma,\ve) &= \sup_{0 \leq W \leq \id} \frac{\left(\Tr (W \rho) - \ve\right)_+^2}{\Tr W^2 \sigma}\\
  &= \dual(\rho,\sigma,\ve)
\end{aligned}\end{equation}
as claimed.
\end{proof}

The next step is to establish the validity of Theorem~\ref{thm:D2_smoothed_forms} in the case of classical probability distributions.

\begin{lemma}[Classical case]\label{lem:D2_classicalcase}
For all probability distributions $p$ and all nonnegative distributions $q$, we have
\begin{equation}\begin{aligned}
  \Ds{2}{\ve}{T}(p\|q) = \log \primal(p,q,\ve).
\end{aligned}\end{equation}
Furthermore, it holds that $\Ds{2}{\ve}{T}(p\|q) = D_2(p' \| q)$, 
where $p'$ is defined as
\begin{equation}\begin{aligned}\label{eq:pdash}
  p'(x) \coloneqq \begin{cases} p(x) & \text{if } p(x) \leq \gamma q(x)\\ \gamma q(x) & \text{if } p(x) > \gamma q(x), \end{cases}
\end{aligned}\end{equation}
with $\gamma$ chosen so that $E_\gamma(p\|q) = \Tr(p - \gamma q)_+ = \ve$, that is, $\log \gamma = \Dmax{\ve}{M}(p\|q) = D_{\max}(p'\|q)$.
\end{lemma}
The latter part of the lemma was previously shown in~\cite[Lemma~2]{yang_2019}, whose construction inspired our proof approach. The same sub-distribution $p'$ also played an important role in~\cite{hayashi_2013,hayashi_2016-2}.

We will use the notation $\Qs{2}{\ve}{T}(p\|q) = \exp(\Ds{2}{\ve}{T}(p\|q))$. 
\begin{proof}
By definition,
\begin{equation}\begin{aligned}\label{eq:classical_opt}
  Q^{\ve,T}_{2}(p\|q) = \inf \lset \sum_x \frac{p'(x)^2}{q(x)} \bar p' \geq 0,\; \Tr p' \leq 1,\; \norm{p - p'}{+} \leq \ve \rset.
\end{aligned}\end{equation}
Notice that we can, without loss of generality, restrict ourselves to $p' \leq p$; were it the case that $p'(x) > p(x)$ for any $x$, defining $p'' \coloneqq \min \{ p, p' \}$ could only improve the objective function in~\eqref{eq:classical_opt}. This immediately tells us that the feasible set in the optimisation for~\eqref{eq:classical_opt} is a subset of the feasible set for the optimisation of $\primal(p,q,\ve)$, as the latter allows also non-positive distributions. This gives $Q^{\ve,T}_{2}(p\|q) \geq \primal(p,q,\ve)$.

Take now the distribution $p'$ as in~\eqref{eq:pdash}. Denoting $S \coloneqq \{ x : p(x) \leq \gamma q(x) \}$ and $S^\perp \coloneqq \{ x : p(x) > \gamma q(x) \}$, we have
\begin{equation}\begin{aligned}\label{eq:classical_upperbound}
  Q^{\ve,T}_{2}(p\|q) &\leq Q_2(p' \| q) = \sum_{x \in S } \frac{p(x)^2}{q(x)} + \sum_{x \in S^\perp} \gamma^2 q(x).
\end{aligned}\end{equation}
On the other hand, consider the operator $W$ defined as
\begin{equation}\begin{aligned}
  W(x) \coloneqq \begin{cases} \dfrac{p(x)}{\gamma q(x)} & \text{if } x \in S\\ \vphantom{\bigg(}1 & \text{if } x \in S^\perp. \end{cases}
\end{aligned}\end{equation}
By construction, $0 \leq W \leq \id$, so $W$ is a feasible solution for $\dual(p,q,\ve)$. But then, using the strong duality shown in Lemma~\ref{lem:D2_primaldual},
\begin{equation}\begin{aligned}
  Q^{\ve,T}_{2}(p\|q) &\geq \primal(p,q,\ve) \\
  &= \dual(p,q,\ve)\\
  &\geq \frac{\left(\sum_x W(x) p(x) - \ve\right)_+^2}{\sum_x W(x)^2 q(x)}\\
  &= \dfrac{\left(\displaystyle \sum_{x \in S} \dfrac{p(x)^2}{\gamma q(x)} +  \sum_{x \in S^\perp} p(x) - \Tr(p - \gamma q)_+ \right)_+^2}{\displaystyle  \sum_{x \in S} \dfrac{p(x)^2}{\gamma^2 q(x)} + \sum_{x \in S^\perp} q(x)} \\
    &= \dfrac{\left(\displaystyle \sum_{x \in S} \dfrac{p(x)^2}{\gamma q(x)} +  \sum_{x \in S^\perp} \gamma q(x) \right)^2}{\displaystyle  \sum_{x \in S} \dfrac{p(x)^2}{\gamma^2 q(x)} + \sum_{x \in S^\perp} q(x)}\\
    &= Q_2(p' \| q),
\end{aligned}\end{equation}
implying together with~\eqref{eq:classical_upperbound} that all the inequalities must in fact be equalities.
\end{proof}

The final ingredient is a monotonicity property under measurement channels, which we state in a general form.

\begin{lemma}[Data processing]\label{lem:D2_dataproc}
For any positive and trace--non-increasing map $\E$, it holds that
\begin{equation}\begin{aligned}
  \dual(\E(\rho),\E(\sigma),\ve) \leq \dual(\rho,\sigma,\ve).
\end{aligned}\end{equation}
\end{lemma}
\begin{proof}
Let $W$ be any feasible solution in the optimisation for $\dual(\E(\rho),\E(\sigma),\ve)$. As $\E^\dagger$ is a positive, sub-unital map, we have that $\E^\dagger(W) \in [0,\id]$ and hence $\E^\dagger(W)$ is a feasible solution for $\dual(\rho,\sigma,\ve)$. By Kadison's inequality~\cite[Theorem~2.3.2]{bhatia_2007} (cf.~\cite[Lemma~3.5]{hiai_2011}) sub-unitality and positivity of $\E^\dagger$ give $(\E^\dagger(W))^2 \leq \E^\dagger(W^2)$, implying that
\begin{equation}\begin{aligned}
  \dual(\rho,\sigma,\ve) &\geq  \frac{\left(\Tr (\E^\dagger(W) \rho) - \epsilon\right)_+^2}{\Tr \E^\dagger(W)^2 \sigma}\\
  &\geq \frac{\left(\Tr (\E^\dagger(W) \rho) - \epsilon\right)_+^2}{\Tr \E^\dagger(W^2) \sigma}\\
  &= \frac{\left(\Tr (W \E(\rho)) - \epsilon\right)_+^2}{\Tr W^2 \E(\sigma)}.
\end{aligned}\end{equation}
Optimising over all $0 \leq W \leq \id$ yields the statement of the lemma.
\end{proof}

We are finally ready to conclude the proof of the equivalence between $\Ds{2}{\ve}{M}(\rho\|\sigma)$ and the variational programs.

\begin{proof}[Proof of Theorem~\ref{thm:D2_smoothed_forms}]
For any measurement channel $\M$, data processing (Lemma~\ref{lem:D2_dataproc}) and the classical equivalence (Lemma~\ref{lem:D2_classicalcase}) give us
\begin{equation}\begin{aligned}
  \log \dual(\rho,\sigma,\ve) &\geq \log \dual(\M(\rho),\M(\sigma),\ve)\\
  &= \Ds{2}{\ve}{T}(\M(\rho)\|\M(\sigma)).
\end{aligned}\end{equation}
Optimising over measurements, we then get $\log \dual(\rho,\sigma,\ve) \geq \Ds{2}{\ve}{M}(\rho\|\sigma)$.

On the other hand, let $W$ be any operator feasible for the optimisation in $\dual(\rho,\sigma,\ve)$, i.e.\ $0 \leq W \leq \id$. Write it in its spectral decomposition as $W = \sum_i \lambda_i \Pi_i$. Choosing $\M_W$ to be the measurement channel corresponding to the measurement with POVM elements $(\Pi_i)_i$, we have
\begin{equation}\begin{aligned}
  \Ds{2}{\ve}{M}(\rho\|\sigma) &\geq \Ds{2}{\ve}{T}(\M_W(\rho)\|\M_W(\sigma))\\
  &= \log \dual(\M_W(\rho),\M_W(\sigma),\ve),
\end{aligned}\end{equation}
where the last line is by the classical case in Lemma~\ref{lem:D2_classicalcase}. Letting now $V$ be a classical (diagonal) operator defined through $V(i) = \lambda_i$, we have $0 \leq V \leq \id$, which means that it is a feasible solution for the dual optimisation. Thus
\begin{equation}\begin{aligned}
  \dual(\M_W(\rho),\M_W(\sigma),\ve) &\geq \frac{\left(\Tr (V M_W(\rho)) - \ve\right)_+^2}{\Tr V^2 M_W(\sigma)}\\
  &= \frac{\left(\Tr (W \rho) - \ve\right)_+^2}{\Tr W^2 \sigma}.
\end{aligned}\end{equation}
and optimising over all $W$ gives $\Ds{2}{\ve}{M}(\rho\|\sigma) \geq \log \dual(\rho,\sigma,\ve)$, and so the quantities must be equal. Together with the strong duality given in Lemma~\ref{lem:D2_primaldual} and the divergent case considered in Lemma~\ref{lem:D2_diverging_case}, this concludes the proof.
\end{proof}

We observe that $\Ds{2}{\ve}{M}$ can be computed as a semidefinite program, which can aid its evaluation in practice.
\begin{boxed}
\begin{corollary}\label{lem:D2_SDP}
For any quantum state $\rho$, any $\sigma \geq 0$, and any $\ve \in [0,1)$, the measured smooth collision divergence equals the optimal value of a semidefinite program. Specifically,
\begin{align}\label{eq:D2_smooth_SDP}
 \Ds{2}{\ve}{M}(\rho\|\sigma) &= \log \sup_{B,C,t} \lset 2 \Tr B \rho - 2 \ve t - \Tr C \sigma \bar B, C \geq 0,\; B \leq t \id,\; \begin{pmatrix}C & B \\ B & \id \end{pmatrix} \geq 0 \rset\\
 &= \log \inf_{Z,T,R} \lset \Tr T \bar T \geq 0,\; Z \in \mathbb{C}^{d \times d},\; \Re(Z) = R,\; R \leq \rho,\; \Tr(\rho-R) \leq \ve, \begin{pmatrix}\sigma & Z \\ Z^\dagger & T \end{pmatrix} \geq 0 \rset,\nonumber
\end{align}
where $\mathbb{C}^{d \times d}$ denotes the set of complex matrices of the same dimension as the states $\rho$ and $\sigma$, and $\Re(Z)=\frac12(Z+Z^\dagger)$.
\end{corollary}
\end{boxed}
\begin{proof}
Recall that in the course of the proof of Lemma~\ref{lem:D2_primaldual} (Eq.~\eqref{eq:D2smooth_dual}--\eqref{eq:D2smooth_dual2}) we showed that
\begin{equation}\begin{aligned}
 \Ds{2}{\ve}{M}(\rho\|\sigma) &= \log \sup_{B,t} \lset 2 \Tr B \rho - 2 \ve t - \Tr B^2 \sigma \bar  B \geq 0, \, B \leq t \id \rset.
\end{aligned}\end{equation}
Noting that $\Tr B^2 \sigma \leq \Tr C \sigma$ for any $C \geq B^2$, relaxing the above to an optimisation of the function $2 \Tr B \rho - 2\ve t - \Tr C \sigma$ over such $B$ and $C$ cannot change the optimal value. By a standard result on Schur complements, $C \geq B^2$ if and only if $\pig(\begin{smallmatrix}C & B \\ B & \id \end{smallmatrix}\pig) \geq 0$~\cite[Thm.~1.3.3, Ex.~1.3.5]{bhatia_2007}, from which the first equality in~\eqref{eq:D2_smooth_SDP} follows. The second line is a rewriting of the corresponding Lagrange dual.
\end{proof}

\subsection{One-shot divergence inequalities}

Before proceeding with applications, we derive a number of inequalities that connect the measured smooth collision divergence with other one-shot quantities encountered in quantum information. The results will find direct applications in the asymptotic study of privacy amplification in the second part of the manuscript. 
Their proofs will also showcase an extremely useful property of the measured smooth divergences: by their very definition, results can be shown for the case of classical distributions and then immediately lifted to quantum states.

We first establish inequalities that tightly connect $\Ds{2}{\ve}{M}$ with the measured smooth max-relative entropy $\Dmax{\ve}{M}$ as well as the \emph{hypothesis testing relative entropy} ${D^\ve_H}$~\cite{wang_2012,buscemi_2010}. The latter quantity, besides finding operational use in many problems, is often useful due to its well-studied properties in settings such as large deviations~\cite{hayashi_2007,audenaert_2008,mosonyi_2015} and second-order asymptotics~\cite{tomamichel_2013,li_2014}. The hypothesis testing relative entropy is  given simply by the best exponent of type II error probability in quantum hypothesis testing, subject to the type I error probability being at most $\ve$:
\begin{equation}\begin{aligned}\label{eq:dh_def}
    D^\ve_H (\rho \| \sigma) \coloneqq - \log \inf \lset \Tr M \sigma \bar 0 \leq M \leq \id,\; \Tr (\id - M) \rho \leq \ve \rset.
\end{aligned}\end{equation}
Although often studied alongside smooth divergences in quantum information, the definition of $D^\ve_H$ does not match standard smooth divergences as given in Eq.~\eqref{eq:smooth_def}, and it may not be clear how to interpret this quantity in the smooth entropic formalism. As a matter of fact, we can show that the hypothesis testing relative entropy naturally fits within the family of measured smooth R\'enyi divergences that we defined in~\eqref{eq:measured_smooth_def}.
\begin{lemma}\label{lem:DH_as_D0}
For all quantum states $\rho$, all $\sigma \geq 0$, and all $\ve \in [0,1)$, the hypothesis testing relative entropy equals the measured smooth R\'enyi divergence of order 0:
\begin{equation}\begin{aligned}
  D^\ve_H (\rho \| \sigma) = \Ds{0}{\ve}{M} (\rho \| \sigma) = \sup_{\M\in\MM} \Ds{0}{\ve}{T}(\M(\rho)\|\M(\sigma)).
\end{aligned}\end{equation}
\end{lemma}
This provides further justification for the definition of measured smooth R\'enyi divergences and unifies $D^\ve_H$, $\Ds{2}{\ve}{M}$, and $\Dmax{\ve}{M}$ in a single class of functions. 
We believe this to be the first direct interpretation of $D^\ve_H$ as part of a family of smooth divergences that includes variants of the conventional smooth R\'enyi entropies, and in particular the max-relative entropy (cf.~\cite{buscemi_2010}).  
Lemma~\ref{lem:DH_as_D0} here serves to demonstrate a conceptual connection but will not be used in subsequent proofs, so we defer its proof to Appendix~\ref{app:DH}. 

We now proceed with the derivation of the one-shot bounds for $\Ds{2}{\ve}{M}$. To connect with $D^\ve_H$, the proofs will make use of the strong quantitative connections between $D^\ve_H$ and $D_{\max}^{\ve,\mathbb{M}}$ that were recently shown in~\cite{regula_2025}. It is in fact shown there that the functions $\Dmax\ve M$ and $D^{1-\ve}_H$ are in a precise sense equivalent to each other, allowing one to understand bounds stated in terms of $\Dmax\ve M$ also in terms of $D^{1-\ve}_H$~\cite[Theorem~4]{regula_2025}.

\begin{boxed}
\begin{lemma}\label{lem:secondorder_equiv}
For any quantum states $\rho$ and $\sigma$, any $\ve \in (0,1)$, and any $\delta \in (0, 1-\ve)$, it holds that
\begin{align}
  &\begin{aligned}\Ds{2}{\ve}{M}(\rho\|\sigma) &\leq \Dmax{\ve}{M}(\rho\|\sigma) - \log\frac{1}{1-\ve} \\&\leq D^{1-\ve}_H(\rho\|\sigma) - \log \frac{1}{\ve (1-\ve)^2}, \end{aligned}\label{eq:bounds1}\\
  &\Ds{2}{\ve}{M}(\rho\|\sigma) \geq D^{1-\ve-\delta}_H(\rho\|\sigma) - \log \frac{1}{\delta^2}
  ,\label{eq:bounds2}\\
  &\Ds{2}{\ve}{M}(\rho\|\sigma) \geq \Dmax{\ve+\delta}{M} (\rho\|\sigma)  - \log \frac{1-\delta}{\delta}. \label{eq:bounds3}
\end{align}
\end{lemma}
\end{boxed}

One immediate consequence of the bounds is in establishing an exact second-order i.i.d.\ expansion of the measured smooth divergence $\Ds{2}{\ve}{M}$, which we will return to in Section~\ref{sec:secondorder}. We will also show the bounds to find use in the asymptotic analysis of strong converse exponents of privacy amplification in Section~\ref{sec:strong-converse}.

\begin{proof}[Proof of Lemma~\ref{lem:secondorder_equiv}]
From Theorem~\ref{thm:D2_smoothed_forms} we know that $\Ds{2}{\ve}{M}(\rho\|\sigma) = \infty$ iff $\Tr \Pi_\sigma^\perp \rho > \ve$, in which case also $\Dmax{\ve}{M}(\rho\|\sigma) = D^{1-\ve}_{H}(\rho\|\sigma) = \infty$ (see e.g.~\cite{regula_2025}). We can thus exclude such diverging cases as the bounds are then trivially true.

To see~\eqref{eq:bounds1}, consider first any classical distributions $p$ and $q$. 
By Lemma~\ref{lem:D2_classicalcase}, we know that the distribution
\begin{equation}\begin{aligned}
  p'(x) \coloneqq \begin{cases} p(x) & \text{if } p(x) \leq \gamma q(x)\\ \gamma q(x) & \text{if } p(x) > \gamma q(x), \end{cases}
\end{aligned}\end{equation}
with $\gamma$ chosen so that $\Tr(p - \gamma q)_+ = \ve$ is optimal for both $\Ds{2}{\ve}{T}$ and $\Dmax{\ve}{T}$, in the sense that
\begin{equation}\begin{aligned}
  \Ds{2}{\ve}{T}(p\|q) = D_2(p' \| q), \qquad \Dmax{\ve}{T}(p\|q) = D_{\max}(p'\|q) = \log \gamma.
\end{aligned}\end{equation}
Denoting $S = \{ x : p(x) \leq \gamma q(x) \}$ and $S^\perp = \{ x : p(x) > \gamma q(x) \}$, we thus have
\begin{align}
  Q^{\ve,T}_{2}(p\|q) &=\sum_{x \in S } p(x)^2\, q(x)^{-1} + \sum_{x \in S^\perp} \gamma^2 q(x) \nonumber\\
  &\leq \sum_{x \in S } p(x)\, \gamma q(x) \, q(x)^{-1} + \sum_{x \in S^\perp} \gamma^2 q (x) \nonumber\\
  &= \gamma \left( \sum_{x \in S } p(x) + \sum_{x \in S^\perp} \gamma q(x) \right)\\
  &= \gamma \left( \sum_x p(x) - \sum_{x \in S^\perp} \big( p(x) - \gamma q(x) \big) \right) \nonumber\\
  &= \gamma \, ( 1- \ve ) \nonumber
\end{align}
since $\Tr (p - \gamma q)_+ = \ve$. 
Applying this to the choice $p = \M(\rho)$ and $q = \M(\sigma)$ gives
\begin{equation}\begin{aligned}
  \Ds{2}{\ve}{T}(\M(\rho)\|\M(\sigma)) &\leq \Dmax{\ve}{T}(\M(\rho)\|\M(\sigma)) - \log \frac{1}{1-\ve}.
\end{aligned}\end{equation}
Taking the supremum of both sides over measurements gives the first inequality. The second inequality is then by~\cite[Lemma~7]{regula_2025}.

For~\eqref{eq:bounds2}, start again with classical distributions $p$ and $q$. Pick $p'$ such that $D^\ve_2(p\|q) = D_2(p'\|q)$. A standard argument based on the data processing of the R\'enyi divergences (see e.g.~\cite[Lemma~IV.7]{mosonyi_2015}) gives $D^{1-\delta}_H(p'\|q) \leq D_\alpha(p'\|q) + \frac{\alpha}{\alpha-1} \log \frac{1}{\delta}$ for all $\alpha >1$, and in particular for $\alpha=2$. Since $\|p - p'\|_+ \leq \ve$, an application of triangle inequality gives $D^{1-\delta-\ve}_H(p\|q) \leq D^{1-\delta}_H(p'\|q) \leq D_2(p'\|q) + 2 \log\frac{1}{\delta}$. The inequality for quantum states is then obtained by lifting the classical result through measurements as above.

The proof of~\eqref{eq:bounds3} is analogous. Pick $p'$ such that $D^\ve_2(p\|q) = D_2(p'\|q)$. Use the fact that $D^{\ve+\delta}_{\max}(p\|q) \leq D^{\delta}_{\max}(p'\|q) \leq D_\alpha(p'\|q) + \frac{1}{\alpha-1}\log\frac{1}{\delta} - \log\frac{1}{1-\delta}$, with the last inequality by~\cite[Corollary~14]{regula_2025}. The claimed result follows by optimising over measurements.
\end{proof}

We note that the above already strengthens the na\"ive upper bound $\Ds{2}{\ve}{M} \leq \Dmax{\ve}{M}$. 
Another strengthening of this relation can be obtained by employing R\'enyi divergences to allow for a type of extrapolation from the values of the two smooth quantities.

\begin{boxed}
\begin{lemma}\label{lem:D2_upper_exponent_bound}
For any $\ve \in (0,1)$ and any $\alpha \in (1,2]$, it holds that
\begin{equation}\begin{aligned}\label{eq:D2_upper_bounds}
  \Ds{2}{\ve}{M}(\rho\|\sigma) &\leq (\alpha-1)\, D_\alpha^{\MM}(\rho\|\sigma) + (2-\alpha)\, \Dmax{\ve}{M}(\rho\|\sigma)\\
  &\leq D_\alpha^{\MM} (\rho\|\sigma) + \frac{2-\alpha}{\alpha-1} \log \frac{1}{\ve} - (2-\alpha) \log \frac{1}{1-\ve}.
\end{aligned}\end{equation}
\end{lemma}
\end{boxed}
Our derivation here adapts and extends a classical proof approach of Hayashi~\cite{hayashi_2013}. The first inequality in~\eqref{eq:D2_upper_bounds} in fact strictly improves on the `hybrid bound' of~\cite[Theorem~6]{watanabe_2013} even in the classical case, and the second inequality gives a generalisation and slight tightening of the classical findings in~\cite[Theorem~1]{hayashi_2013}. We will see in Section~\ref{sec:renyi_bounds} that this bound is a key ingredient for the derivation of exponential constraints in the large deviation regime.
\begin{proof}
Consider first classical probability distributions $p$ and $q$. As before, choosing $\gamma$ so that $\Tr(p - \gamma q)_+ = \ve$ and letting $S = \{ x : p(x) \leq \gamma q(x) \}$ and $S^\perp = \{ x : p(x) > \gamma q(x) \}$, we have
\begin{equation}\begin{aligned}
  Q^{\ve,T}_2(p\|q) &= \sum_{x \in S} p(x)^2 \,q(x)^{-1} + \sum_{x\in S^\perp} \gamma^{2}\, q(x)\\
  &=\sum_{x \in S} p(x)^\alpha\, p(x)^{2-\alpha}\, q(x)^{-1} + \sum_{x\in S^\perp} \gamma^{\alpha}\, q(x)^\alpha\, \gamma^{2-\alpha}\, q(x)^{1-\alpha}\\
  &\leq \sum_{x \in S} p(x)^\alpha \,\gamma^{2-\alpha}\, q(x)^{2-\alpha} \,q(x)^{-1} + \sum_{x\in S^\perp} p(x)^\alpha\, \gamma^{2-\alpha}\, q(x)^{1-\alpha}\\
  &= Q_\alpha(p\|q) \, \gamma^{2-\alpha}.
\end{aligned}\end{equation}
Taking logarithms, we can write this as
\begin{equation}\begin{aligned}
  \Ds{2}{\ve}{T}(p\|q) \leq (\alpha-1) \,D_\alpha(p\|q) + (2-\alpha)\, \Dmax{\ve}{T}(p\|q).
\end{aligned}\end{equation}
Choosing $p = \M(\rho)$ and $q = \M(\sigma)$ we have
\begin{equation}\begin{aligned}
  \Ds{2}{\ve}{T}(\M(\rho)\|\M(\sigma)) &\leq (\alpha-1)\, \sup_{\M \in \MM} D_\alpha(\M(\rho)\|\M(\sigma)) + (2-\alpha)\, \sup_{\M \in \MM} \Dmax{\ve}{T}(\M(\rho)\|\M(\sigma))\\
  &=(\alpha-1)\, D_\alpha^{\MM}(\rho\|\sigma) + (2-\alpha)\, \Dmax{\ve}{M}(\rho\|\sigma),
\end{aligned}\end{equation}
by definition of $D_{\alpha}^{\MM}$ and $\Dmax{\ve}{M}$. Maximising over $\M$ on the left-hand side and using the definition of $\Ds{2}{\ve}{M}$ lifts this inequality to $\Ds{2}{\ve}{M}(\rho\|\sigma)$. We conclude by recalling from~\cite[Corollary~14]{regula_2025} that
\begin{equation}\begin{aligned}
  \Dmax{\ve}{M}(\rho\|\sigma) &\leq D_\alpha^{\MM}(\rho\|\sigma) + \frac{1}{\alpha-1} \log \frac{1}{\ve} - \log\frac{1}{1-\ve}
\end{aligned}\end{equation}
for all $\alpha > 1$.
\end{proof}

Although the bound of Lemma~\ref{lem:D2_upper_exponent_bound} already features an optimal scaling in $\ve$ and would thus be sufficient for us in the study of the asymptotics of $\Ds{2}{\ve}{M}$ and of privacy amplification, one could pursue a further tightening of the constants involved in the bound. 
Shortly after the appearance of the preprint of this paper, the independent work~\cite{gour_2026} studied R\'enyi divergence inequalities and observed that a tighter inequality between the smooth collision divergence and R\'enyi divergences can be obtained. We give an alternative direct proof of this bound, which may be of independent interest.
\begin{boxed}
\begin{lemma}[{\cite[Theorem~1]{gour_2026}}]\label{lem:D2_bound_with_h2}
For all $\ve \in (0,1)$ and all $\alpha \in (1,2]$, it holds that
\begin{equation}\begin{aligned}
  \Ds{2}{\ve}{M}(\rho\|\sigma) &\leq D_\alpha^{\MM} (\rho\|\sigma) + \frac{2-\alpha}{\alpha-1} \log \frac1\ve - \frac{\alpha}{\alpha-1}\, h\!\left(\frac{2-\alpha}{\alpha}\right)\\
  &= D_\alpha^{\MM} (\rho\|\sigma) + \frac{2-\alpha}{\alpha-1} \log \frac1\ve + \frac{2-\alpha}{\alpha-1} \log \frac{2-\alpha}{\alpha} + 2 \log \frac{2 (\alpha-1)}{\alpha},
\end{aligned}\end{equation}
where $h(a) \coloneqq -a \log a - (1-a) \log (1-a)$ denotes the binary entropy.
\end{lemma}
\end{boxed}
\begin{proof}
As before, we will show the inequality for classical distributions $p$ and $q$, from which the quantum extension will follow by lifting both sides by measurements. Recall the variational form of smooth collision divergence from Theorem~\ref{thm:D2_smoothed_forms}, which we will write here as
\begin{equation}\begin{aligned}\label{eq:class_variational}
  Q_2^\ve (p\|q) &= \sup \lset \frac{\left(\Tr Wp - \ve\right)^2}{\Tr W^2 q} \bar 0 \leq W \leq \id,\; \Tr W p > \ve \rset
\end{aligned}\end{equation}
with $W$ here being classical (diagonal) operators. 
Using H\"older's inequality for the conjugate pair of coefficients $\alpha$ and $\frac{\alpha}{\alpha-1}$, it h\"olds that
\begin{equation}\begin{aligned}
  \Tr Wp &= \sum_x \big(p(x)^\alpha \big)^{1/\alpha} W(x)\\
  &= \sum_x \left(p(x)^\alpha q(x)^{1-\alpha} \right)^{1/\alpha} W(x) \, q(x)^{\frac{\alpha-1}{\alpha}}\\
  &\leq \left( \sum_x p(x)^\alpha q(x)^{1-\alpha} \right)^{1/\alpha} \left( \sum_x W(x)^{\frac\alpha{\alpha-1}} \, q(x) \right)^{\frac{\alpha-1}{\alpha}}\\
  &\leq Q_\alpha(p\|q)^{1/\alpha} \, \left( \sum_x W(x)^{2} \, q(x) \right)^{\frac{\alpha-1}{\alpha}}
\end{aligned}\end{equation}
where in the last line we used that $W(x) \in [0,1]$ and $\frac{\alpha}{\alpha-1} \in [2, \infty)$. Rewriting this as $\sum_x W(x)^{2} \, q(x) \geq \big(\!\Tr Wp \big)^{\frac\alpha{\alpha-1}}\, Q_\alpha(p\|q)^{- \frac{1}{\alpha-1}}$
and plugging into Eq.~\eqref{eq:class_variational} gives
\begin{equation}\begin{aligned}
  Q_2^\ve (p\|q) &\leq \sup_{W : \Tr Wp \in (\ve,\infty)} \frac{ \left(\Tr Wp - \ve\right)^2 }{\big(\!\Tr Wp \big)^{\frac\alpha{\alpha-1}}} \, Q_\alpha(p\|q)^{\frac{1}{\alpha-1}}.
\end{aligned}\end{equation}
Introduce now the variable $u \coloneqq \frac{\ve}{\Tr Wp} \in (0,1)$ to write the above as
\begin{equation}\begin{aligned}
   Q_2^\ve (p\|q) &\leq \sup_{u \in (0,1)} \left( \frac1u - 1 \right)^2 \ve^2 \, u^{\frac\alpha{\alpha-1}}\, \ve^{-\frac\alpha{\alpha-1}} \, Q_\alpha(p\|q)^{\frac{1}{\alpha-1}}\\
   &= \sup_{u \in (0,1)}\,  (1-u)^2 \, u^{\frac{2-\alpha}{\alpha-1}} \, \ve^{- \frac{2-\alpha}{\alpha-1}} \, Q_\alpha(p\|q)^{\frac{1}{\alpha-1}}\\
   &= Q_\alpha(p\|q)^{\frac{1}{\alpha-1}} \, \ve^{- \frac{2-\alpha}{\alpha-1}} \,  \sup_{u \in (0,1)} \exp \left[ \frac{\alpha}{\alpha-1} \left( - h\!\left(\frac{2-\alpha}{\alpha}\right) - D\!\left(\frac{2-\alpha}{\alpha} \middle\| u \right) \right) \right]\\
   &= Q_\alpha(p\|q)^{\frac{1}{\alpha-1}} \, \ve^{- \frac{2-\alpha}{\alpha-1}} \, \exp \left[ - \frac{\alpha}{\alpha-1} \,h\!\left(\frac{2-\alpha}{\alpha}\right) \right],
\end{aligned}\end{equation}
where we used $D$ to denote the binary relative entropy $D(a\|b) = D((a,1\!-\!a)\|(b,1\!-\!b))$ and observed that, by the non-negativity of $D$, the supremum is clearly achieved at $u = \frac{2-\alpha}{\alpha}$. Choosing $p = \M(\rho)$ and $q = \M(\sigma)$ and optimising over measurements $\M$ gives the claimed result for the measured smooth collision divergence.
\end{proof}

\section{Achievability for privacy amplification}\label{sec:achievability}

Randomness extraction typically makes use of hash functions with constrained collision probabilities, and different variants of such hash families have appeared in the literature. 
We use the term \deff{\univ-universal} to refer to any family of hash functions $f\colon \X \to \Z$ such that
\begin{equation}\begin{aligned}\label{eq:2univ}
  \operatorname*{Pr}\limits_{f} \big[ f(x) = f(x') \big] = \frac{1}{|\Z|} \;\, \forall x \neq x'.
\end{aligned}\end{equation}
Such functions are sometimes simply called \emph{2-universal} in the literature; our choice of terminology is made to distinguish this definition from the more standard  definition of 2-universal hash functions that only imposes that $\operatorname*{Pr}_f\! \big[ f(x) = f(x') \big] \leq \frac{1}{|\Z|}$~\cite{carter_1979}.
The equality condition in~\eqref{eq:2univ} will be important when dealing with non-positive operators. 
The definition of \univ-universality is less demanding than so-called \emph{strong 2-universality}~\cite{wegman_1981,stinson_1994}, which would require pairwise independence of the random variables $f(x)$ and $f(x')$.
Many classes of universal hash functions, including those commonly found in cryptographic applications such as linear Toeplitz hashing, are indeed \univ-universal~\cite{carter_1979,krawczyk_1994}. 

Conventional formulations of the leftover hash lemma tell us that, for any 2-universal family of hash functions, it holds that~\cite{renner_2005,tomamichel_2011}
\begin{equation}\begin{aligned}
  \EE_f \norm{ \rho_{ZE}^f - \frac{\id_Z}{|\Z|} \otimes \rho_E }{+} \leq \frac12 \sqrt{ \big(|\Z|-1\big) \, \wt{Q}_2(\rho_{XE} \| \id_X \otimes \sigma_E) },
\end{aligned}\end{equation}
recalling that $\wt{Q}_2 = \exp \wt{D}_2$. 
It is not difficult to see by inspecting the standard proofs that, with the assumption of \univ-universality, this can be applied not only to states $\rho_{XE}$ but also to Hermitian operators $R_{XE}$ --- a fact that was used e.g.\ in~\cite{dupuis_2023}. This could already suggest that employing Hermitian smoothing may lead to improvements in standard approaches. 
However, a na\"ive application of the previous methods would result in a quantity of the form $\wt{D}_2(R_{XE} \| \id_X \otimes \sigma_E)$ based on standard sandwiched R\'enyi divergence, which does not appear to improve over prior bounds, and it is unclear how to connect it with tighter smooth divergences like $\Ds{2}{\ve}{M}$ or~$\Dmax{\ve}{M}$.

Through a somewhat different approach, the work \cite{shen_2024} observed that an application of spectral pinching together with a careful analysis of the interplay between the trace distance and pinched projectors gives an achievability result under \univ-universal hash functions that exhibits improved scaling at the level of second-order asymptotics. However, at the one-shot level the result incurs a number of penalty terms stemming from the pinching inequality and the use of the suboptimal information spectrum divergence $D^\ve_s$; this also prevents its applicability for large-deviation bounds. 

To improve on the prior findings, we derive a new variant of the leftover hash lemma.

\subsection{Leftover hashing with measured smooth collision divergence}\label{sec:leftover_hashing}

\begin{boxed}[color]
\begin{lemma}[Tightened leftover hash lemma]\label{lem:leftover}
For any classical--quantum Hermitian operator $R_{XE}$, any state $\sigma_{E}$, and any \univ-universal family of hash functions $f \colon \X \to \Z$, it holds that
\begin{equation}\begin{aligned}
  \EE_f \norm{ R_{ZE}^f - \frac{\id_Z}{|\Z|} \otimes R_E }{+} \leq \frac12 \sqrt{|\Z|-1}\, \norm{R_{XE}}{\Bures\id_X \otimes \sigma_E},
\end{aligned}\end{equation}
where we recall that $\norm{R}{\Bures\sigma} = \sqrt{\Tr R \J_\sigma^{-1}(R)}$, with $\J_\sigma(X) = \frac12 (\sigma X + X \sigma)$.
\end{lemma}
\end{boxed}
The proof relies on the following fact, which extends~\cite[Lemma~5]{temme_2010}. 

\begin{lemma}[\cite{temme_2010}]\label{lem:variance}
For any operator $\sigma \geq 0$ and any Hermitian operator $H$,
\begin{equation}\begin{aligned}
  \norm{H}{1} \leq \norm{H}{\Bures\sigma} \!\sqrt{\vphantom{\norm{H}{B}}\Tr \sigma}.
\end{aligned}\end{equation}
\end{lemma}
\begin{proof}
If $H \not\in \ran(\J_\sigma)$, then the right-hand side diverges and the inequality is trivial. 
Assuming otherwise, for any Hermitian operator $K$ the Cauchy--Schwarz inequality for the inner product $\<X,Y\>_\sigma = \Tr X \J_\sigma^{-1}(Y)$ gives
\begin{equation}\begin{aligned}
 \Tr HK = \< H, \J_\sigma(K) \>_\sigma \leq\sqrt{\Tr H \J_\sigma^{-1}(H)}  \sqrt{\vphantom{\J_\sigma^{-1}} \Tr \J_\sigma(K) K } =  \norm{H}{\Bures\sigma} \sqrt{\Tr K^2 \sigma}.
\end{aligned}\end{equation}
If $\norm{K}{\infty} \leq 1$, then $0 \leq K^2 \leq \id$, and hence $\Tr K^2 \sigma \leq \Tr \sigma$. The result then follows by using that $\norm{H}{1} = \max \lset \Tr HK \bar \norm{K}{\infty} \leq 1 \rset$. 
\end{proof}

\begin{remark}
Although above we opted for a direct proof, another way to see Lemma~\ref{lem:variance} is to observe that the classical case of it follows simply as $\norm{h}{1} = \sum_x |h(x)| \leq \left( \sum_x h(x)^2 q(x)^{-1} \right)\!{}^{1/2} \left( \sum_x q(x) \right)\!{}^{1/2} = \norm{h}{q} \sqrt{\Tr q}$, and then both sides of the inequality can be lifted to quantum states by choosing $h = \M(H)$, $q = \M(\sigma)$, and maximising over measurements. This also shows intuitively that the variant that we consider here is the tightest possible lifting of the classical inequality: the Bures seminorm is the minimal extension of the corresponding classical seminorm that obeys operator monotonicity (as a metric) or data processing (as a divergence).
\end{remark}

With this in place, the rest of the proof proceeds in a way analogous to standard leftover hashing.

\begin{proof}[{Proof of Lemma~\ref{lem:leftover}}]
Let $R_{E,x}$ denote the operators such that $R_{XE} = \sum_x \proj{x} \otimes R_{E,x}$. To avoid a trivial inequality, we assume without loss of generality that $R_{E,x} \in \ran(\J_{\sigma_E})$ for each $x$. Write
\begin{equation}\begin{aligned}
  R_{ZE}^f = \sum_z \proj{z} \otimes R_{E,z}^{f} \coloneqq \sum_z \proj{z} \otimes \sum_{x : f(x) = z} R_{E,x}.
\end{aligned}\end{equation}
Applying Lemma~\ref{lem:variance} with the choice $\sigma = \id_Z \otimes \sigma_E$ tells us that 
\begin{equation}\begin{aligned}\label{eq:yeah}
  \norm{ R_{ZE}^f - \frac{\id_Z}{|\Z|} \otimes R_E }{+} &= \frac12 \norm{ R_{ZE}^f - \frac{\id_Z}{|\Z|} \otimes R_E }{1}\\
  &\leq\frac12 \sqrt{|\Z| \norm{ R_{ZE}^f -  \frac{\id_Z}{|\Z|} \otimes R_E }{\Bures\id_Z \otimes \sigma_E}^2}\\
  &= \frac12 \sqrt{ |\Z| \sum_z \norm{ R_{E,z}^f -  \frac{1}{|\Z|} R_E }{\Bures \sigma_E}^2 }\,,
\end{aligned}\end{equation}
where the last line is seen as follows. Since the superoperator $\J_{\sigma}$ satisfies that $\J_{\id_Z \otimes \sigma_E} = \idc_Z \otimes \J_{\sigma_E}$, with $\idc$ standing for the identity channel, this means that also $\J_{\id_Z \otimes \sigma_E}^{-1} = \idc_Z \otimes \J_{\sigma_E}^{-1}$ on $\ran(\J_{\sigma_E})$. This makes the inner product $\< X, Y \>_\sigma = \Tr X \J_{\sigma}^{-1}(Y)$ and the induced squared Bures norm $\|\cdot\|_\sigma^2$ inherit the property that
\begin{equation}\begin{aligned}
\< Q_{ZE}, Q_{ZE} \>_{\id_Z \otimes \sigma_E} = \Tr \sum_z \proj{z} \otimes Q_{E,z} \J_{\sigma_E}^{-1}(Q_{E,z}) = \sum_z \< Q_{E,z}, Q_{E,z} \>_{\sigma_E}
\end{aligned}\end{equation}
for any CQ operator $Q_{ZE} = \sum_z \proj{z} \otimes Q_{E,z}$. 
Now, taking expectation and using Jensen's inequality for the square root gives
\begin{equation}\begin{aligned}\label{eq:after_jensen}
  \EE_f \norm{ R_{ZE}^f - \frac{\id_Z}{|\Z|} \otimes R_E }{+} &\leq  \frac12 \sqrt{ \left|\Z\right| \, \EE_f  \sum_z \norm{ R_{E,z}^f -  \frac{1}{|\Z|}R_E }{\Bures\sigma_E}^2 }.
\end{aligned}\end{equation}
Noting that $\sum_z R_{E,z}^f = R_E$, we can rewrite the term inside the square root as
\begin{equation}\begin{aligned}
  \EE_f \sum_z \norm{ R_{E,z}^f - \frac{1}{|\Z|} R_E }{\Bures \sigma_E}^2  &= \EE_f \sum_z \< R_{E,z}^f -  \frac{1}{|\Z|}R_E , R_{E,z}^f -  \frac{1}{|\Z|}R_E \>_{\sigma_E}\\
  &= \EE_f \sum_z \norm{ R_{E,z}^f }{\Bures\sigma_E}^2 - \frac{1}{|\Z|} \norm{R_E}{\Bures\sigma_E}^2.
\end{aligned}\end{equation}
The \univ-universality of $f$, i.e.\ the assumption that  $\operatorname*{Pr}_f [ f(x) = f(x')] = \frac{1}{|\Z|}$ for all $x \neq x'$, then ensures that
\begin{equation}\begin{aligned}\label{eq:exact_equality}
\EE_f \sum_z \norm{ R_{E,z}^f - \frac{1}{|\Z|} R_E }{\Bures \sigma_E}^2  &= \EE_f \sum_z \< \sum_{x : f(x) = z} R_{E,x}\,,\, \sum_{x' : f(x') = z} R_{E,x'} \>_{\sigma_E} - \frac{1}{|\Z|} \norm{R_E}{\Bures\sigma_E}^2\\
  &= \EE_f \sum_{x,x'} \delta_{f(x), f(x')} \< R_{E,x}, R_{E,x'} \>_{\sigma_E} - \frac{1}{|\Z|} \norm{R_E}{\Bures\sigma_E}^2\\
  &= \sum_{x} \norm{R_{E,x}}{\Bures\sigma_E}^2 + \EE_f \sum_{x\neq x'} \delta_{f(x), f(x')} \< R_{E,x}, R_{E,x'} \>_{\sigma_E} - \frac{1}{|\Z|} \norm{R_E}{\Bures\sigma_E}^2\\
  &= \sum_{x} \norm{R_{E,x}}{\Bures\sigma_E}^2 + \frac{1}{|\Z|} \sum_{x\neq x'} \< R_{E,x}, R_{E,x'} \>_{\sigma_E} - \frac{1}{|\Z|} \norm{R_E}{\Bures\sigma_E}^2\\
  &= \sum_{x} \norm{R_{E,x}}{\Bures\sigma_E}^2 + \frac{1}{|\Z|} \left( \norm{ \sum_x R_{E,x} }{\Bures\sigma_E}^2 - \sum_x \norm{R_{E,x}}{\Bures\sigma_E}^2  \right) - \frac{1}{|\Z|} \norm{R_E}{\Bures\sigma_E}^2\\
  &= \left(1 - \frac{1}{|\Z|} \right) \sum_{x} \norm{R_{E,x}}{\Bures\sigma_E}^2.
\end{aligned}\end{equation}
Plugging this back into~\eqref{eq:after_jensen}, we thus have
\begin{equation}\begin{aligned}
   \EE_f \norm{ R_{ZE}^f - \frac{\id_Z}{|\Z|} \otimes R_E }{+} &\leq \frac12 \sqrt{ \big(|\Z|-1\big) \, \sum_{x} \norm{R_{E,x}}{\Bures\sigma_E}^2 },
\end{aligned}\end{equation}
and the proof is completed by recalling that $\sum_{x} \norm{R_{E,x}}{\Bures\sigma_E}^2  = \norm{R_{XE}}{\Bures\id_X \otimes \sigma_E}^2$.
\end{proof}

\begin{remark}
We made use of \univ-universality of the hash functions in~\eqref{eq:exact_equality}; if we could ensure the positivity of the overlaps $\< R_{E,x}, R_{E,x'} \>_{\sigma_E}$, then 2-universality would have sufficed. Indeed, some previous incarnations of the leftover hash lemma such as~\cite{tomamichel_2011} used an even weaker assumption of almost 2-universality of the hash functions, but this does not immediately generalise to the case where the operator $R_{XE}$ is possibly non-positive.
\end{remark}

Combining this result with our characterisation of the measured smooth collision divergence, we can then give our main result in the form of a `smoothed' leftover hash lemma. It will in particular establish a connection with the measured smooth conditional entropies, which we define in full analogy with our previous definitions as
\begin{equation}\begin{aligned}
  \Hs{\alpha}{\ve}{M}[up] (X|E)_\rho = -\!\! \inf_{\substack{\sigma_E \geq 0\\\Tr \sigma_E = 1}} \Ds{\alpha}{\ve}{M} (\rho_{XE}\|\id_X \otimes \sigma_E).
\end{aligned}\end{equation}

\begin{boxed}[filled]
\begin{theorem}\label{thm:privacy_amp}
For any classical--quantum state $\rho_{XE}$, any state $\sigma_{E}$, any $\ve \in [0,1)$, and any \univ-universal family of hash functions  $f \colon \X \to \Z$, it holds that
\begin{equation}\begin{aligned}\label{eq:achiev_optimised_R}
  \EE_f \frac12 \norm{ \rho_{ZE}^f - \frac{\id_Z}{|\Z|} \otimes \rho_E }{1} &\leq \ve + \frac12 \sqrt{\big(|\Z|-1\big) \, \Qs{2}{\ve}{M}(\rho_{XE} \| \id_X \otimes \sigma_E)}\\
  &\leq \ve + \frac12 \sqrt{\big(|\Z|-1\big)\, (1-\ve) \, \Qs{\max}{\ve}{M}(\rho_{XE} \| \id_X \otimes \sigma_E)}.
\end{aligned}\end{equation}

As a result, for any $\ve \in (0,1)$ and any $\mu \in (0,\ve)$ there exists a randomness extraction protocol such that
\begin{equation}\begin{aligned}\label{eq:achiev_Hmin}
  \ell_\ve(\rho_{XE}) &\geq \Hs{2}{\ve-\mu}{M}[up] (X|E)_\rho - \log\frac{1}{4\mu^2}\\
  &\geq \Hmin{\ve-\mu}{M}[up] (X | E)_\rho - \log\frac{1-\ve + \mu}{4\mu^2}. 
\end{aligned}\end{equation}
\end{theorem}
\end{boxed}
\begin{remark}
As the result of this theorem concerns the practical task of randomness extraction, a reader interested in applications in this context may be concerned about the reliance on seemingly unphysical non-positive smoothing operators $R_{XE}$ instead of conventional smoothing in terms of states. We stress that the role played by smoothing in the leftover hash lemma is purely mathematical --- our Hermitian smoothing allows us to obtain a tighter estimate on the amount of randomness that is extracted in the protocol, but the only operation that needs to be physically implemented is the application of a random \univ-universal hash function.
\end{remark}
\begin{proof}
Theorem~\ref{thm:D2_smoothed_forms} tells us that the right-hand side is only finite when $\Tr \Pi_{\sigma_E}^\perp \rho_E \leq \ve$, so let us assume this to avoid trivial cases. 
Recall then the variational form established in Theorem~\ref{thm:D2_smoothed_forms} as
\begin{equation}\begin{aligned}\label{eq:Q2_variational}
   \Qs{2}{\ve}{M} (\rho_{XE} \| \id_X \otimes \sigma_E) = \min \lset \norm{R_{XE}}{\Bures\id_X \otimes \sigma_E}^2 \bar R_{XE} = R^\dagger_{XE},\; \norm{\rho_{XE} - R_{XE}}{+} \leq \ve,\; R_{XE} \leq \rho_{XE} \rset.
\end{aligned}\end{equation}
The basic idea of our proof will be to take an optimal $R_{XE}$ in this optimisation and use it, instead of $\rho_{XE}$, in the leftover hash result of Lemma~\ref{lem:leftover}. Let us carefully consider what properties of $R_{XE}$ will be needed for this.

For any Hermitian $R_{XE}$ satisfying $R_{E} \leq \rho_{E}$ and any hash function $f \colon \X \to \Z$, one has
\begin{equation}\begin{aligned}
  \frac12 \norm{ \rho_{ZE}^f - \frac{\id_Z}{|\Z|} \otimes \rho_E }{1} &= \Tr \left( \rho_{ZE}^f - \frac{\id_Z}{|\Z|} \otimes \rho_E \right)_+\\
  &\leq \Tr \left( \rho_{ZE}^f - \frac{\id_Z}{|\Z|} \otimes R_E \right)_+\\
  &\leq \norm{ \rho_{ZE}^f - R_{ZE}^f }{+} + \norm{R_{ZE}^f - \frac{\id_Z}{|\Z|} \otimes R_E}{+}\\
  &\leq \norm{ \rho_{XE} - R_{XE} }{+} + \norm{R_{ZE}^f - \frac{\id_Z}{|\Z|} \otimes R_E}{+},
\end{aligned}\end{equation}
with the last line by data processing of the generalised trace distance. 
Assuming additionally that $R_{XE}$ is a classical--quantum operator, we could invoke Lemma~\ref{lem:leftover} to give
\begin{equation}\begin{aligned}\label{eq:achiev_any_R}
  \EE_f \frac12 \norm{ \rho_{ZE}^f - \frac{\id_Z}{|\Z|} \otimes \rho_E }{1} \leq \| \rho_{XE} - R_{XE} \|_+ + \frac12 \sqrt{|\Z|-1}\, \norm{R_{XE}}{\Bures\id_X \otimes \sigma_E}.
\end{aligned}\end{equation}
If $R_{XE}$ were chosen as an optimiser of~\eqref{eq:Q2_variational}, this would conclude the proof.

We thus only need to argue that we can always choose an optimal $R_{XE}$ in~\eqref{eq:Q2_variational} that has a classical--quantum structure and thus satisfies the conditions for the applicability of Lemma~\ref{lem:leftover}. To see this, it suffices to observe the data-processing inequality $\norm{\E(R)}{\Bures\sigma} \leq \norm{R}{\Bures\sigma}$ for any channel $\E$ such that $\E(\sigma)=\sigma$. This allows us to apply the dephasing channel $\Delta_X (\cdot) = \sum_{x} \proj{x} \cdot \proj{x}$ to any feasible $R_{XE}$ without sacrificing the feasible optimal value, since both $\rho_{XE}$ and $\id_X \otimes \sigma_E$ are invariant under such dephasing. One way to see that this monotonicity holds also for Hermitian $R$ is to use the fact that the inner product $\<\cdot,\cdot\>_{\sigma}$ in the definition of $\norm{\cdot}{\Bures\sigma}$ defines a monotone operator metric~\cite{petz_1996,lesniewski_1999}. More explicitly, following Lemma~\ref{lem:D2_primaldual} we obtain the variational form
\begin{equation}\begin{aligned}
  \norm{R}{\Bures\sigma}^2 = \sup_{X=X^\dagger} \frac{\Tr(XR)^2}{\Tr X^2 \sigma}
\end{aligned}\end{equation}
valid for all Hermitian $R$, and an application of Kadison's inequality as in Lemma~\ref{lem:D2_dataproc} shows that monotonicity in fact holds under any positive, trace non-increasing  $\E$ that satisfies $\E(\sigma) \leq \sigma$.

Altogether, we can then choose an optimal $R_{XE}$ for the variational form of Eq.~\eqref{eq:Q2_variational} and apply Eq.~\eqref{eq:achiev_any_R} to give the first inequality in~\eqref{eq:achiev_optimised_R}; the second then follows by Lemma~\ref{lem:secondorder_equiv}. To see~\eqref{eq:achiev_Hmin}, observe that, for any $\sigma_E$, choosing $|\Z|$ so that $\ve - \mu + \frac12 \sqrt{|\Z|\, \Qs{2}{\ve-\mu}{M} (\rho_{XE}\|\id_X \otimes \sigma_E)} = \ve$ ensures that $|\Z|$ is an achievable size of extracted randomness; optimising over $\sigma_E$ yields the claim. 
\end{proof}

The result of Theorem~\ref{thm:privacy_amp} implies all prior iterations of the leftover hash lemma that employed conventionally smoothed divergences, with the only caveat being that our result requires \univ-universal hash functions while some of the previous results such as~\cite{tomamichel_2011} apply under the weaker assumption of (almost) 2-universality.
The prior results were typically stated in terms of $\Hmin{\ve}{P}(X|E)_\rho$~\cite{tomamichel_2011} or a more recent variant known as partially smoothed min-entropy~\cite{anshu_2020}. One immediate way to see that our formulation improves on those statements is to note that
\begin{equation}\begin{aligned}\label{eq:relax_D2}
  \Ds{2}{\ve}{M}(\rho\|\sigma) &\leq \widetilde{D}_2\s{\ve}{T} (\rho \| \sigma) \leq \widetilde{D}_2\s{\ve}{P} (\rho \| \sigma),\\
  \Dmax{\ve}{M}(\rho\|\sigma) &\leq \Dmax{\ve}{T} (\rho \| \sigma) \leq \Dmax{\ve}{P} (\rho \| \sigma).
\end{aligned}\end{equation}
Here, the leftmost inequalities stem from the fact that any (sub-normalised) state $\rho'$ feasible for the conventionally smoothed $\wt{D}^{\ve,T}_{\alpha}$ gives a feasible solution for the measured variant. Specifically, $\norm{\M(\rho) - \M(\rho')}{+} \leq \norm{\rho - \rho'}{+} \leq \ve$ for any measurement channel $\M$, and hence $\Ds{\alpha}{\ve}{T}(\M(\rho)\|\M(\sigma)) \leq D_\alpha(\M(\rho')\|\M(\sigma)) \leq \wt{D}_\alpha(\rho'\|\sigma)$ which gives the claim upon maximising over $\M$. 
The rightmost inequalities then follow by one of the Fuchs--van de Graaf inequalities~\cite{fuchs_1999,tomamichel_2016}.

To show an even stronger improvement over previous results, we can refine the simple relaxation of~\eqref{eq:relax_D2} to give a tighter bound that connects the achievability result of Theorem~\ref{thm:privacy_amp} with purified distance smooth min-entropy. (See also~\cite{tan_2026} for further improvements in such inequalities.) 
\begin{boxed}
\begin{corollary}\label{cor:achiev_purified}
For all $\ve \in (0,1)$ and any \univ-universal family of hash functions, it holds that
\begin{equation}\begin{aligned}
   \EE_f \frac12 \norm{ \rho_{ZE}^f - \frac{\id_Z}{|\Z|} \otimes \rho_E }{1} \leq \ve + \left[ \frac{|\Z|-1}{4} \, \exp\left(-\Hmin{\sqrt\ve}{P} (X|E)_\rho \right) \right]^{1/3}.
\end{aligned}\end{equation}
Accordingly, for all $\ve \in (0,1)$ and all $\mu \in (0,\ve)$, there exists a randomness extraction protocol such that
\begin{equation}\begin{aligned}\label{eq:achiev_HminP}
  \ell_\ve(\rho_{XE}) &\geq \Hmin{\sqrt{\ve - \mu}}{P} (X | E)_\rho - \log\frac{1}{4\mu^3}. 
  \end{aligned}\end{equation}
\end{corollary}
\end{boxed}
\begin{proof}
Lemma~10 of~\cite{regula_2025} tells us that, for all $\delta \in (0,1-\ve)$,
\begin{equation}\begin{aligned}
  \Dmax{\ve+\delta}{M} (\rho \| \sigma) \leq \Dmax{\sqrt\ve}{P}(\rho\|\sigma) + \log \frac{(\ve+\delta)(1-\ve-\delta)}{\delta}.
\end{aligned}\end{equation}
Using this, we obtain from Theorem~\ref{thm:privacy_amp} that
\begin{equation}\begin{aligned}
  \EE_f \frac12 \norm{ \rho_{ZE}^f - \frac{\id_Z}{|\Z|} \otimes \rho_E }{1} \leq  \ve + \delta + \sqrt{\frac{|\Z|-1}{4} \, \frac{(\ve+\delta)(1-\ve-\delta)^2}{\delta} \,\Qs{\max}{\sqrt\ve}{P}(\rho_{XE} \| \id_X \otimes \sigma_E)}
\end{aligned}\end{equation}
holds for any such $\ve$ and $\delta$. 
Denoting $k \coloneqq \frac{|\Z|-1}{4} \Qs{\max}{\sqrt\ve}{P}(\rho_{XE} \| \id_X \otimes \sigma_E)$, the choice of $\delta = \ve \frac{k^{1/3}}{1- k^{1/3}}$ gives
\begin{equation}\begin{aligned}
  \EE_f \frac12 \norm{ \rho_{ZE}^f - \frac{\id_Z}{|\Z|} \otimes \rho_E }{1} \leq \ve + \left[ \frac{|\Z|-1}{4} \, \,\Qs{\max}{\sqrt\ve}{P}(\rho_{XE} \| \id_X \otimes \sigma_E) \right]^{1/3},
\end{aligned}\end{equation}
which yields the statement upon minimisation over $\sigma_E$.
\end{proof}

We will see in Section~\ref{sec:converse} that a closely matching converse bound to our achievability results can be obtained, and in particular that the scaling of the bounds obtained here is asymptotically optimal.

This tighter achievability result suggests that it is our measured smooth approach that defines the most appropriate notion of smoothing to be used in the analysis of privacy amplification. In particular, the natural role of smooth min-entropy is taken here by $\Hmin{\ve}{M}(X|E)_\rho$. Classically, this reduces to the standard smooth min-entropy $\Hmin{\ve}{T}(X|Y)_p$, but in the quantum case it constitutes a disparate type of smoothing that in general differs from previously considered definitions.

We discuss in Appendix~\ref{app:guess_prob} how the quantity $\Hmin{\ve}{M}$ can also be considered to be a generalisation of the notion of guessing probability --- an important operational role of the (unsmoothed) min-entropy $H_{\min}^{\smash\uparrow}$~\cite{konig_2009} --- thus further strengthening the interpretation of $\Hmin{\ve}{M}$ as an operational quantifier of randomness. 

A property of the measured smooth collision divergence $\Ds{2}{\ve}{M}$ that found use in the proof of Theorem~\ref{thm:privacy_amp} is the fact that the smoothing is performed over operators $R_{XE} \leq \rho_{XE}$, which in particular satisfy that $R_E \leq \rho_E$. Such a property is not always satisfied in conventional smoothing variants, which motivated the introduction of the concept of partial smoothing~\cite{hayashi_2016-2,anshu_2020} that proposed explicitly enforcing this assumption on the marginal system in addition to standard smoothing constraints. We see that this is not needed in our approach to smoothing, mirroring how partial smoothing is superfluous in the case of classical trace distance smoothing~\cite{yang_2019,abdelhadi_2020}.

We note also that both $\Hmin{\ve}{M}(X|E)_\rho$ and the tighter bound in terms of $\Hs{2}{\ve}{M}(X|E)_\rho$ can be computed as semidefinite programs; we discuss this in more detail in Appendix~\ref{app:sdp} for completeness.

\begin{remark}
In the context of security proofs in quantum key distribution, the leftover hash lemma is often applied to subnormalised states $\rho_{XE}$, i.e.\ ones satisfying $\Tr \rho_{XE} \leq 1$, due to the probabilistic nature of the involved protocols (see e.g.~\cite{tomamichel_2017-1,tupkary_2025} for discussions). All statements of this section, and in particular the achievability results in Theorem~\ref{thm:privacy_amp} and Corollary~\ref{cor:achiev_purified}, remain valid in such cases. This is particularly easy to see in the case of the smoothing based on trace distance that was used in Theorem~\ref{thm:privacy_amp}, as our definitions and derivations there do not depend on the normalisation of the involved states whatsoever. In the proof of Corollary~\ref{cor:achiev_purified}, we additionally relied on Lemma~10 of~\cite{regula_2025}, which is stated there only for normalised states; it can however be verified that the proof applies in the same way to subnormalised operators.%
\footnote{Explicitly, the generalised fidelity and the purified distance between subnormalised states $\rho$ and $\rho'$ are defined as the corresponding measures evaluated between the normalised extensions $\hat\rho \coloneqq \rho \oplus [1\!-\!\Tr\rho]$ and $\hat\rho' \coloneqq \rho' \oplus [1\!-\!\Tr \rho']$~\cite{tomamichel_2016}. The steps of the proof of~\cite[Lemma~10]{regula_2025} then establish that $F(\hat\rho,\hat\rho')\geq 1-\ve$ implies that $\Tr\!\smash{\big(}\hat\rho - \frac{(\ve+\delta)(1-\ve-\delta)}{\delta} \hat\rho'\smash{\big)}_+ \leq \ve + \delta$. Due to the block-diagonal structure of the operators, $\Tr (\hat\rho - k \hat\rho')_+ = \Tr(\rho - k \rho')_+ + \left(1 - \Tr \rho - k (1 - \Tr \rho')\right)_+ \geq \Tr(\rho - k \rho')_+$, and hence also $\Tr\!\smash{\big(}\rho - \frac{(\ve+\delta)(1-\ve-\delta)}{\delta} \rho'\smash{\big)}_+ \leq \ve + \delta$. The rest of the proof then applies verbatim.}%
\end{remark}

\subsection{From smoothing to R\'enyi divergences and error exponents}\label{sec:renyi_bounds}

The analysis of the many-copy and asymptotic properties of quantum privacy amplification often relies on a reduction to the i.i.d.\ case~\cite{renner_2005,dupuis_2020,portmann_2022,arqand_2025}, where one is provided a large supply of copies of a given quantum state, $\rho_{XE}^{\otimes n}$, and studies the rate at which randomness can be extracted, i.e.\ the best $R$ such that $\ell_{\ve_n}(\rho_{XE}^{\otimes n}) \sim n R$. The highest such rate that can be achieved with arbitrarily small error $\ve_n$ equals the conditional entropy $H(X|E)_\rho$~\cite{renner_2005,renner_2005-1}, which, however, is only achievable when the error $\ve_n$ is constant or decays sufficiently slowly. 
A more detailed understanding of asymptotic privacy amplification is based on the study of the trade-offs between achievable rates and the error exponents, i.e.\ the rates of exponential decrease of the error $\ve_n \sim \exp(-n E)$. As is typical in large deviation analysis, R\'enyi entropies play a fundamental role here.

In the classical case, the tightest known general bounds on the error exponent of privacy amplification were obtained by Hayashi~\cite{hayashi_2013,hayashi_2016-2}. This was done through an approach based on smooth entropies, showing their usefulness also in this regime. However, despite some results in this direction~\cite{hayashi_2014,hayashi_2015}, previous works were not able to extend these techniques to give matching results for quantum side information. Dupuis~\cite{dupuis_2023} later used a different approach based on interpolation of Schatten norms to give the improved achievability result
\begin{equation}\begin{aligned}\label{eq:fred}
    \EE_f \frac12 \norm{ \rho_{ZE}^f - \frac{\id_Z}{|\Z|} \otimes \rho_E }{1} \leq 4^{- \frac{\alpha-1}{\alpha}} \exp\left( - \frac{\alpha-1}{\alpha} \left( \wt{H}_{\alpha}^\uparrow (X|E)_\rho - \log |\Z|\right) \right)
\end{aligned}\end{equation}
for all $\alpha \in (1,2]$. 
By the additivity of sandwiched R\'enyi conditional entropies~\cite{beigi_2013}, the bound constrains in particular the asymptotic error exponent, while its one-shot character means that it can be readily applied also in the finite-copy regime. The bound is the tightest known such result for both classical and quantum states.

We will show that this result is implied by our tightened leftover hash lemma, and in fact the corresponding one-shot bound can be expressed in terms of measured R\'enyi entropies rather than sandwiched ones.

\begin{boxed}
\begin{theorem}\label{thm:renyi_achiev}
For all $\alpha \in (1,2]$, under any \univ-universal family of hash functions  $f \colon \X \to \Z$ it holds that
\begin{equation}\begin{aligned}\label{eq:renyi_achiev_oneshot}
    \EE_f \frac12 \norm{ \rho_{ZE}^f - \frac{\id_Z}{|\Z|} \otimes \rho_E }{1} &\leq 4^{- \frac{\alpha-1}{\alpha}} \, \exp\left(- 
    \frac{\alpha-1}{\alpha} \left( H_\alpha^{\MM,\uparrow} (X|E)_\rho - \log |\Z| \right) \right)\\
    &\leq \exp\left(- 
    \frac{\alpha-1}{\alpha} \left( H_\alpha^{\MM,\uparrow} (X|E)_\rho - \log |\Z| \right) \right).
\end{aligned}\end{equation}

Consequently, for any given rate $R < H(X|E)_\rho$, there exists a sequence of randomness extraction protocols such that $\ell_{\ve_n}(\rho_{XE}^{\otimes n}) = n R$ for all $n$ and
\begin{equation}\begin{aligned}\label{eq:renyi_achiev_asymp}
  \liminf_{n\to\infty} - \frac1n \log(\ve_n) \geq \sup_{\alpha \in (1,2]} \frac{\alpha-1}{\alpha} \left( \wt{H}_{\alpha}^\uparrow (X|E)_\rho - R \right).
\end{aligned}\end{equation}
\end{theorem}
\end{boxed}
Compared to the result of Dupuis~\cite{dupuis_2023} recalled in Eq.~\eqref{eq:fred}, our one-shot bound in~\eqref{eq:renyi_achiev_oneshot} is tighter at the level of the entropic quantities used, since $D_\alpha^{\MM}(\rho\|\sigma) \leq \wt{D}_\alpha(\rho\|\sigma)$ with equality for $\alpha \in (1,2]$ if and only if $\rho$ and $\sigma$ commute~\cite{berta_2017-1}. Our proof here avoids complex interpolation techniques and instead directly follows from our smooth entropy approach.
The assumption of \univ-universality is the same as in~\cite{dupuis_2023}.

\begin{proof}[Proof of Theorem~\ref{thm:renyi_achiev}]
By the smooth leftover hash lemma of Theorem~\ref{thm:privacy_amp}, we have
\begin{equation}\begin{aligned}
  \EE_f \frac12 \norm{ \rho_{ZE}^f - \frac{\id_Z}{|\Z|} \otimes \rho_E }{1} \leq \ve + \frac12 \sqrt{|\Z| \, \Qs{2}{\ve}{M}(\rho_{XE} \| \id_X \otimes \sigma_E)}
\end{aligned}\end{equation}
for any choice of $\sigma_E$ and $\ve \in [0,1]$. 
Applying Lemma~\ref{lem:D2_bound_with_h2} then leads to
\begin{align}
    &\EE_f \frac12 \norm{ \rho_{ZE}^f - \frac{\id_Z}{|\Z|} \otimes \rho_E }{1} \nonumber\\
    &\leq \ve + \frac{\alpha-1}{\alpha}  \sqrt{ |\Z| \, Q_\alpha^{\MM} (\rho_{XE} \| \id_X \otimes \sigma_E)^{\frac{1}{\alpha-1}} \, \ve^{- \frac{2-\alpha}{\alpha-1}} \, \left(\tfrac{2-\alpha}{\alpha}\right)^{\frac{2-\alpha}{\alpha-1}}  }\\
    &= \ve + \frac12 \left(1 - \frac{2-\alpha}{\alpha}\right) \exp\left[ \frac12 \log |\Z| + \frac12 D_\alpha^{\MM} (\rho_{XE} \| \id_X \otimes \sigma_E) + \frac{2-\alpha}{2(\alpha-1)} \left( \log \frac{1}{\ve} + \log \frac{2-\alpha}{\alpha}\right) \right].\nonumber
\end{align}
Notice now that, for any $\nu$, the choice of
\begin{equation}\begin{aligned}
  \ve = \nu\, \frac{2-\alpha}{\alpha} \exp\left[ \frac{\alpha-1}{\alpha} \left( \log |\Z| + D_\alpha^{\MM}(\rho_{XE} \| \id_X \otimes \sigma_E) \right) \right]
\end{aligned}\end{equation}
(or the trivial choice $\ve = 1$ in case the above value exceeds 1) 
leads to
\begin{equation}\begin{aligned}\label{eq:renyi_bound_nu}
  &\EE_f \frac12 \norm{ \rho_{ZE}^f - \frac{\id_Z}{|\Z|} \otimes \rho_E }{1} \\
  &\leq \left( \nu\, \frac{2-\alpha}{\alpha} + \left(1 - \frac{2-\alpha}{\alpha}\right) \nu^{- \frac{2-\alpha}{2(\alpha-1)}} \, 2^{-1} \right)  \,\exp\left[ \frac{\alpha-1}{\alpha} \left( \log |\Z| + D_\alpha^{\MM}(\rho_{XE} \| \id_X \otimes \sigma_E) \right) \right].
\end{aligned}\end{equation}
Pick then
\begin{equation}\begin{aligned}
  \nu = 2^{\frac{2-\alpha}\alpha - 1} = 2^{-2 \frac{\alpha-1}\alpha}
\end{aligned}\end{equation}
to minimise the prefactor in~\eqref{eq:renyi_bound_nu}. 
Now optimise the above over the choice of $\sigma_E$, and enjoy.

The asymptotic statement in Eq.~\eqref{eq:renyi_achiev_asymp} follows by choosing $|\Z| = \exp(nR)$ and using the fact that the measured R\'enyi conditional entropies regularise to the sandwiched R\'enyi ones~\cite{hayashi_2016-1}, that is,
\begin{equation}\begin{aligned}
  \lim_{n\to\infty} \frac1n H_\alpha^{\MM,\uparrow} (X^n|E^n)_{\rho^{\otimes n}} = \wt{H}^\uparrow_{\alpha}(X|E)_\rho
\end{aligned}\end{equation}
for all $\alpha \in [1/2,\infty]$.
\end{proof}

\begin{remark}
In the proof of Theorem~\ref{thm:renyi_achiev}, if we were to apply the bound $ \Ds{2}{\ve}{M}(\rho\|\sigma) \leq D_\alpha^{\MM} (\rho\|\sigma) + \frac{2-\alpha}{\alpha-1} \log \frac{1}{\ve}$ from Lemma~\ref{lem:D2_upper_exponent_bound} instead of the tighter Lemma~\ref{lem:D2_bound_with_h2}, the resulting asymptotic bound would be the same, the only difference being that the prefactor $4^{- \frac{\alpha-1}{\alpha}} \leq 1$ would be replaced by $\frac32$. In the classical case, this would recover exactly the bound shown by Hayashi~\cite{hayashi_2013}.
\end{remark}

\subsection{Fully quantum extension to decoupling}\label{sec:decoupling}

The mechanism behind privacy amplification can be generalised to the case where, instead of a classical--quantum state, one aims to coherently extract randomness from a bipartite quantum state $\rho_{AE}$. The corresponding task is known as \emph{decoupling}, the goal of which is to find a bipartition $A=A_1 A_2$ such that the subsystem $A_1$ is close to maximally mixed and independent of the reference system $E$. Decoupling forms an important primitive underlying problems such as quantum state merging and coding theorems in quantum Shannon theory~\cite{horodecki_2005-1,abeyesinghe_2009,dupuis_2010}, and its one-shot formulations have formed the foundation of the finite-blocklength study of many tasks in quantum information~\cite{berta_2009,berta_2011,dupuis_2010,dupuis_2014}.

Whereas the leftover hash lemma for privacy amplification randomises a classical register using a 2-universal hash family, decoupling randomises a quantum register using a random ensemble of unitaries whose second moment match those of the Haar measure, i.e.\ a unitary $2$-design. 
Formally, we say that a probability distribution $\nu$ on the set of all unitary operators of a system $A$ is a \emph{unitary $2$-design} if, for every operator $M$ on $A^{\otimes 2}$,
\begin{equation}
  \EE_{U_A\sim \nu}\!\left[ U_A^{\otimes 2} M (U_A^\dagger)^{\otimes 2} \right]
  \;=\;
  \EE_{U_A\sim \mu_{\mathrm{Haar}}}\!\left[ U_A^{\otimes 2} M (U_A^\dagger)^{\otimes 2} \right],
\end{equation}
where $\mu_{\mathrm{Haar}}$ denotes the Haar measure on $A$, and we will henceforth omit $\nu$ from the notation. 
The analyses of randomness extraction and decoupling both rely on the same second-moment principle that turns trace distance bounds into collision-type quantities. Our contribution here will be to generalise the approach we used to improve on the leftover hash lemma (Lemma~\ref{lem:leftover} and Theorem~\ref{thm:privacy_amp}) and show that, by working with the Bures norm $\norm{\cdot}{\Bures\sigma}$ rather than standard quantum collision entropies, we will obtain a result in terms of the measured R\'enyi divergence of order~2. Applying our smooth entropy machinery will then give a strict improvement over prior formulations of one-shot decoupling with smooth entropies~\cite{dupuis_2014,shen_2023}, recovering both the error exponent shown in~\cite{dupuis_2023} and the second-order achievability result of~\cite{shen_2023}.

Below we will use $\U_A$ to denote the unitary channel corresponding to a given unitary $U_A$, acting as $\U_A(X_{AE}) = (U_A \otimes \id_E) X_{AE} (U^\dagger_A \otimes \id_E)$.

\begin{boxed}[filled]
\begin{theorem}[Tightened one-shot decoupling]\label{thm:decoupling_smollision}
Let $A = A_1 A_2$ be a bipartite quantum system of dimension $|A| = |A_1| |A_2|$.
For any Hermitian operator $R_{AE}$, any state $\sigma_E$, and any unitary 2-design on $A$, it holds that
\begin{equation}\label{eq:decoupling_bound_Bures}
  \EE_{U_A}\;
  \norm{\Tr_{A_2}\!\big[\U_A(R_{AE})\big]
  \;-\; \frac{\id_{A_1}}{|A_1|} \otimes R_E}{+}  \leq
  \frac12\,
  \sqrt{\frac{|A|\,\big(|A_1|^2-1\big)}{|A|^2-1}}\;
  \norm{R_{AE} - \frac{\id_{A}}{|A|}\otimes R_E}{\Bures\id_A\otimes \sigma_E}.
\end{equation}

As a consequence, for any state $\rho_{AE}$ and any $\ve \in [0,1)$, we have
\begin{equation}\begin{aligned}\label{eq:decoupling_smoothed}
    \EE_{U_A}\; \frac12 \norm{\Tr_{A_2}\!\big[\U_A(\rho_{AE})\big] - \frac{\id_{A_1}}{|A_1|} \otimes \rho_E}{1} 
&\leq  \ve + \frac12 \sqrt{ \frac{|A_1|}{|A_2|} \, \Qs{2}{\ve}{M} (\rho_{AE} \| \id_A \otimes \sigma_E)  }\\
&\leq  \ve + \frac12 \sqrt{ \frac{|A_1|}{|A_2|} \, (1-\ve) \, \Qmax{\ve}{M} (\rho_{AE} \| \id_A \otimes \sigma_E) }.
\end{aligned}\end{equation}
\end{theorem}
\end{boxed}
Optimising over all $\sigma_E$ then gives a bound in terms of $\Hmin{\ve}{M}(A|E)_\rho$. 
Put another way, the decoupling distance is guaranteed to be at most $\ve$ whenever
\begin{equation}\begin{aligned}
  \log |A_1| - \log |A_2| &\leq \Hmin{\ve-\mu}{M}(A|E)_\rho - \log \frac{1-\ve+\mu}{4\mu^2}
\end{aligned}\end{equation}
or when
\begin{equation}\begin{aligned}
  \log |A_1| - \log |A_2| &\leq \Hmin{\sqrt{\ve-\mu}}{P}(A|E)_\rho - \log \frac{1}{4\mu^3}
\end{aligned}\end{equation}
as in Corollary~\ref{cor:achiev_purified}.

The proof will use two standard technical ingredients of decoupling arguments, for which we refer e.g.\ to~\cite[Lemmas~3.4 and~3.5]{dupuis_2014}.

\begin{lemma}[Swap trick] \label{lem:swaptrick}
Denote by $F_{AA'}$ the swap operator on $A \otimes A'$, i.e.\ $F_{AA'}\!\ket{ij} =\ket{ji}$.
Then for any Hermitian operators $M,N$ on $A$, $\Tr(MN)=\Tr\!\left((M\otimes N)\,F_{AA'}\right)$.
\end{lemma}

\begin{lemma}[Unitary twirl] \label{lem:twirl}
Consider any Hermitian operator $K_{AA'EE'}$ with $A' \cong A$, $E' \cong E$ and let $F_{AA'}$ be the swap operator on $AA'$.
For the Haar measure on $A$, and hence for any unitary $2$-design, one has
\begin{equation}\label{eq:twirl_id_swap}
  \EE_{U_A}\!\left((U_A^\dagger)^{\otimes 2}\,K\,U_A^{\otimes 2}\right)
  =
  \id_{AA'}\otimes \alpha_{EE'} \;+\; F_{AA'}\otimes \beta_{EE'},
\end{equation}
where
\begin{equation}\label{eq:alpha_beta_op}
  \alpha_{EE'} = \frac{ \Tr_{AA'}[K] - \frac{1}{|A|}\Tr_{AA'}[F_{AA'}K]}{|A|^2-1},
  \qquad
  \beta_{EE'}  = \frac{ |A|\Tr_{AA'}[F_{AA'}K] - \Tr_{AA'}[K]}{|A|\big(|A|^2-1\big)}.
\end{equation}
\end{lemma}

\begin{proof}[Proof of Theorem~\ref{thm:decoupling_smollision}]
Denote
\begin{equation}
   \pi_A \coloneqq \frac{\id_A}{|A|}, \qquad Z_{AE} \coloneqq R_{AE} - \pi_A \otimes R_E.
\end{equation}
Since $\Tr_{A_2}[(U_A\otimes\id_E)(\pi_A\otimes R_E)(U_A^\dagger\otimes\id_E)]
= \pi_{A_1}\otimes R_E$ for all unitaries $U_A$, 
we can write
\begin{equation}\begin{aligned}
  \Tr_{A_2}\!\big[\U_A(R_{AE}))\big] - \pi_{A_1} \otimes R_E &= \Tr_{A_2}\!\big[(U_A\otimes\id_E)\,Z_{AE}\,(U_A^\dagger\otimes\id_E)\big]\\
  &\eqqcolon X_{A_1E}^U.
\end{aligned}\end{equation}
Noting that $\big\|X_{A_1E}^U\big\|_+ = \frac12 \big\|X_{A_1E}^U\big\|_1$ as $X_{A_1E}^U$ has trace zero, an application of Lemma~\ref{lem:variance} with $\sigma = \id_{A_1}\otimes\sigma_E$ gives
\begin{equation}\label{eq:decoup_temme}
  \norm{X_{A_1E}^U}{+}
  \le
  \frac12
  \norm{X_{A_1E}^U}{\Bures\id_{A_1}\otimes\sigma_E}
  \sqrt{\Tr(\id_{A_1}\otimes\sigma_E)}
  =
  \frac12 \sqrt{|A_1|}
  \norm{X_{A_1E}^U}{\Bures\id_{A_1}\otimes\sigma_E}.
\end{equation}
Taking expectation over $U_A$ and using Jensen's inequality, we get
\begin{equation}\label{eq:decoup_jensen}
  \EE_{U_A} \norm{X_{A_1E}^U}{+}
  \le
  \frac12 \sqrt{ |A_1| \,
  \EE_{U_A} \norm{X_{A_1E}^U}{\Bures\id_{A_1}\otimes\sigma_E}^2 } .
\end{equation}

We then move on to computing the quadratic term. By definition of the Bures seminorm,
\begin{equation}\label{eq:decoup_quad_def}
  \norm{X_{A_1E}^U}{\Bures\id_{A_1}\otimes\sigma_E}^2
  =
  \Tr\!\left(
    X_{A_1E}^U\,
    \J^{-1}_{\id_{A_1}\otimes\sigma_E}(X_{A_1E}^U)
  \right).
\end{equation}
We will use two properties of the superoperator $\J^{-1}_{\id_{A_1}\otimes\sigma_E}$: first, as already used in our previous derivations, since $\J_{\id_{A_1}\otimes\sigma_E} = \idc_{A_1}\otimes \J_{\sigma_E}$, we have $\J^{-1}_{\id_{A_1}\otimes\sigma_E}=\idc_{A_1}\otimes \J^{-1}_{\sigma_E}$ on $\ran(\J_{\sigma_E})$; second,  $\J_{\sigma_E}$ is self-adjoint with respect to the Hilbert--Schmidt inner product, and hence so is $\J^{-1}_{\sigma_E}$.
Using the swap trick (Lemma~\ref{lem:swaptrick}) we can then rewrite~\eqref{eq:decoup_quad_def} as
\begin{equation}\begin{aligned}\label{eq:decoup_swap}
    \norm{X_{A_1E}^U}{\Bures\id_{A_1}\otimes\sigma_E}^2
  &=
  \Tr\!\left(
    X_{A_1E}^U\,
    (\idc_{A_1}\otimes \J^{-1}_{\sigma_E})(X_{A_1E}^U)
  \right)\\
  &=
  \Tr\!\left(
    \big(X_{A_1E}^U \otimes (\idc_{A_1'}\otimes \J^{-1}_{\sigma_E})(X_{A_1'E'}^U)\big)\,
    (F_{A_1A_1'}\otimes F_{EE'})
  \right)\\
  &=
  \Tr\!\left(
    (X_{A_1E}^U \otimes X_{A_1'E'}^U)\,
    (F_{A_1A_1'}\otimes (\idc_E\otimes \J^{-1}_{\sigma_E})(F_{EE'}))
  \right)\\
  &= \Tr\!\left((Z_{AE}^U\otimes Z_{A'E'}^U)\,K_{AA'EE'}\right) ,
\end{aligned}\end{equation}
where we defined $Z_{AE}^U \coloneqq (U_A\otimes\id_E)\,Z_{AE}\,(U_A^\dagger\otimes\id_E)$ so that $X_{A_1E}^U=\Tr_{A_2}[Z_{AE}^U]$ and introduced 
\begin{equation}\label{eq:def_K}
  K_{AA'EE'} \coloneqq (F_{A_1A_1'}\otimes \id_{A_2A_2'}) \otimes J_{EE'}, \qquad J_{EE'} \coloneqq (\idc_E \otimes \J^{-1}_{\sigma_E})(F_{EE'}).
\end{equation}

Taking expectation over $U$ and using cyclicity of the trace gives
\begin{equation}\label{eq:twirl_target}
  \EE_{U_A}\, \norm{X_{A_1E}^U}{\Bures\id_{A_1}\otimes\sigma_E}^2
  =
  \Tr\!\left(
    (Z_{AE}\otimes Z_{A'E'})\,
    \EE_{U_A}\!\left((U_A^\dagger)^{\otimes 2}\,K_{AA'EE'}\,U_A^{\otimes 2}\right)
  \right).
\end{equation}
Lemma~\ref{lem:twirl} then tells us that $\EE_{U_A}\!\left((U_A^\dagger)^{\otimes 2}\,K_{AA'EE'}\,U_A^{\otimes 2}\right) = \id_{AA'}\otimes \alpha_{EE'} + F_{AA'}\otimes \beta_{EE'}$ with coefficients as in~\eqref{eq:alpha_beta_op}. 
For the first term, observe that $\Tr\!\left((Z_{AE}\otimes Z_{A'E'})(\id_{AA'}\otimes \alpha_{EE'})\right)
=
\Tr\!\left(( Z_E \otimes Z_{E'})\,\alpha_{EE'}\right)$, and since $Z_E=0$, the whole term vanishes.
We now compute $\beta_{EE'}$ for the specific $K_{AA'EE'}$ in~\eqref{eq:def_K}.
Using $F_{AA'} = F_{A_1A_1'}\otimes F_{A_2A_2'}$ and $\Tr F_{A_iA_i'}=|A_i|$, we obtain
\begin{align}
  \Tr_{AA'}[K]
  &= \Tr(F_{A_1A_1'})\,\Tr(\id_{A_2A_2'})\, J_{EE'}
  = |A_1|\, |A_2|^2 \, J_{EE'}, \\
  \Tr_{AA'}[F_{AA'}K]
  &= \Tr(\id_{A_1A_1'})\,\Tr(F_{A_2A_2'})\, J_{EE'}
  = |A_1|^2\, |A_2| \, J_{EE'}.
\end{align}
Substituting into~\eqref{eq:alpha_beta_op} gives
\begin{equation}\label{eq:beta_explicit}
  \beta_{EE'}
  =
  \frac{ |A|\Tr_{AA'}[F_{AA'}K] - \Tr_{AA'}[K]}{|A|\big(|A|^2-1\big)}
  =
  \frac{|A_2|\big(|A_1|^2-1\big)}{|A|^2-1}\;J_{EE'}.
\end{equation}
Hence
\begin{equation}\begin{aligned}\label{eq:decoup_expect_Bures2}
  \EE_{U_A}\, \norm{X_{A_1E}^U}{\Bures\id_{A_1}\otimes\sigma_E}^2
  &=
  \frac{|A_2|\big(|A_1|^2-1\big)}{|A|^2-1}\;
  \Tr\!\big((Z_{AE}\otimes Z_{A'E'})(F_{AA'}\otimes J_{EE'})\big)\\
  &= \frac{|A_2|\big(|A_1|^2-1\big)}{|A|^2-1}\; \norm{Z_{AE}}{\Bures\id_A\otimes\sigma_E}^2
\end{aligned}\end{equation}
where the second line uses the same rewriting as in~\eqref{eq:decoup_swap} but for $Z_{AE}$ (with $A$ instead of $A_1$). 
Combining this with~\eqref{eq:decoup_jensen} gives
\begin{equation}\begin{aligned}
  \EE_{U_A} \norm{X_{A_1E}^U}{+} &\leq \frac12 \sqrt{ |A_1| \,
  \EE_{U_A} \norm{X_{A_1E}^U}{\Bures\id_{A_1}\otimes\sigma_E}^2 }
  &\leq
  \frac12
  \sqrt{\frac{|A|\big(|A_1|^2-1\big)}{|A|^2-1}}\;
  \norm{Z_{AE}}{\Bures\id_A\otimes\sigma_E}
\end{aligned}\end{equation}
completing the proof of Eq.~\eqref{eq:decoupling_bound_Bures}.

The adaptation of this to the smoothed statement in~\eqref{eq:decoupling_smoothed} proceeds in analogy to the proof of Theorem~\ref{thm:privacy_amp}. Namely, pick an optimal $R_{AE}$ for the variational form of $\Qs{2}{\ve}{M}(\rho_{AE}\|\id_A \otimes \sigma_E)$ in Theorem~\ref{thm:D2_smoothed_forms} and consider that
\begin{align}
  \EE_{U_A}\, \norm{\vphantom{\big|}\Tr_{A_2}[\U_A(\rho_{AE})] - \pi_{A_1}\otimes \rho_E}{+}
  &\textleq{(i)}
  \EE_{U_A}\, \norm{\vphantom{\big|}\Tr_{A_2}[\U_A(\rho_{AE})] - \pi_{A_1}\otimes R_E}{+} \nonumber\\
  &\leq
  \EE_{U_A}\, \norm{\vphantom{\big|}\Tr_{A_2}[\U_A(\rho_{AE}-R_{AE})]}{+}
  + \EE_{U_A}\, \norm{\vphantom{\big|}\Tr_{A_2}[\U_A(R_{AE})] - \pi_{A_1}\otimes R_E}{+} \nonumber\\
  &\textleq{(ii)}
  \norm{\rho_{AE}-R_{AE}}{+}
  + \frac12 \sqrt{\frac{|A|\big(|A_1|^2-1\big)}{|A|^2-1}}\;
  \norm{R_{AE} - \pi_A \otimes R_E}{\Bures\id_A\otimes\sigma_E} \nonumber\\
  &\textleq{(iii)}
  \norm{\rho_{AE}-R_{AE}}{+}
  + \frac12 \sqrt{\frac{|A|\big(|A_1|^2-1\big)}{|A|^2-1}}\;
  \norm{R_{AE}}{\Bures\id_A\otimes\sigma_E}\\
  &\textleq{(iv)}
  \ve + \frac12 \sqrt{\frac{|A_1|}{|A_2|} \, \Qs{2}{\ve}{M}(\rho_{AE} \| \id_A\otimes \sigma_E)}. \nonumber
\end{align}
where: in (i) we used that $R_{AE} \leq \rho_{AE}$; (ii) is a consequence of data processing for trace distance and the decoupling result in~\eqref{eq:decoupling_bound_Bures} that we just proved; in (iii) we observed that 
\begin{equation}\begin{aligned}
  \norm{R_{AE} - \pi_A \otimes R_E}{\Bures\id_A\otimes\sigma_E}^2 =  \norm{R_{AE}}{\Bures\id_A\otimes\sigma_E}^2 - \frac{1}{|A|} \norm{R_E}{\Bures\sigma_E}^2 \leq \norm{R_{AE}}{\Bures\id_A\otimes\sigma_E}^2
\end{aligned}\end{equation}
as easily seen by expanding the Bures inner product; 
and finally in (iv) we used the inequality $\frac{x-1}{y-1} \leq \frac xy$ valid for all $y \geq x > 1$. 
The second inequality in~\eqref{eq:decoupling_smoothed} then follows by Lemma~\ref{lem:secondorder_equiv}.
\end{proof}

In complete analogy with our result for privacy amplification in Theorem~\ref{thm:renyi_achiev}, a consequence of Theorem~\ref{thm:decoupling_smollision} is that we obtain a one-shot achievability bound in terms of R\'enyi conditional entropies, taking a form analogous to the prior finding of~\cite{dupuis_2023} but with measured R\'enyi divergences instead of sandwiched ones.
\begin{boxed}
\begin{corollary}
For all $\alpha \in (1,2]$ and any unitary 2-design on $A = A_1 A_2$, it holds that
\begin{equation}\begin{aligned}
     \EE_{U_A}\; \frac12 \norm{\Tr_{A_2}\!\big[\U_A(\rho_{AE})\big] - \frac{\id_{A_1}}{|A_1|} \otimes \rho_E}{1}  \leq 4^{- \frac{\alpha-1}{\alpha}} \, \exp\left(- \frac{\alpha-1}{\alpha} \left( H_\alpha^{\MM,\uparrow} (A|E)_\rho + \log |A_2| - \log |A_1| \right) \right).
\end{aligned}\end{equation}
\end{corollary}
\end{boxed}

In the remainder of the work, we return to the case of privacy amplification.

\section{Converse bounds for privacy amplification}\label{sec:converse}

\subsection{One-shot converse}\label{sec:converse-oneshot}

To certify the optimality of any achievability result, one needs to prove that no other randomness extraction protocol can perform better. 
This is a surprisingly tricky problem in quantum privacy amplification, as several bounds that one would intuitively expect to apply are not valid converse restrictions, and care needs to be taken to use the right entropic quantities. 

For classical privacy amplification --- that is, when both $X$ and $E = Y$ are classical systems --- a natural converse is given by the smooth min-entropy~\cite{renner_2005-2}. This can be expressed in our notation~as
\begin{equation}\begin{aligned}\label{eq:converse_classical}
  \ell_\ve(p_{XY}) &\leq \Hmin{\ve}{M}[down](X|Y)_p = \Hmin{\ve}{T}[down](X|Y)_p.
\end{aligned}\end{equation}
This converse has been rephrased or rediscovered in several different formulations, e.g.\ using slightly different definitions of smoothing~\cite{watanabe_2013}, expressed in terms of the hockey-stick divergence $E_\gamma$~\cite{yang_2019}, or recast as an optimisation of the hypothesis testing relative entropy~\cite{renes_2018}; it can be shown however that the statements are all equivalent~\cite{renes_2018}. The converse implies several other commonly used restrictions, including ones based on information spectrum~\cite{hayashi_2016-3,yang_2019}. 

It is then rather surprising that the converse in Eq.~\eqref{eq:converse_classical} simply does not apply in general when the classical distribution $p_{XY}$ is replaced with a quantum state $\rho_{XE}$. An explicit counterexample was given by Renes in~\cite[\sect 3.1]{renes_2018}. Such problems have complicated the establishment of converse bounds in quantum privacy amplification with trace distance.
To circumvent this, some papers employed the idea of `ensemble converses'~\cite{gallager_1973} --- converse bounds that do not apply to all hash functions, but instead make some assumptions about the allowed families of hashes, commonly assuming their strong 2-universality or even more restrictive constraints~\cite{hayashi_2013,hayashi_2016-3, shen_2022,shen_2024}. 
That is, instead of upper bounding the quantity $\ell_\ve(\rho_{XE})$ defined as in Eq.~\eqref{eq:ell_def}, such results would instead upper bound
\begin{equation}\begin{aligned}
    \ell_\ve^{\mathfrak{F}}(\rho_{XE}) \coloneqq \max \lset \log \left|\Z\right| \bar \F \in \mathfrak{F},\; \EE_{f \sim \mu_\F}\, \norm{ \rho^f_{ZE} - \frac{\id_Z}{\left|\Z\right|} \otimes \rho_E }{+} \leq \ve \rset,
\end{aligned}\end{equation}
with the optimisation being over all distributions $\mu_\F$ over families $\F$ of hash functions that belong to some chosen class $\mathfrak{F}$.
From a practical perspective, a priori restricting the hash functions in this way may not be operationally justified unless the optimality of such restricted families can be argued; 
demanding properties such as strong 2-universality could also make the hash functions much more costly to implement due to larger required seed size~\cite{stinson_1994}.

Interestingly, this issue does not affect the modified security criterion where trace distance is replaced with purified distance, for which a converse bound analogous to~\eqref{eq:converse_classical} does hold for quantum states~\cite{tomamichel_2013}. We will use this fact, together with recently improved one-shot divergence inequalities, to show that the one-shot converse bound of~\eqref{eq:converse_classical} is still `morally' valid for trace distance in the quantum case --- it applies in an approximate sense, enough to give restrictions on second-order asymptotics and error exponents that exactly match those obtained from Eq.~\eqref{eq:converse_classical} classically.

\begin{boxed}
\begin{proposition}\label{prop:converse}
For any CQ state $\rho_{XE}$, any $\ve \in (0,1)$, and any $\delta \in (0, 1-\ve)$, it holds that
\begin{equation}\begin{aligned}
  \ell_\ve(\rho_{XE}) &\leq \Hmin{\sqrt\ve}{P}[down] (X|E)_{\rho} + \log\frac{1}{1-\ve}\\
  & \leq \Hmin{\ve+\delta}{M}[down] (X|E)_{\rho} +  \log\frac{\ve+\delta}{\delta}.
\end{aligned}\end{equation}

For the choice $\delta = (k-1) \ve$ with any $k \in \left(1, 1/\ve\right)$ this gives in particular
\begin{equation}\begin{aligned}
  \ell_\ve(\rho_{XE}) \leq  \Hmin{k \ve}{M}[down] (X|E)_{\rho} + \log \frac{k}{k-1}.
\end{aligned}\end{equation}

\end{proposition}
\end{boxed}

The proof will use the following.

\begin{lemma}[{\cite[Theorem~8]{tomamichel_2013}}]\label{lem:monotonicity}
For any CQ state $\rho_{XE}$ and any function $f \colon \mathcal{X} \to \mathcal{Z}$, it holds that
\begin{equation}\begin{aligned}
  \Dmax{\ve}{P} (\rho_{XE} \| \id_X \otimes \rho_E ) \leq \Dmax{\ve}{P} (\rho^f_{ZE} \| \id_Z \otimes \rho_E ).
\end{aligned}\end{equation}
\end{lemma}
This result relies on Uhlmann's theorem for the fidelity, which means that the proof idea does not immediately extend to other types of smoothing for quantum states.

Strictly speaking, the original result in~\cite{tomamichel_2013} was shown for the quantity $\Hmin{\ve}{P}[up]$, but the same proof is easily noticed to apply to $\Dmax{\ve}{P} (\rho_{XE} \| \id_X \otimes \rho_E ) = -\Hmin{\ve}{P}[down] (X|E)_\rho$, as was explicitly done in~\cite[Proposition~4]{li_2023}.

We will also need inequalities that allow us to tightly bound the measured smooth max-relative entropy with its conventionally smoothed cousins.

\begin{lemma}[\cite{regula_2025}]\label{lem:thelemma} The following inequalities hold for all $\ve \in (0,1)$ and all $\delta \in (0,1-\ve)$:
\begin{enumerate}[label=(\roman*)]
\item \label{Dsmax-DmaxT} $\displaystyle \Dmax{\ve}{M} (\rho \| \sigma) \leq \Dmax{\ve}{T}(\rho\|\sigma)$,
\item \label{DmaxP-Dsmax} $\displaystyle \Dmax{\sqrt\ve}{P} (\rho \| \sigma) \leq \Dmax{\ve}{M} (\rho \| \sigma) + \log \frac{1}{1-\ve}$,
\item \label{Dsmax-DmaxP} $\displaystyle \Dmax{\ve+\delta}{M} (\rho \| \sigma) \leq \Dmax{\sqrt\ve}{P} + \log \frac{(\ve+\delta)(1-\ve-\delta)}{\delta}$.
\end{enumerate}
\end{lemma}
All are shown in~\cite{regula_2025}: \ref{Dsmax-DmaxT} is Lemma~9 (cf.~\eqref{eq:relax_D2}), \ref{DmaxP-Dsmax} is Corollary~8,  \ref{Dsmax-DmaxP} is Lemma~10 therein.

\begin{proof}[Proof of Proposition~\ref{prop:converse}]
The starting point is a standard converse argument: by definition, a randomness extraction protocol must ensure that $\norm{ \rho_{ZE}^f - \frac{1}{|\Z|}\id_Z \otimes \rho_E}{+} \leq \ve$, which means that $\Dmax{\ve}{T}\big(\rho_{ZE}^f \big\| \frac{1}{|\Z|} \id_Z \otimes \rho_E\big) \leq D_{\max}\big(\frac{1}{|\Z|} \id_Z \otimes \rho_E \big\| \frac{1}{|\Z|} \id_Z \otimes \rho_E\big) = 0$, i.e.\ $\Dmax{\ve}{T}\big(\rho_{ZE}^f \big\| \id_Z \otimes \rho_E\big) \leq \log \frac{1}{|\Z|}$.

Then
\begin{equation}\begin{aligned}\label{eq:tight}
 -\log |\Z| &\geq \Dmax{\ve}{T}(\rho_{ZE}^f\| \id_Z \otimes \rho_E)\\
 &\textgeq{(i)} \Dmax{\ve}{M}(\rho_{ZE}^f\| \id_Z \otimes \rho_E) \\
 &\textgeq{(ii)} \Dmax{\sqrt\ve}{P} (\rho_{ZE}^f\| \id_Z \otimes \rho_E) - \log\frac{1}{1-\ve} \\
 &\textgeq{(iii)} \Dmax{\sqrt\ve}{P} (\rho_{XE} \| \id_X \otimes \rho_E ) - \log\frac{1}{1-\ve} \\
 &\textgeq{(iv)} \Dmax{\ve+\delta}{M}(\rho_{XE} \| \id_X \otimes \rho_E ) - \log\frac{(1-\ve-\delta)(\ve+\delta)}{(1-\ve)\delta}.
\end{aligned}\end{equation}
with the inequalities following by (i) Lemma~\ref{lem:thelemma}\ref{Dsmax-DmaxT}, (ii) Lemma~\ref{lem:thelemma}\ref{DmaxP-Dsmax}, (iii) Lemma~\ref{lem:monotonicity}, and (iv) Lemma~\ref{lem:thelemma}\ref{Dsmax-DmaxP}. 
Noting that $1-\ve-\delta \leq 1-\ve$ gives the stated result.
\end{proof}

\begin{remark}
Our result can be compared with a converse bound of~\cite{shen_2024}, which is an ensemble converse that applies only to the expected error under strongly 2-universal hash functions and takes the form
\begin{equation}\begin{aligned}\label{eq:haochung_converse}
  \ell^{\text{strong 2-univ}}_\ve(\rho_{XE}) &\leq - D_s^{1-\ve-\delta} (\rho_{XE} \| \id_X \otimes \rho_E) + \log \frac{(\ve+c)(1 + c)}{c(\delta - c)},
\end{aligned}\end{equation}
where $c \in (0,\delta)$ and $D_s^\ve$ stands for the information spectrum relative entropy~\cite{tomamichel_2013}. Our bound not only applies to all hash functions, but it is in fact a strict quantitative improvement. To see this, adjusting the notation to match~\cite{shen_2024}, we can relax our bound in Proposition~\ref{prop:converse} to
\begin{equation}\begin{aligned}
  \ell_\ve(\rho_{XE}) &\leq - \Dmax{\ve+c}{M} (\rho_{XE} \| \id_X \otimes \rho_E) + \log\frac{\ve+c}{c}\\
  &\leq - D_s^{1-\ve-\delta} (\rho_{XE} \| \id_X \otimes \rho_E) + \log\frac{(\ve+c)(\ve+\delta)}{c(\delta-c)}
\end{aligned}\end{equation}
where the second line follows by~\cite[Proposition~17]{regula_2025}. Since obviously $\ve+\delta \leq 1 \leq 1+c$, this implies~\eqref{eq:haochung_converse}.
\end{remark}

\subsection{Second-order expansion}\label{sec:secondorder}

Small deviation analysis is another way to refine the characterisation of many-copy i.i.d.\ performance of a task. 
In contrast to the setting of large deviations, which assumes a fixed, $n$-independent gap from the optimal first-order rate of a protocol and asks for the error to vanish exponentially fast, we now aim to capture both the optimal first-order asymptotic rate (i.e.\ the $O(n)$ term) as well as the second-order term $O(\sqrt{n})$ when the error is a fixed constant.

Using the second-order analysis of quantum hypothesis testing~\cite{tomamichel_2013,li_2014}, we can then obtain an exact second-order expansion of privacy amplification with quantum side information under trace distance. The achievability of this expansion was previously shown in~\cite{shen_2024}, although its optimality was only argued as an ensemble converse for strongly 2-universal hash functions.
The classical case of this result was previously shown by~\cite{hayashi_2016-2}.

\begin{boxed}
\begin{corollary}\label{cor:second_order_tight}
For any $\ve \in (0,1)$, the second-order expansion of extractable randomness under trace distance is given by
\begin{equation}\begin{aligned}
  \ell_\ve(\rho_{XE}^{\otimes n}) = n \,H(X|E)_\rho + \sqrt{n \, V(X|E)_\rho} \,\, \Phi^{-1}(\ve) + O(\log n)
\end{aligned}\end{equation}
with \univ-universal hashes being optimal among all choices of hash functions.
Here, $V(X|E)_\rho \coloneqq V\!\left(\rho_{XE} \| \id_X \otimes \rho_E\right)$ with
$V(\rho\|\sigma) \coloneqq \Tr \left( \rho ( \log \rho - \log \sigma)^2 \right) - D(\rho\|\sigma)^2$ denoting the quantum relative entropy variance, and $\Phi^{-1}(\ve)$ is the inverse of the cumulative normal distribution function $\Phi$.
\end{corollary}
\end{boxed}
\begin{proof}
The converse in Proposition~\ref{prop:converse} tells us that, for any choice of hash function, the size of extractable randomness is at most
\begin{equation}\begin{aligned}
  \ell_\ve(\rho_{XE}^{\otimes n}) &\leq \Hmin{\ve+\delta}{M}[down] (X^n|E^n)_{\rho^{\otimes n}} +  \log\frac{1}{\delta}\\
  &\leq - D_H^{1-\ve-2\delta}\!\left(\rho_{XE}^{\otimes n} \middle\| \id_X^{\otimes n} \otimes \rho_E^{\otimes n}\right) + \log\frac{1}{\delta^2}
\end{aligned}\end{equation}
with the second line by~\cite[Lemma~7]{regula_2025}.

On the other hand, from our achievability result in Theorem~\ref{thm:privacy_amp} we have that a choice of \univ-universal hash function family achieves
\begin{equation}\begin{aligned}
  \ell_\ve(\rho_{XE}^{\otimes n}) &\geq \Hs{2}{\epsilon-\mu}{M}[up] (X^n|E^n)_{\rho^{\otimes n}} - \log\frac{1}{4\mu^2}\\
  &\geq - D_H^{1-\ve+\mu} \!\left(\rho_{XE}^{\otimes n} \middle\| \id_X^{\otimes n} \otimes \rho_E^{\otimes n}\right) - \log\frac{1}{4\mu^2},
\end{aligned}\end{equation}
where we used the one-shot inequality in Lemma~\ref{lem:secondorder_equiv}.

The second-order expansion of the hypothesis testing relative entropy~\cite{tomamichel_2013,li_2014}  then tells us that
\begin{equation}\begin{aligned}
    D^{1-\ve \pm \vphantom{\mu}\smash{\frac{c}{\sqrt{n}}}}_{H} \!\left(\rho_{XE}^{\otimes n} \middle\| \id_X^{\otimes n} \otimes \rho_E^{\otimes n}\right) = n D\big(\rho_{XE} \big\| \id_X \otimes \rho_E \big) + \sqrt{n\, V\big(\rho_{XE}\big\|\id_X \otimes \rho_E\big)} \, \Phi^{-1}(1\!-\!\ve) + O(\log n)
\end{aligned}\end{equation}
for any constant $c$; the additional term $\frac{c}{\sqrt n}$ does not affect the expansion due to the continuous differentiability of $\Phi^{-1}$~\cite[Sec.~VI]{tomamichel_2013}.
Picking $\delta = \mu = \frac{1}{\sqrt{n}}$ and using that $\Phi^{-1}(1-\ve) = -\Phi^{-1}(\ve)$ concludes the proof.
\end{proof}

Following the above arguments, from Lemma~\ref{lem:secondorder_equiv} we see generally that the second-order expansion of $\Ds{2}{\ve}{M}$ is equivalent to that of the hypothesis testing relative entropy $D^{1-\ve}_H$ and of $\Dmax\ve M$, in the sense that
\begin{equation}\begin{aligned}
  \Ds{2}{\ve}{M}(\rho^{\otimes n}\|\sigma^{\otimes n}) = n D(\rho\|\sigma) - \sqrt{n \, V(\rho\|\sigma)} \, \Phi^{-1}(\ve) + O(\log n).
\end{aligned}\end{equation}
It is worth contrasting this with other smooth divergences used in quantum information that employ different smoothing. The second-order asymptotics of trace-distance--smoothed divergences such as max-relative entropy are simply not known --- it is in general not possible to give a uniform second-order expansion for quantum states~\cite{regula_2025}. 
Purified distance smoothing does lead to known second-order expansions of divergences~\cite{tomamichel_2013}, scaling as $\Phi^{-1}(\ve^2)$ rather than $\Phi^{-1}(\ve)$ as above. However, attempting to obtain second-order asymptotics of privacy amplification with a composable security criterion based on purified distance once again encounters issues, and there is no uniform second-order expansion valid across all quantum states~\cite{anshu_2020,abdelhadi_2020}. 

\subsection{Converse for error exponent}

Even classically, the asymptotic optimality of the bounds on the error exponent (such as in Theorem~\ref{thm:renyi_achiev}) does not appear to be precisely understood. 
Hayashi and Watanabe~\cite{hayashi_2013,hayashi_2016-3} derived asymptotic ensemble converse bounds that match the achievable ones only under an assumption of not merely pairwise, but complete independence of the considered hash functions. Such hash functions are known as strongly universal$_\omega$ in the terminology of Wegman and Carter~\cite{wegman_1981}, and they can be understood as a random-binning--style hashing that corresponds to uniform sampling over \emph{all} functions $\X \to \Z$, making it extremely hard to implement in practice due to the required seed size. Although it may be tempting to conjecture that the asymptotic exponent resulting from leftover hashing (as in Theorem~\ref{thm:renyi_achiev} or the earlier classical~\cite{hayashi_2013} and quantum~\cite{dupuis_2023} results) should be tight more generally, Hayashi~\cite{hayashi_2013} showed that tailored hash functions that depend on the source distribution can achieve even larger exponents. Due to this, it does not appear possible to obtain a matching converse that would prove Theorem~\ref{thm:renyi_achiev} to be unambiguously tight. It remains plausible, however, that the exponent of Theorem~\ref{thm:renyi_achiev} could be argued to be ensemble-tight under some mild assumptions of universality. Any such argument would give stronger evidence for the optimality of \univ-universal hashes for practical purposes, also at the level of asymptotic i.i.d.\ error exponents.

However, we are here especially interested in general converse bounds that do not require any assumptions about the considered hash functions, as they give strongest possible constraints. Such bounds can be understood as constraining the asymptotics of the best achievable error,
\begin{equation}\begin{aligned}
  \ve \left(\rho_{XE}, \log\left|\Z\right|\right) \coloneqq \min_{f \colon \X \to \Z} \frac12 \norm{ \rho^f_{ZE} - \frac{\id_Z}{\left|\Z\right|} \otimes \rho_E }{1}.
\end{aligned}\end{equation}
Converse bounds of this type are known classically~\cite{hayashi_2013,hayashi_2016-3}, although little is known about general converse bounds on the exponent in the case of quantum side information~\cite[Table~3]{cheng_2025}. 
From our Proposition~\ref{prop:converse} we can obtain a quantum bound on the exponent, which in the classical case indeed recovers the tight classical converse bound established in~\cite{hayashi_2013,hayashi_2016-3}.

\begin{boxed}
\begin{corollary}\label{cor:error_exp_general}
The error exponent of any randomness extraction protocol must satisfy
\begin{equation}\begin{aligned}
  \limsup_{n\to\infty} - \frac1n \log \ve(\rho_{XE}^{\otimes n}, nR ) \leq \sup_{\alpha > 1} \, (\alpha - 1) \left( \wt{H}^\downarrow_{\alpha} (X | E)_{\rho} - R \right).
\end{aligned}\end{equation}
\end{corollary}
\end{boxed}
\begin{proof}
Assume that there exists a sequence of randomness extraction protocols such that  $\ell_{\ve_n} (\rho_{XE}^{\otimes n}) = n R$.
Proposition~\ref{prop:converse} then gives
\begin{equation}\begin{aligned}
  - n R &\geq \Dmax{k \ve_n}{M} \left( \rho_{XE}^{\otimes n} \middle\| \id_{X^n} \otimes \rho_E^{\otimes n} \right) - \log \frac{k}{k-1}\\
  &= \Dmax{k \ve_n}{M} \left( \rho_{XE}^{\otimes n} \middle\| \left(\frac{\id_{X}}{|\X|}\right)^{\otimes n} \otimes \rho_E^{\otimes n} \right) - n \log |\X| - \log \frac{k}{k-1}
\end{aligned}\end{equation}
for some constant $k > 1$. 
Introducing the notation $\tau_{XE} \coloneqq \frac{\id_{X}}{|\X|} \otimes \rho_E$ and recalling that $\Dmax{\ve}{M} (\rho\|\sigma) = \min \lset \lambda \bar \Tr\!\left( \rho - \exp[\lambda]\, \sigma\right)_+ \leq \ve \rset$, we can equivalently write this as
\begin{equation}\begin{aligned}
  \Tr\!\left(  \rho_{XE}^{\otimes n} - \exp\!\left[ n \log\left|\X\right| - nR + \log\frac{k}{k-1} \right] \tau_{XE}^{\otimes n} \right)_+  \leq k \ve_n.
\end{aligned}\end{equation}
Then
\begin{equation}\begin{aligned}
  \limsup_{n\to\infty} - \frac1n \log \ve_n &= \limsup_{n\to\infty} - \frac1n \log k \ve_n\\
  &\leq \limsup_{n\to\infty} -\frac1n \log  \Tr\!\left(  \rho_{XE}^{\otimes n} - \frac{k}{k-1} \exp\!\left[ n \left(\log\left|\X\right| - R\right) \right] \tau_{XE}^{\otimes n} \right)_+\\
  &= \limsup_{n\to\infty} -\frac1n \log  \Tr\!\Big(  \rho_{XE}^{\otimes n} - \exp\!\left[ n \left(\log\left|\X\right| - R\right) \right] \tau_{XE}^{\otimes n} \Big)_+\\
  &\texteq{(i)} \sup_{\alpha > 1} \,\, (\alpha - 1) \left( \log |\X| - R - \wt{D}_\alpha \!\left( \rho_{XE} \middle\| \tau_{XE} \right) \right)\\
  &= \sup_{\alpha > 1} \,\, (\alpha - 1) \left( \wt{H}^\downarrow_{\alpha} (X | E)_{\rho} - R \right)
\end{aligned}\end{equation}
where (i) is exactly the result of~\cite[Theorem~IV.4]{mosonyi_2015}.
\end{proof}

\subsection{Strong converse exponent}\label{sec:strong-converse}

When attempting to extract randomness at a rate higher than the conditional entropy $H(X|E)_\rho$, the error of the protocol must inevitably converge to one~\cite{salzmann_2022}. The study of the so-called strong converse exponents is concerned with understanding the exponent of this convergence. 
Several bounds on the strong converse exponent can be obtained from our approach.

\begin{boxed}
\begin{proposition}\label{prop:sc_bounds}
The strong converse exponent of quantum privacy amplification must satisfy 
\begin{equation}\begin{aligned}\label{eq:sc_bound_petz}
  \liminf_{n\to\infty} -\frac1n \log \!\left(1 - \ve(\rho_{XE}^{\otimes n}, nR )\right) &\geq\, \sup_{\alpha \in (0,1)} \frac{\alpha-1}{2-\alpha} \left( \overline{H}^\downarrow_\alpha(X|E)_\rho - R\right)
\end{aligned}\end{equation}
and
\begin{equation}\begin{aligned}\label{eq:sc_bound_sandwiched}
  \liminf_{n\to\infty} -\frac1n \log \!\left(1 - \ve(\rho_{XE}^{\otimes n}, nR )\right) &\geq\, \sup_{\alpha \in (0,1)} \,(\alpha-1) \left( \wt{H}^\downarrow_\alpha(X|E)_\rho - R\right).
\end{aligned}\end{equation}

Furthermore, there exists a sequence of randomness extraction protocols that gives an achievability bound of
\begin{equation}\begin{aligned}\label{eq:sc_bound_achiev}
  \limsup_{n\to\infty} -\frac1n \log \!\left(1 - \ve(\rho_{XE}^{\otimes n}, nR )\right) &\leq\, \sup_{\alpha \in (0,1)} \frac{\alpha-1}{\alpha} \left( \overline{H}^\uparrow_\alpha(X|E)_\rho - R\right).
\end{aligned}\end{equation}
\end{proposition}
\end{boxed}

The first of the bounds improves on a bound of~\cite{salzmann_2022}. The second one features an even better prefactor $(\alpha-1)$ but, because of the change from Petz to sandwiched R\'enyi divergences, may not be comparable in general. Although both of the lower bounds here appear weaker than a result found in~\cite{shen_2022}, the latter is only an ensemble converse --- a caveat that we will discuss shortly.

The achievability result given in~\eqref{eq:sc_bound_achiev} appears to be new even classically.

\begin{proof}
From the one-shot converse in Proposition~\ref{prop:converse} we get
\begin{equation}\begin{aligned}
  \ell_\ve(\rho_{XE}) \leq \Hmin{\ve+\delta}{M}[down](X|E)_{\rho}  + \log \frac{1}{\delta}
  \leq \overline{H}^\downarrow_\alpha(X|E)_\rho + \frac{1}{1-\alpha} \log \frac{1}{1-\ve - \delta} + \log \frac{1}{\delta},
\end{aligned}\end{equation}
where the second inequality, valid for all $\alpha \in (0,1)$, follows by an application of Audenaert's inequality~\cite{audenaert_2007} (see~\cite[Corollary~15]{regula_2025}). With the choice of $\delta = k (1-\ve)$ for some $k \in (0,1)$, this translates to
\begin{equation}\begin{aligned}
  \ell_\ve(\rho_{XE}) &\leq \overline{H}^\downarrow_\alpha(X|E)_\rho + \frac{2-\alpha}{1-\alpha} \log\frac{1}{1-\ve} + \log\frac{1}{k(1-k)^{1/(1-\alpha)}}.
\end{aligned}\end{equation}
Applying this to $\rho_{XE}^{\otimes n}$, setting $|\Z| = \exp(nR)$, and using the additivity of the Petz--R\'enyi divergences gives
\begin{equation}\begin{aligned}
  \liminf_{n\to\infty} -\frac1n \log \!\left(1 - \ve(\rho_{XE}^{\otimes n}, nR )\right) &\geq \sup_{\alpha \in (0,1)}\, \frac{\alpha-1}{2-\alpha} \left( \overline{H}^\downarrow_\alpha(X|E)_\rho - R\right).
\end{aligned}\end{equation}

For the second bound in~\eqref{eq:sc_bound_sandwiched}, starting as in the proof of Proposition~\ref{prop:converse}, we proceed as
\begin{equation}\begin{aligned}\label{eq:converse_petz_derivation}
  -\log |\Z| &\geq \Dmax{\ve}{T}(\rho_{ZE}^f\| \id_Z \otimes \rho_E)\\
 &\textgeq{(i)} \ol{D}_\alpha (\rho_{ZE}^f\| \id_Z \otimes \rho_E) - \frac{1}{1-\alpha} \log\frac{1}{1-\ve} \\
 &\textgeq{(ii)} \wt{D}_\alpha (\rho_{ZE}^f\| \id_Z \otimes \rho_E) - \frac{1}{1-\alpha} \log\frac{1}{1-\ve} \\
 &\textgeq{(iii)} \wt{D}_\alpha (\rho_{XE} \| \id_X \otimes \rho_E ) - \frac{1}{1-\alpha} \log\frac{1}{1-\ve},
\end{aligned}\end{equation}
where: (i) is again a consequence of Audenaert's inequality, shown in~\cite[Corollary~15]{regula_2025};  (ii) holds for all $\alpha \in (0,1)$ by an application of the Araki--Lieb--Thirring inequality~\cite[Lemma~3]{datta_2014}; (iii) follows by the monotonicity of the sandwiched R\'enyi divergences under coarse-graining, first shown in~\cite[Lemma~5.14]{tomamichel_2016} for $\alpha \geq \frac12$ and clarified to hold for all $\alpha \in (0,1)$ in~\cite[Lemma~3]{rubboli_2026}.
Rearranging, this gives the stated bound.

The achievability in~\eqref{eq:sc_bound_achiev} is seen as follows. Recall from Theorem~\ref{thm:privacy_amp} that, for a choice of state $\sigma_E$ and $\ve \in (0,1)$, any family of \univ-universal hash functions yields a randomness extraction protocol satisfying
\begin{equation}\begin{aligned}\label{eq:privacy_amp_strong_conv_achiev}
  \EE_f \frac12 \norm{ \rho_{ZE}^f - \frac{\id_Z}{|\Z|} \otimes \rho_E }{1} &\leq \ve + \frac12 \sqrt{|\Z| \, \Qs{2}{\ve}{M}(\rho_{XE} \| \id_X \otimes \sigma_E)}\\
  &\leq \ve + \frac12 \sqrt{|\Z| \, (1-\ve) \, \Qmax{\ve}{M}(\rho_{XE} \| \id_X \otimes \sigma_E)}\\
  &\leq \ve + \frac12 \sqrt{|\Z| \, (1-\ve)^2 \, \exp\!\left(D^{1-\ve}_{H}(\rho_{XE} \| \id_X \otimes \sigma_E)\right)},
\end{aligned}\end{equation}
with the last two lines using the inequalities in Lemma~\ref{lem:secondorder_equiv}. 
Taking $\ve = 1 -\exp(-n T_n)$ for some choice of exponent $T_n$ and setting $|\Z| = \exp(nR)$ , this gives
\begin{equation}\begin{aligned}
  1 - \ve(\rho_{XE}^{\otimes n}, nR ) &\geq \exp(-n T_n) - \frac12 \exp\left(\frac12 n R - n T_n + \frac12 \,\inf_{\sigma_{E^n}}\, D^{\exp(-n T_n)}_{H}\left(\rho_{XE}^{\otimes n} \middle\| \id^{\otimes n}_X \otimes \sigma_{E^n}\right)\right),
\end{aligned}\end{equation}
where $\sigma_{E^n}$ can be general, non-i.i.d.\ states. 
Setting $\pi_X \coloneqq \frac{1}{|\X|} \id_X$, notice that if we choose $T_n$ so that
\begin{equation}\begin{aligned}\label{eq:product_testing}
\inf_{\sigma_{E^n}}\, D^{\exp(-nT_n)}_{H}\!\left(\rho_{XE}^{\otimes n} \middle\| \pi^{\otimes n}_X \otimes \sigma_{E^n}\right) \leq n \log \left| \X \right| - n R,
\end{aligned}\end{equation}
then $\inf_{\sigma_{E^n}}\! D^{\exp(-nT_n)}_{H}\!\left(\rho_{XE}^{\otimes n} \middle\| \id^{\otimes n}_X \otimes \sigma_{E^n}\right) 
= \inf_{\sigma_{E^n}}\! D^{\exp(-nT_n)}_{H}\!\left(\rho_{XE}^{\otimes n} \middle\| \pi^{\otimes n}_X \otimes \sigma_{E^n}\right) - n \log \left| \X \right| \leq - n R$, and hence
\begin{equation}\begin{aligned}\label{eq:sc_bound_with_hyptest}
  1 - \ve(\rho_{XE}^{\otimes n}, nR ) &\geq \exp(-n T_n) - \frac12 \exp(-nT_n) = \frac12 \exp( -n T_n).
\end{aligned}\end{equation}
The problem of hypothesis testing between an i.i.d.\ state and non-i.i.d.\ product states as in~\eqref{eq:product_testing} is precisely the setting of the work~\cite{hayashi_2016-1}, whose Corollary~13 shows that the best choice of such $T_n$ satisfies
\begin{equation}\begin{aligned}
 \lim_{n\to\infty} T_n &= \sup_{\alpha \in (0,1)} \frac{\alpha-1}{\alpha} \left( \log\left|\X\right| - R - \inf_{\sigma_E} \ol{D}_\alpha(\rho_{XE}\|\pi_X \otimes \sigma_E)\right)\\
 &= \sup_{\alpha \in (0,1)} \frac{\alpha-1}{\alpha} \left( \overline{H}^\uparrow_\alpha(X|E)_\rho - R\right).
\end{aligned}\end{equation}
Combining this with~\eqref{eq:sc_bound_with_hyptest} gives the claimed result.
\end{proof}

There are several points to remark here. First, if we had the one-shot converse bound $\ell_\ve(\rho_{XE}) \leq \Hmin{\ve}{M}[down](X|E)_{\rho}$, which is true classically, then we could improve the converse bound of Eq.~\eqref{eq:sc_bound_sandwiched} from the sandwiched R\'enyi conditional entropies $\smash{\wt{H}^{\downarrow}_\alpha}$ to the Petz--R\'enyi conditional entropies $\smash{\ol{H}^{\downarrow}_\alpha}$. However, as we noted earlier, this one-shot converse bound was ruled out for general quantum states in~\cite{renes_2018}; our approximate converse in Proposition~\ref{prop:converse}, although sufficient for error exponents, is weaker in the strong converse regime. Alternatively, the same improved strong converse bound would follow if the inequality $\wt{D}_\alpha (\rho_{ZE}^{\smash{f}}\| \id_Z \otimes \rho_E) \geq \wt{D}_\alpha (\rho_{XE} \| \id_X \otimes \rho_E )$ that we used in Eq.~\eqref{eq:converse_petz_derivation} could be shown for the Petz--R\'enyi divergences. However, a numerical investigation shows that this inequality is not true in general, and indeed such coarse-graining properties appear closely connected to sandwiched R\'enyi entropies~\cite{tomamichel_2016,rubboli_2026}.

These observations suggest that a strong converse bound of the form~\eqref{eq:sc_bound_sandwiched} but with Petz--R\'enyi entropies should not be possible. This very bound is, however, exactly the result of~\cite{shen_2022}. This is no contradiction, as~\cite{shen_2022} only claimed this result as an ensemble converse bound for strongly 2-universal (pairwise independent) hash functions --- the result of~\cite{shen_2022} does not show a strong converse exponent, or even a strong converse property, if more general hash functions are allowed. 
This points to a conjecture that the strong converse exponent under general hash functions could be strictly smaller than the ensemble strong converse exponent under (strongly) 2-universal hashes, connecting also to our earlier discussion of similar phenomena in the study of error exponents. This motivates a closer investigation of the monotonicity properties of different types of R\'enyi and smooth divergences, especially in the context of their interplay with universal hash functions. 

In relation to the non-matching upper and lower bounds that we obtained in Eqs.~\eqref{eq:sc_bound_sandwiched}--\eqref{eq:sc_bound_achiev}, as well as the very similar bounds for the error exponent in Theorem~\ref{thm:renyi_achiev} and Corollary~\ref{cor:error_exp_general}, we also point to the recent developments in~\cite{li_2025-1,rubboli_2026} that studied strong converse exponents under purified distance and showed the need to consider a class of conditional R\'enyi entropies that go beyond the two extremal variants $\uparrow$ and $\downarrow$ that we studied here. It would be interesting to understand whether this can find use in the study of trace distance.

\section{Discussion}

We introduced the family of measured smooth R\'enyi divergences and showed how it enables a tight one-shot analysis of quantum privacy amplification with trace distance, providing also the first unified derivation of asymptotic constraints in the form of an optimal second-order expansion and the tightest known bounds on the asymptotic exponents. 
A fundamental role here was played by the measured smooth collision entropy $H_{2}\s{\ve}{M}$ and the measured smooth min-entropy $H_{\min}\s{\ve}{M}$, identifying them as the most appropriate generalisations of the classical smooth conditional entropies in the context of quantum privacy amplification, and yielding improvements over all prior quantum formulations of the leftover hash lemma and its fully quantum variant of decoupling.

The results indicate the need to reconsider how smooth entropies and divergences should be defined in the study of quantum information processing, as conventional choices~\cite{renner_2005,tomamichel_2011,tomamichel_2016} lead to suboptimal achievability results. 
Although we certainly wish we could claim that this indisputably establishes our definitions as the `right' approach to smoothing, there are caveats to such a claim. 
For example, in the derivation of the converse bounds in Section~\ref{sec:converse}, a key step necessitated the use of the purified distance to argue a suitable monotonicity under all hash functions. This caused a somewhat surprising, direct appearance of the purified-distance--smoothed min-entropy in the analysis of privacy amplification under trace distance. Whether this can be avoided and a more direct argument for the trace distance can be shown is a very interesting open question. 
Additionally, we discussed how the measured smoothing can be equivalently understood as smoothing over non-positive Hermitian operators, the latter being a natural non-commutative analogue of similar classical notions. This did not lead to any issues in our analysis of privacy amplification, as the smoothing operator $R$ there plays a purely mathematical role in bounding the amount of extractable randomness. However, in some other quantum information processing tasks, smoothing has a more direct interpretation as a process of identifying a suitable approximation of $\rho$ that one aims to prepare or use in practice through a physical protocol. The need to smooth over Hermitian operators could then be an obstruction to the application of measured smooth divergences in such contexts, perhaps hinting at the impossibility of generalising the tightest classical results to quantum states. 
Nevertheless, due to the fundamental role of privacy amplification and decoupling in the study of many other operational problems in quantum information, we are sure that the approach of this work will lead to improvements in the analysis of a broad range of tasks, even if measured smoothing does not completely replace conventional smoothing notions but rather complements them by enabling tight one-shot achievability results.

One conceptual point to note about the conditional entropies used in our work is that they are defined using global measurements: $\Hmin{\ve}{M}(X|E)_\rho \!=\! - \inf_{\sigma_E} \sup_{\M\in\MM} D^{\ve,\smash{T}}_{\max}(\M(\rho_{XE})\|\M(\id_X \otimes \sigma_E))$ where $\M$ are measurement channels on the whole system $XE$. 
As we discuss in Appendix~\ref{app:guess_prob}, the global character of the measurements in our definitions can be interpreted in the context of the guessing probability of the value of $X$ by the adversary $E$.
Alternatively, the choice of the set of all measurements $\MM$ can be understood as coming from the constraints placed on the distinguisher --- the external reference whose task is to distinguish the output of the randomness extraction protocol from uniform randomness, commonly used to motivate security criteria~\cite{portmann_2022,ferradini_2025}. This is because the trace distance itself is a measured divergence, in the sense that $\norm{X}{1} = \sup_{\M\in\MM}\norm{\M(X)}{1}$~\cite{helstrom_1969,holevo_1973}. 
This immediately gives an idea of how our approach could be extended: if one constrains the set of measurements $\MM$ that the distinguisher may implement, this would replace the trace distance with a suitable distinguishability norm~\cite{matthews_2009}, naturally suggesting a connection with a restricted-measurement notion of measured smooth entropies, that is, one where $\MM$ is not the set of all measurements but some appropriate subset thereof. A recently studied application in this direction is the restriction to computationally-bounded distinguishers~\cite{chen_2017,avidan_2026}, leading to notions of `pseudo-randomness' that mirror classical studies of computational entropies~\cite{hastad_1999,barak_2003}.
The fact that our family of divergences is defined directly through measurements makes them perfectly suited for applications to settings of restricted measurements, and an extension of our techniques in this way would be an interesting follow-up development. Yet another variant of conditional entropies based on measured R\'enyi divergences was defined in~\cite{hanson_2022} in the context of quantum guesswork, although it is not clear if those definitions could be used in the applications studied in this work.

We also point out that one often encounters relaxed statements of the leftover hash lemma in terms of smooth min-entropy, which in our case would correspond to an upper bound of the form $\ve + \frac12 \big[ |\Z| \, \exp\!\smash{\big(}\!- \Hmin{\ve}{M}(X|E)_\rho \smash{\big)} \smash{\big]}^{1/2}$. Although already tighter than standard quantum formulations~\cite{renner_2005,tomamichel_2011,anshu_2020}, this expression would be insufficient to yield several of our results. For instance, the derivation of the achievability of the error exponent in Theorem~\ref{thm:renyi_achiev}, the asymptotic achievability of the strong converse exponent in Proposition~\ref{prop:sc_bounds}, and one-shot bounds such as Corollary~\ref{cor:achiev_purified} relied on tighter bounds that follow from the analysis of the measured smooth collision entropy $H_2\s{\ve}{M}$\!. This shows the importance of smoothing the R\'enyi divergence of order 2, reinforcing a point first made by Hayashi~\cite{hayashi_2013,hayashi_2016-2} in the classical~case.

\let\oldaddcontentsline\addcontentsline
\renewcommand{\addcontentsline}[3]{}

\begin{acknowledgments}

We thank Hao-Chung Cheng, Ian George, Roberto Rubboli, and Ernest Y.-Z. Tan for helpful discussions. We acknowledge Red Bull~GmbH, without whom this work would not have been possible.
   
This project is supported by the Japan Science and Technology Agency (JST) PRESTO grant no.\ JPMJPR25FB and by the Singapore NRF Investigatorship award (NRF-NRFI10-2024-0006).

\end{acknowledgments}


 \bibliographystyle{apsca}
 \bibliography{main}

\newcommand{\etalchar}[1]{$^{#1}$}
\begin{thebibliography}{HKDW22}
\makeatletter
\providecommand \@ifxundefined [1]{%
 \@ifx{#1\undefined}
}%
\providecommand \@ifnum [1]{%
 \ifnum #1\expandafter \@firstoftwo
 \else \expandafter \@secondoftwo
 \fi
}%
\providecommand \@ifx [1]{%
 \ifx #1\expandafter \@firstoftwo
 \else \expandafter \@secondoftwo
 \fi
}%
\providecommand \natexlab [1]{#1}%
\providecommand \emph  [1]{``#1''}%
\providecommand \bibnamefont  [1]{#1}%
\providecommand \bibfnamefont [1]{#1}%
\providecommand \citenamefont [1]{#1}%
\providecommand \href@noop [0]{\@secondoftwo}%
\providecommand \href [0]{\begingroup \@sanitize@url \@href}%
\providecommand \@href[1]{\@@startlink{#1}\@@href}%
\providecommand \@@href[1]{\endgroup#1\@@endlink}%
\providecommand \@sanitize@url [0]{\catcode `\\12\catcode `\$12\catcode `\&12\catcode `\#12\catcode `\^12\catcode `\_12\catcode `\%12\relax}%
\providecommand \@@startlink[1]{}%
\providecommand \@@endlink[0]{}%
\providecommand \url  [0]{\begingroup\@sanitize@url \@url }%
\providecommand \@url [1]{\endgroup\@href {#1}{\urlprefix }}%
\providecommand \urlprefix  [0]{URL }%
\providecommand \Eprint [0]{\href }%
\providecommand \doibase [0]{http://dx.doi.org/}%
\providecommand \selectlanguage [0]{\@gobble}%
\providecommand \bibinfo  [0]{\@secondoftwo}%
\providecommand \bibfield  [0]{\@secondoftwo}%
\providecommand \translation [1]{[#1]}%
\providecommand \BibitemOpen [0]{}%
\providecommand \bibitemStop [0]{}%
\providecommand \bibitemNoStop [0]{.\EOS\space}%
\providecommand \EOS [0]{\spacefactor3000\relax}%
\providecommand \BibitemShut  [1]{\csname bibitem#1\endcsname}%
\let\auto@bib@innerbib\@empty
\bibitem[AA26]{avidan_2026}
\bibfield  {author} {\bibinfo {author} {\bibfnamefont {N.}~\bibnamefont {Avidan}}\ and\ \bibinfo {author} {\bibfnamefont {R.}~\bibnamefont {Arnon}},\ }\bibfield  {title} {\emph {\bibinfo {title} {Quantum {{Computational Unpredictability Entropy}} and {{Quantum Leakage Resilience}}},}\ }\href {http://dx.doi.org/10.1109/TIT.2026.3658830} {\bibfield  {journal} {\bibinfo  {journal} {IEEE Trans. Inf. Theory}\ }\textbf {\bibinfo {volume} {72}},\ \bibinfo {pages} {3129} (\bibinfo {year} {2026})}\BibitemShut {NoStop}%
\bibitem[ABJT20]{anshu_2020}
\bibfield  {author} {\bibinfo {author} {\bibfnamefont {A.}~\bibnamefont {Anshu}}, \bibinfo {author} {\bibfnamefont {M.}~\bibnamefont {Berta}}, \bibinfo {author} {\bibfnamefont {R.}~\bibnamefont {Jain}}, \ and\ \bibinfo {author} {\bibfnamefont {M.}~\bibnamefont {Tomamichel}},\ }\bibfield  {title} {\emph {\bibinfo {title} {Partially {{Smoothed Information Measures}}},}\ }\href {http://dx.doi.org/10.1109/TIT.2020.2981573} {\bibfield  {journal} {\bibinfo  {journal} {IEEE Trans. Inf. Theory}\ }\textbf {\bibinfo {volume} {66}},\ \bibinfo {pages} {5022} (\bibinfo {year} {2020})}\BibitemShut {NoStop}%
\bibitem[ACM{\etalchar{+}}07]{audenaert_2007}
\bibfield  {author} {\bibinfo {author} {\bibfnamefont {K.~M.~R.}\ \bibnamefont {Audenaert}}, \bibinfo {author} {\bibfnamefont {J.}~\bibnamefont {Calsamiglia}}, \bibinfo {author} {\bibfnamefont {R.}~\bibnamefont {{Mu{\~n}oz-Tapia}}}, \bibinfo {author} {\bibfnamefont {E.}~\bibnamefont {Bagan}}, \bibinfo {author} {\bibfnamefont {{\relax Ll}.}~\bibnamefont {Masanes}}, \bibinfo {author} {\bibfnamefont {A.}~\bibnamefont {Acin}}, \ and\ \bibinfo {author} {\bibfnamefont {F.}~\bibnamefont {Verstraete}},\ }\bibfield  {title} {\emph {\bibinfo {title} {Discriminating {{States}}: {{The Quantum Chernoff Bound}}},}\ }\href {http://dx.doi.org/10.1103/PhysRevLett.98.160501} {\bibfield  {journal} {\bibinfo  {journal} {Phys. Rev. Lett.}\ }\textbf {\bibinfo {volume} {98}},\ \bibinfo {pages} {160501} (\bibinfo {year} {2007})}\BibitemShut {NoStop}%
\bibitem[ADHW09]{abeyesinghe_2009}
\bibfield  {author} {\bibinfo {author} {\bibfnamefont {A.}~\bibnamefont {Abeyesinghe}}, \bibinfo {author} {\bibfnamefont {I.}~\bibnamefont {Devetak}}, \bibinfo {author} {\bibfnamefont {P.}~\bibnamefont {Hayden}}, \ and\ \bibinfo {author} {\bibfnamefont {A.}~\bibnamefont {Winter}},\ }\bibfield  {title} {\emph {\bibinfo {title} {The mother of all protocols: Restructuring quantum information's family tree},}\ }\href {http://dx.doi.org/10.1098/rspa.2009.0202} {\bibfield  {journal} {\bibinfo  {journal} {Proc. A}\ }\textbf {\bibinfo {volume} {465}},\ \bibinfo {pages} {2537} (\bibinfo {year} {2009})}\BibitemShut {NoStop}%
\bibitem[AHT25]{arqand_2025}
\bibfield  {author} {\bibinfo {author} {\bibfnamefont {A.}~\bibnamefont {Arqand}}, \bibinfo {author} {\bibfnamefont {T.~A.}\ \bibnamefont {Hahn}}, \ and\ \bibinfo {author} {\bibfnamefont {E.~Y.-Z.}\ \bibnamefont {Tan}},\ }\bibfield  {title} {\emph {\bibinfo {title} {Generalized {{R}}{\'enyi} {{Entropy Accumulation Theorem}} and {{Generalized Quantum Probability Estimation}}},}\ }\href {http://dx.doi.org/10.1103/pgrn-mz9j} {\bibfield  {journal} {\bibinfo  {journal} {Phys. Rev. X}\ }\textbf {\bibinfo {volume} {15}},\ \bibinfo {pages} {041013} (\bibinfo {year} {2025})}\BibitemShut {NoStop}%
\bibitem[ANSV08]{audenaert_2008}
\bibfield  {author} {\bibinfo {author} {\bibfnamefont {K.~M.~R.}\ \bibnamefont {Audenaert}}, \bibinfo {author} {\bibfnamefont {M.}~\bibnamefont {Nussbaum}}, \bibinfo {author} {\bibfnamefont {A.}~\bibnamefont {Szko{\l}a}}, \ and\ \bibinfo {author} {\bibfnamefont {F.}~\bibnamefont {Verstraete}},\ }\bibfield  {title} {\emph {\bibinfo {title} {Asymptotic {{Error Rates}} in {{Quantum Hypothesis Testing}}},}\ }\href {http://dx.doi.org/10.1007/s00220-008-0417-5} {\bibfield  {journal} {\bibinfo  {journal} {Commun. Math. Phys.}\ }\textbf {\bibinfo {volume} {279}},\ \bibinfo {pages} {251} (\bibinfo {year} {2008})}\BibitemShut {NoStop}%
\bibitem[AR20]{abdelhadi_2020}
\bibfield  {author} {\bibinfo {author} {\bibfnamefont {D.}~\bibnamefont {Abdelhadi}}\ and\ \bibinfo {author} {\bibfnamefont {J.~M.}\ \bibnamefont {Renes}},\ }\bibfield  {title} {\emph {\bibinfo {title} {On the {{Second-Order Asymptotics}} of the {{Partially Smoothed Conditional Min-Entropy}} \& {{Application}} to {{Quantum Compression}}},}\ }\href {http://dx.doi.org/10.1109/JSAIT.2020.3016899} {\bibfield  {journal} {\bibinfo  {journal} {IEEE J. Sel. Areas Inf. Theory}\ }\textbf {\bibinfo {volume} {1}},\ \bibinfo {pages} {416} (\bibinfo {year} {2020})}\BibitemShut {NoStop}%
\bibitem[BBCM95]{bennett_1995}
\bibfield  {author} {\bibinfo {author} {\bibfnamefont {C.}~\bibnamefont {Bennett}}, \bibinfo {author} {\bibfnamefont {G.}~\bibnamefont {Brassard}}, \bibinfo {author} {\bibfnamefont {C.}~\bibnamefont {Crepeau}}, \ and\ \bibinfo {author} {\bibfnamefont {U.}~\bibnamefont {Maurer}},\ }\bibfield  {title} {\emph {\bibinfo {title} {Generalized privacy amplification},}\ }\href {http://dx.doi.org/10.1109/18.476316} {\bibfield  {journal} {\bibinfo  {journal} {IEEE Trans. Inf. Theory}\ }\textbf {\bibinfo {volume} {41}},\ \bibinfo {pages} {1915} (\bibinfo {year} {1995})}\BibitemShut {NoStop}%
\bibitem[BBR88]{bennett_1988}
\bibfield  {author} {\bibinfo {author} {\bibfnamefont {C.~H.}\ \bibnamefont {Bennett}}, \bibinfo {author} {\bibfnamefont {G.}~\bibnamefont {Brassard}}, \ and\ \bibinfo {author} {\bibfnamefont {J.-M.}\ \bibnamefont {Robert}},\ }\bibfield  {title} {\emph {\bibinfo {title} {Privacy {{Amplification}} by {{Public Discussion}}},}\ }\href {http://dx.doi.org/10.1137/0217014} {\bibfield  {journal} {\bibinfo  {journal} {SIAM J. Comput.}\ }\textbf {\bibinfo {volume} {17}},\ \bibinfo {pages} {210} (\bibinfo {year} {1988})}\BibitemShut {NoStop}%
\bibitem[BC94]{braunstein_1994}
\bibfield  {author} {\bibinfo {author} {\bibfnamefont {S.~L.}\ \bibnamefont {Braunstein}}\ and\ \bibinfo {author} {\bibfnamefont {C.~M.}\ \bibnamefont {Caves}},\ }\bibfield  {title} {\emph {\bibinfo {title} {Statistical distance and the geometry of quantum states},}\ }\href {http://dx.doi.org/10.1103/PhysRevLett.72.3439} {\bibfield  {journal} {\bibinfo  {journal} {Phys. Rev. Lett.}\ }\textbf {\bibinfo {volume} {72}},\ \bibinfo {pages} {3439} (\bibinfo {year} {1994})}\BibitemShut {NoStop}%
\bibitem[BCR11]{berta_2011}
\bibfield  {author} {\bibinfo {author} {\bibfnamefont {M.}~\bibnamefont {Berta}}, \bibinfo {author} {\bibfnamefont {M.}~\bibnamefont {Christandl}}, \ and\ \bibinfo {author} {\bibfnamefont {R.}~\bibnamefont {Renner}},\ }\bibfield  {title} {\emph {\bibinfo {title} {The {{Quantum Reverse Shannon Theorem Based}} on {{One-Shot Information Theory}}},}\ }\href {http://dx.doi.org/10.1007/s00220-011-1309-7} {\bibfield  {journal} {\bibinfo  {journal} {Commun. Math. Phys.}\ }\textbf {\bibinfo {volume} {306}},\ \bibinfo {pages} {579} (\bibinfo {year} {2011})}\BibitemShut {NoStop}%
\bibitem[BD10]{buscemi_2010}
\bibfield  {author} {\bibinfo {author} {\bibfnamefont {F.}~\bibnamefont {Buscemi}}\ and\ \bibinfo {author} {\bibfnamefont {N.}~\bibnamefont {Datta}},\ }\bibfield  {title} {\emph {\bibinfo {title} {The {{Quantum Capacity}} of {{Channels With Arbitrarily Correlated Noise}}},}\ }\href {http://dx.doi.org/10.1109/TIT.2009.2039166} {\bibfield  {journal} {\bibinfo  {journal} {IEEE Trans. Inf. Theory}\ }\textbf {\bibinfo {volume} {56}},\ \bibinfo {pages} {1447} (\bibinfo {year} {2010})}\BibitemShut {NoStop}%
\bibitem[Bei13]{beigi_2013}
\bibfield  {author} {\bibinfo {author} {\bibfnamefont {S.}~\bibnamefont {Beigi}},\ }\bibfield  {title} {\emph {\bibinfo {title} {Sandwiched {{R{\'e}nyi}} divergence satisfies data processing inequality},}\ }\href {http://dx.doi.org/10.1063/1.4838855} {\bibfield  {journal} {\bibinfo  {journal} {J. Math. Phys.}\ }\textbf {\bibinfo {volume} {54}},\ \bibinfo {pages} {122202} (\bibinfo {year} {2013})}\BibitemShut {NoStop}%
\bibitem[Ber09]{berta_2009}
\bibfield  {author} {\bibinfo {author} {\bibfnamefont {M.}~\bibnamefont {Berta}},\ }\emph {\bibinfo {title} {Single-Shot {{Quantum State Merging}}}},\ \href@noop {} {\bibinfo {type} {diploma thesis, {ETH Z}urich}} (\bibinfo {year} {2009}),\ \Eprint {http://arxiv.org/abs/0912.4495} {arXiv:0912.4495} \BibitemShut {NoStop}%
\bibitem[BFT17]{berta_2017-1}
\bibfield  {author} {\bibinfo {author} {\bibfnamefont {M.}~\bibnamefont {Berta}}, \bibinfo {author} {\bibfnamefont {O.}~\bibnamefont {Fawzi}}, \ and\ \bibinfo {author} {\bibfnamefont {M.}~\bibnamefont {Tomamichel}},\ }\bibfield  {title} {\emph {\bibinfo {title} {On variational expressions for quantum relative entropies},}\ }\href {http://dx.doi.org/10.1007/s11005-017-0990-7} {\bibfield  {journal} {\bibinfo  {journal} {Lett. Math. Phys.}\ }\textbf {\bibinfo {volume} {107}},\ \bibinfo {pages} {2239} (\bibinfo {year} {2017})}\BibitemShut {NoStop}%
\bibitem[Bha07]{bhatia_2007}
\bibfield  {author} {\bibinfo {author} {\bibfnamefont {R.}~\bibnamefont {Bhatia}},\ }\href@noop {} {\emph {\bibinfo {title} {Positive {{Definite Matrices}}}}}\ (\bibinfo  {publisher} {Princeton University Press},\ \bibinfo {year} {2007})\BibitemShut {NoStop}%
\bibitem[BK15]{bae_2015}
\bibfield  {author} {\bibinfo {author} {\bibfnamefont {J.}~\bibnamefont {Bae}}\ and\ \bibinfo {author} {\bibfnamefont {L.-C.}\ \bibnamefont {Kwek}},\ }\bibfield  {title} {\emph {\bibinfo {title} {Quantum state discrimination and its applications},}\ }\href {http://dx.doi.org/10.1088/1751-8113/48/8/083001} {\bibfield  {journal} {\bibinfo  {journal} {J. Phys. A: Math. Theor.}\ }\textbf {\bibinfo {volume} {48}},\ \bibinfo {pages} {083001} (\bibinfo {year} {2015})}\BibitemShut {NoStop}%
\bibitem[BSW03]{barak_2003}
\bibfield  {author} {\bibinfo {author} {\bibfnamefont {B.}~\bibnamefont {Barak}}, \bibinfo {author} {\bibfnamefont {R.}~\bibnamefont {Shaltiel}}, \ and\ \bibinfo {author} {\bibfnamefont {A.}~\bibnamefont {Wigderson}},\ }\bibfield  {title} {\emph {\bibinfo {title} {Computational {{Analogues}} of {{Entropy}}},}\ }in\ \href {http://dx.doi.org/10.1007/978-3-540-45198-3_18} {\emph {\bibinfo {booktitle} {Approximation, {{Randomization}}, and {{Combinatorial Optimization}}. {{Algorithms}} and {{Techniques}}}}},\ \bibinfo {editor} {edited by\ \bibinfo {editor} {\bibfnamefont {S.}~\bibnamefont {Arora}}, \bibinfo {editor} {\bibfnamefont {K.}~\bibnamefont {Jansen}}, \bibinfo {editor} {\bibfnamefont {J.~D.~P.}\ \bibnamefont {Rolim}}, \ and\ \bibinfo {editor} {\bibfnamefont {A.}~\bibnamefont {Sahai}}}\ (\bibinfo  {publisher} {Springer},\ \bibinfo {address} {Berlin},\ \bibinfo {year} {2003})\ pp.\ \bibinfo {pages} {200--215}\BibitemShut {NoStop}%
\bibitem[BV04]{boyd_2004}
\bibfield  {author} {\bibinfo {author} {\bibfnamefont {S.}~\bibnamefont {Boyd}}\ and\ \bibinfo {author} {\bibfnamefont {L.}~\bibnamefont {Vandenberghe}},\ }\href@noop {} {\emph {\bibinfo {title} {Convex {{Optimization}}}}}\ (\bibinfo  {publisher} {Cambridge University Press},\ \bibinfo {address} {New York},\ \bibinfo {year} {2004})\BibitemShut {NoStop}%
\bibitem[Can01]{canetti_2001}
\bibfield  {author} {\bibinfo {author} {\bibfnamefont {R.}~\bibnamefont {Canetti}},\ }\bibfield  {title} {\emph {\bibinfo {title} {Universally composable security: A new paradigm for cryptographic protocols},}\ }in\ \href {http://dx.doi.org/10.1109/SFCS.2001.959888} {\emph {\bibinfo {booktitle} {Proceedings 42nd {{IEEE Symposium}} on {{Foundations}} of {{Computer Science}}}}}\ (\bibinfo {year} {2001})\ pp.\ \bibinfo {pages} {136--145}\BibitemShut {NoStop}%
\bibitem[CCL{\etalchar{+}}17]{chen_2017}
\bibfield  {author} {\bibinfo {author} {\bibfnamefont {Y.-H.}\ \bibnamefont {Chen}}, \bibinfo {author} {\bibfnamefont {K.-M.}\ \bibnamefont {Chung}}, \bibinfo {author} {\bibfnamefont {C.-Y.}\ \bibnamefont {Lai}}, \bibinfo {author} {\bibfnamefont {S.~P.}\ \bibnamefont {Vadhan}}, \ and\ \bibinfo {author} {\bibfnamefont {X.}~\bibnamefont {Wu}},\ }\bibfield  {title} {\emph {\bibinfo {title} {Computational {{Notions}} of {{Quantum Min-Entropy}}},}\ }\href@noop {} {\Eprint {http://arxiv.org/abs/1704.07309} {arXiv:1704.07309}  (\bibinfo {year} {2017})}\BibitemShut {NoStop}%
\bibitem[Che23]{cheng_2023-1}
\bibfield  {author} {\bibinfo {author} {\bibfnamefont {H.-C.}\ \bibnamefont {Cheng}},\ }\bibfield  {title} {\emph {\bibinfo {title} {Simple and {{Tighter Derivation}} of {{Achievability}} for {{Classical Communication Over Quantum Channels}}},}\ }\href {http://dx.doi.org/10.1103/PRXQuantum.4.040330} {\bibfield  {journal} {\bibinfo  {journal} {PRX Quantum}\ }\textbf {\bibinfo {volume} {4}},\ \bibinfo {pages} {040330} (\bibinfo {year} {2023})}\BibitemShut {NoStop}%
\bibitem[CL25]{cheng_2025}
\bibfield  {author} {\bibinfo {author} {\bibfnamefont {H.-C.}\ \bibnamefont {Cheng}}\ and\ \bibinfo {author} {\bibfnamefont {P.-C.}\ \bibnamefont {Liu}},\ }\bibfield  {title} {\emph {\bibinfo {title} {Error {{Exponents}} for {{Quantum Packing Problems}} via {{An Operator Layer Cake Theorem}}},}\ }\href@noop {} {\Eprint {http://arxiv.org/abs/2507.06232} {arXiv:2507.06232}  (\bibinfo {year} {2025})}\BibitemShut {NoStop}%
\bibitem[CW79]{carter_1979}
\bibfield  {author} {\bibinfo {author} {\bibfnamefont {J.~L.}\ \bibnamefont {Carter}}\ and\ \bibinfo {author} {\bibfnamefont {M.~N.}\ \bibnamefont {Wegman}},\ }\bibfield  {title} {\emph {\bibinfo {title} {Universal classes of hash functions},}\ }\href {http://dx.doi.org/10.1016/0022-0000(79)90044-8} {\bibfield  {journal} {\bibinfo  {journal} {J. Comput. Syst. Sci.}\ }\textbf {\bibinfo {volume} {18}},\ \bibinfo {pages} {143} (\bibinfo {year} {1979})}\BibitemShut {NoStop}%
\bibitem[Dat09]{datta_2009}
\bibfield  {author} {\bibinfo {author} {\bibfnamefont {N.}~\bibnamefont {Datta}},\ }\bibfield  {title} {\emph {\bibinfo {title} {Min- and {{Max-Relative Entropies}} and a {{New Entanglement Monotone}}},}\ }\href {http://dx.doi.org/10.1109/TIT.2009.2018325} {\bibfield  {journal} {\bibinfo  {journal} {IEEE Trans. Inf. Theory}\ }\textbf {\bibinfo {volume} {55}},\ \bibinfo {pages} {2816} (\bibinfo {year} {2009})}\BibitemShut {NoStop}%
\bibitem[DBWR14]{dupuis_2014}
\bibfield  {author} {\bibinfo {author} {\bibfnamefont {F.}~\bibnamefont {Dupuis}}, \bibinfo {author} {\bibfnamefont {M.}~\bibnamefont {Berta}}, \bibinfo {author} {\bibfnamefont {J.}~\bibnamefont {Wullschleger}}, \ and\ \bibinfo {author} {\bibfnamefont {R.}~\bibnamefont {Renner}},\ }\bibfield  {title} {\emph {\bibinfo {title} {One-{{Shot Decoupling}}},}\ }\href {http://dx.doi.org/10.1007/s00220-014-1990-4} {\bibfield  {journal} {\bibinfo  {journal} {Commun. Math. Phys.}\ }\textbf {\bibinfo {volume} {328}},\ \bibinfo {pages} {251} (\bibinfo {year} {2014})}\BibitemShut {NoStop}%
\bibitem[DFR20]{dupuis_2020}
\bibfield  {author} {\bibinfo {author} {\bibfnamefont {F.}~\bibnamefont {Dupuis}}, \bibinfo {author} {\bibfnamefont {O.}~\bibnamefont {Fawzi}}, \ and\ \bibinfo {author} {\bibfnamefont {R.}~\bibnamefont {Renner}},\ }\bibfield  {title} {\emph {\bibinfo {title} {Entropy {{Accumulation}}},}\ }\href {http://dx.doi.org/10.1007/s00220-020-03839-5} {\bibfield  {journal} {\bibinfo  {journal} {Commun. Math. Phys.}\ }\textbf {\bibinfo {volume} {379}},\ \bibinfo {pages} {867} (\bibinfo {year} {2020})}\BibitemShut {NoStop}%
\bibitem[DL14]{datta_2014}
\bibfield  {author} {\bibinfo {author} {\bibfnamefont {N.}~\bibnamefont {Datta}}\ and\ \bibinfo {author} {\bibfnamefont {F.}~\bibnamefont {Leditzky}},\ }\bibfield  {title} {\emph {\bibinfo {title} {A limit of the quantum {{R{\'e}nyi}} divergence},}\ }\href {http://dx.doi.org/10.1088/1751-8113/47/4/045304} {\bibfield  {journal} {\bibinfo  {journal} {J. Phys. A: Math. Theor.}\ }\textbf {\bibinfo {volume} {47}},\ \bibinfo {pages} {045304} (\bibinfo {year} {2014})}\BibitemShut {NoStop}%
\bibitem[DL15]{datta_2015}
\bibfield  {author} {\bibinfo {author} {\bibfnamefont {N.}~\bibnamefont {Datta}}\ and\ \bibinfo {author} {\bibfnamefont {F.}~\bibnamefont {Leditzky}},\ }\bibfield  {title} {\emph {\bibinfo {title} {Second-{{Order Asymptotics}} for {{Source Coding}}, {{Dense Coding}}, and {{Pure-State Entanglement Conversions}}},}\ }\href {http://dx.doi.org/10.1109/TIT.2014.2366994} {\bibfield  {journal} {\bibinfo  {journal} {IEEE Trans. Inf. Theory}\ }\textbf {\bibinfo {volume} {61}},\ \bibinfo {pages} {582} (\bibinfo {year} {2015})}\BibitemShut {NoStop}%
\bibitem[Don86]{donald_1986}
\bibfield  {author} {\bibinfo {author} {\bibfnamefont {M.~J.}\ \bibnamefont {Donald}},\ }\bibfield  {title} {\emph {\bibinfo {title} {On the relative entropy},}\ }\href {http://dx.doi.org/10.1007/BF01212339} {\bibfield  {journal} {\bibinfo  {journal} {Commun. Math. Phys.}\ }\textbf {\bibinfo {volume} {105}},\ \bibinfo {pages} {13} (\bibinfo {year} {1986})}\BibitemShut {NoStop}%
\bibitem[DR09]{datta_2009-1}
\bibfield  {author} {\bibinfo {author} {\bibfnamefont {N.}~\bibnamefont {Datta}}\ and\ \bibinfo {author} {\bibfnamefont {R.}~\bibnamefont {Renner}},\ }\bibfield  {title} {\emph {\bibinfo {title} {Smooth {{Entropies}} and the {{Quantum Information Spectrum}}},}\ }\href {http://dx.doi.org/10.1109/TIT.2009.2018340} {\bibfield  {journal} {\bibinfo  {journal} {IEEE Trans. Inf. Theory}\ }\textbf {\bibinfo {volume} {55}},\ \bibinfo {pages} {2807} (\bibinfo {year} {2009})}\BibitemShut {NoStop}%
\bibitem[Dup09]{dupuis_2010}
\bibfield  {author} {\bibinfo {author} {\bibfnamefont {F.}~\bibnamefont {Dupuis}},\ }\emph {\bibinfo {title} {The Decoupling Approach to Quantum Information Theory}},\ \href@noop {} {Ph.D. thesis},\ \bibinfo  {school} {Universit\'e de Montr\'eal} (\bibinfo {year} {2009}),\ \Eprint {http://arxiv.org/abs/1004.1641} {arXiv:1004.1641} \BibitemShut {NoStop}%
\bibitem[Dup23]{dupuis_2023}
\bibfield  {author} {\bibinfo {author} {\bibfnamefont {F.}~\bibnamefont {Dupuis}},\ }\bibfield  {title} {\emph {\bibinfo {title} {Privacy {{Amplification}} and {{Decoupling Without Smoothing}}},}\ }\href {http://dx.doi.org/10.1109/TIT.2023.3301812} {\bibfield  {journal} {\bibinfo  {journal} {IEEE Trans. Inf. Theory}\ }\textbf {\bibinfo {volume} {69}},\ \bibinfo {pages} {7784} (\bibinfo {year} {2023})}\BibitemShut {NoStop}%
\bibitem[FSWR25]{ferradini_2025}
\bibfield  {author} {\bibinfo {author} {\bibfnamefont {C.}~\bibnamefont {Ferradini}}, \bibinfo {author} {\bibfnamefont {M.}~\bibnamefont {Sandfuchs}}, \bibinfo {author} {\bibfnamefont {R.}~\bibnamefont {Wolf}}, \ and\ \bibinfo {author} {\bibfnamefont {R.}~\bibnamefont {Renner}},\ }\bibfield  {title} {\emph {\bibinfo {title} {Defining {{Security}} in {{Quantum Key Distribution}}},}\ }\href@noop {} {\Eprint {http://arxiv.org/abs/2509.13405} {arXiv:2509.13405}  (\bibinfo {year} {2025})}\BibitemShut {NoStop}%
\bibitem[Fuc96]{fuchs_1996}
\bibfield  {author} {\bibinfo {author} {\bibfnamefont {C.~A.}\ \bibnamefont {Fuchs}},\ }\emph {\bibinfo {title} {Distinguishability and {{Accessible Information}} in {{Quantum Theory}}}},\ \href@noop {} {Ph.D. thesis},\ \bibinfo  {school} {University of New Mexico} (\bibinfo {year} {1996}),\ \Eprint {http://arxiv.org/abs/quant-ph/9601020} {arXiv:quant-ph/9601020} \BibitemShut {NoStop}%
\bibitem[Fv99]{fuchs_1999}
\bibfield  {author} {\bibinfo {author} {\bibfnamefont {C.}~\bibnamefont {Fuchs}}\ and\ \bibinfo {author} {\bibfnamefont {J.}~\bibnamefont {{van de Graaf}}},\ }\bibfield  {title} {\emph {\bibinfo {title} {Cryptographic distinguishability measures for quantum-mechanical states},}\ }\href {http://dx.doi.org/10.1109/18.761271} {\bibfield  {journal} {\bibinfo  {journal} {IEEE Trans. Inf. Theory}\ }\textbf {\bibinfo {volume} {45}},\ \bibinfo {pages} {1216} (\bibinfo {year} {1999})}\BibitemShut {NoStop}%
\bibitem[Gal73]{gallager_1973}
\bibfield  {author} {\bibinfo {author} {\bibfnamefont {R.}~\bibnamefont {Gallager}},\ }\bibfield  {title} {\emph {\bibinfo {title} {The random coding bound is tight for the average code},}\ }\href {http://dx.doi.org/10.1109/TIT.1973.1054971} {\bibfield  {journal} {\bibinfo  {journal} {IEEE Trans. Inf. Theory}\ }\textbf {\bibinfo {volume} {19}},\ \bibinfo {pages} {244} (\bibinfo {year} {1973})}\BibitemShut {NoStop}%
\bibitem[GHVK11]{grigoryan_2011}
\bibfield  {author} {\bibinfo {author} {\bibfnamefont {N.}~\bibnamefont {Grigoryan}}, \bibinfo {author} {\bibfnamefont {A.}~\bibnamefont {Harutyunyan}}, \bibinfo {author} {\bibfnamefont {S.}~\bibnamefont {Voloshynovskiy}}, \ and\ \bibinfo {author} {\bibfnamefont {O.}~\bibnamefont {Koval}},\ }\bibfield  {title} {\emph {\bibinfo {title} {On multiple hypothesis testing with rejection option},}\ }in\ \href {http://dx.doi.org/10.1109/ITW.2011.6089531} {\emph {\bibinfo {booktitle} {2011 {{IEEE Information Theory Workshop}}}}}\ (\bibinfo {year} {2011})\ pp.\ \bibinfo {pages} {75--79}\BibitemShut {NoStop}%
\bibitem[Gou26]{gour_2026}
\bibfield  {author} {\bibinfo {author} {\bibfnamefont {G.}~\bibnamefont {Gour}},\ }\bibfield  {title} {\emph {\bibinfo {title} {Optimal {{Universal Bounds}} for {{Quantum Divergences}}},}\ }\href@noop {} {\Eprint {http://arxiv.org/abs/2603.09885} {arXiv:2603.09885}  (\bibinfo {year} {2026})}\BibitemShut {NoStop}%
\bibitem[Gut89]{gutman_1989}
\bibfield  {author} {\bibinfo {author} {\bibfnamefont {M.}~\bibnamefont {Gutman}},\ }\bibfield  {title} {\emph {\bibinfo {title} {Asymptotically optimal classification for multiple tests with empirically observed statistics},}\ }\href {http://dx.doi.org/10.1109/18.32134} {\bibfield  {journal} {\bibinfo  {journal} {IEEE Trans. Inf. Theory}\ }\textbf {\bibinfo {volume} {35}},\ \bibinfo {pages} {401} (\bibinfo {year} {1989})}\BibitemShut {NoStop}%
\bibitem[Hay07]{hayashi_2007}
\bibfield  {author} {\bibinfo {author} {\bibfnamefont {M.}~\bibnamefont {Hayashi}},\ }\bibfield  {title} {\emph {\bibinfo {title} {Error exponent in asymmetric quantum hypothesis testing and its application to classical-quantum channel coding},}\ }\href {http://dx.doi.org/10.1103/PhysRevA.76.062301} {\bibfield  {journal} {\bibinfo  {journal} {Phys. Rev. A}\ }\textbf {\bibinfo {volume} {76}},\ \bibinfo {pages} {062301} (\bibinfo {year} {2007})}\BibitemShut {NoStop}%
\bibitem[Hay13]{hayashi_2013}
\bibfield  {author} {\bibinfo {author} {\bibfnamefont {M.}~\bibnamefont {Hayashi}},\ }\bibfield  {title} {\emph {\bibinfo {title} {Tight exponential analysis of universally composable privacy amplification and its applications},}\ }\href {http://dx.doi.org/10.1109/TIT.2013.2278971} {\bibfield  {journal} {\bibinfo  {journal} {IEEE Trans. Inf. Theory}\ }\textbf {\bibinfo {volume} {59}},\ \bibinfo {pages} {7728} (\bibinfo {year} {2013})}\BibitemShut {NoStop}%
\bibitem[Hay14]{hayashi_2014}
\bibfield  {author} {\bibinfo {author} {\bibfnamefont {M.}~\bibnamefont {Hayashi}},\ }\bibfield  {title} {\emph {\bibinfo {title} {Large {{Deviation Analysis}} for {{Quantum Security}} via {{Smoothing}} of {{R{\'e}nyi Entropy}} of {{Order}} 2},}\ }\href {http://dx.doi.org/10.1109/TIT.2014.2337884} {\bibfield  {journal} {\bibinfo  {journal} {IEEE Trans. Inf. Theory}\ }\textbf {\bibinfo {volume} {60}},\ \bibinfo {pages} {6702} (\bibinfo {year} {2014})}\BibitemShut {NoStop}%
\bibitem[Hay15]{hayashi_2015}
\bibfield  {author} {\bibinfo {author} {\bibfnamefont {M.}~\bibnamefont {Hayashi}},\ }\bibfield  {title} {\emph {\bibinfo {title} {Precise {{Evaluation}} of {{Leaked Information}} with {{Secure Randomness Extraction}} in the {{Presence}} of {{Quantum Attacker}}},}\ }\href {http://dx.doi.org/10.1007/s00220-014-2174-y} {\bibfield  {journal} {\bibinfo  {journal} {Commun. Math. Phys.}\ }\textbf {\bibinfo {volume} {333}},\ \bibinfo {pages} {335} (\bibinfo {year} {2015})}\BibitemShut {NoStop}%
\bibitem[Hay16]{hayashi_2016-2}
\bibfield  {author} {\bibinfo {author} {\bibfnamefont {M.}~\bibnamefont {Hayashi}},\ }\bibfield  {title} {\emph {\bibinfo {title} {Security analysis of {${\varepsilon}$}-almost dual universal{${_2}$} hash functions: Smoothing of min entropy versus smoothing of {R}ényi entropy of order 2},}\ }\href {http://dx.doi.org/10.1109/TIT.2016.2535174} {\bibfield  {journal} {\bibinfo  {journal} {IEEE Trans. Inf. Theory}\ }\textbf {\bibinfo {volume} {62}},\ \bibinfo {pages} {3451} (\bibinfo {year} {2016})}\BibitemShut {NoStop}%
\bibitem[Hay17]{hayashi_2016}
\bibfield  {author} {\bibinfo {author} {\bibfnamefont {M.}~\bibnamefont {Hayashi}},\ }\href@noop {} {\emph {\bibinfo {title} {Quantum {{Information Theory}}: {{Mathematical Foundation}}}}}\ (\bibinfo  {publisher} {Springer},\ \bibinfo {year} {2017})\BibitemShut {NoStop}%
\bibitem[Hel69]{helstrom_1969}
\bibfield  {author} {\bibinfo {author} {\bibfnamefont {C.~W.}\ \bibnamefont {Helstrom}},\ }\bibfield  {title} {\emph {\bibinfo {title} {Quantum detection and estimation theory},}\ }\href {http://dx.doi.org/10.1007/BF01007479} {\bibfield  {journal} {\bibinfo  {journal} {J. Stat. Phys.}\ }\textbf {\bibinfo {volume} {1}},\ \bibinfo {pages} {231} (\bibinfo {year} {1969})}\BibitemShut {NoStop}%
\bibitem[HILL99]{hastad_1999}
\bibfield  {author} {\bibinfo {author} {\bibfnamefont {J.}~\bibnamefont {H{\aa}stad}}, \bibinfo {author} {\bibfnamefont {R.}~\bibnamefont {Impagliazzo}}, \bibinfo {author} {\bibfnamefont {L.~A.}\ \bibnamefont {Levin}}, \ and\ \bibinfo {author} {\bibfnamefont {M.}~\bibnamefont {Luby}},\ }\bibfield  {title} {\emph {\bibinfo {title} {A {{Pseudorandom Generator}} from any {{One-way Function}}},}\ }\href {http://dx.doi.org/10.1137/S0097539793244708} {\bibfield  {journal} {\bibinfo  {journal} {SIAM J. Comput.}\ }\textbf {\bibinfo {volume} {28}},\ \bibinfo {pages} {1364} (\bibinfo {year} {1999})}\BibitemShut {NoStop}%
\bibitem[HKDW22]{hanson_2022}
\bibfield  {author} {\bibinfo {author} {\bibfnamefont {E.~P.}\ \bibnamefont {Hanson}}, \bibinfo {author} {\bibfnamefont {V.}~\bibnamefont {Katariya}}, \bibinfo {author} {\bibfnamefont {N.}~\bibnamefont {Datta}}, \ and\ \bibinfo {author} {\bibfnamefont {M.~M.}\ \bibnamefont {Wilde}},\ }\bibfield  {title} {\emph {\bibinfo {title} {Guesswork {{With Quantum Side Information}}},}\ }\href {http://dx.doi.org/10.1109/TIT.2021.3118878} {\bibfield  {journal} {\bibinfo  {journal} {IEEE Trans. Inf. Theory}\ }\textbf {\bibinfo {volume} {68}},\ \bibinfo {pages} {322} (\bibinfo {year} {2022})}\BibitemShut {NoStop}%
\bibitem[HM17]{hiai_2017}
\bibfield  {author} {\bibinfo {author} {\bibfnamefont {F.}~\bibnamefont {Hiai}}\ and\ \bibinfo {author} {\bibfnamefont {M.}~\bibnamefont {Mosonyi}},\ }\bibfield  {title} {\emph {\bibinfo {title} {Different quantum f-divergences and the reversibility of quantum operations},}\ }\href {http://dx.doi.org/10.1142/S0129055X17500234} {\bibfield  {journal} {\bibinfo  {journal} {Rev. Math. Phys.}\ }\textbf {\bibinfo {volume} {29}},\ \bibinfo {pages} {1750023} (\bibinfo {year} {2017})}\BibitemShut {NoStop}%
\bibitem[HMPB11]{hiai_2011}
\bibfield  {author} {\bibinfo {author} {\bibfnamefont {F.}~\bibnamefont {Hiai}}, \bibinfo {author} {\bibfnamefont {M.}~\bibnamefont {Mosonyi}}, \bibinfo {author} {\bibfnamefont {D.}~\bibnamefont {Petz}}, \ and\ \bibinfo {author} {\bibfnamefont {C.}~\bibnamefont {Beny}},\ }\bibfield  {title} {\emph {\bibinfo {title} {Quantum f-divergences and error correction},}\ }\href {http://dx.doi.org/10.1142/S0129055X11004412} {\bibfield  {journal} {\bibinfo  {journal} {Rev. Math. Phys.}\ }\textbf {\bibinfo {volume} {23}},\ \bibinfo {pages} {691} (\bibinfo {year} {2011})}\BibitemShut {NoStop}%
\bibitem[Hol73]{holevo_1973}
\bibfield  {author} {\bibinfo {author} {\bibfnamefont {A.~S.}\ \bibnamefont {Holevo}},\ }\bibfield  {title} {\emph {\bibinfo {title} {Statistical decision theory for quantum systems},}\ }\href {http://dx.doi.org/10.1016/0047-259X(73)90028-6} {\bibfield  {journal} {\bibinfo  {journal} {J. Multivar. Anal.}\ }\textbf {\bibinfo {volume} {3}},\ \bibinfo {pages} {337} (\bibinfo {year} {1973})}\BibitemShut {NoStop}%
\bibitem[HOW05]{horodecki_2005-1}
\bibfield  {author} {\bibinfo {author} {\bibfnamefont {M.}~\bibnamefont {Horodecki}}, \bibinfo {author} {\bibfnamefont {J.}~\bibnamefont {Oppenheim}}, \ and\ \bibinfo {author} {\bibfnamefont {A.}~\bibnamefont {Winter}},\ }\bibfield  {title} {\emph {\bibinfo {title} {Partial quantum information},}\ }\href {http://dx.doi.org/10.1038/nature03909} {\bibfield  {journal} {\bibinfo  {journal} {Nature}\ }\textbf {\bibinfo {volume} {436}},\ \bibinfo {pages} {673} (\bibinfo {year} {2005})}\BibitemShut {NoStop}%
\bibitem[HP91]{hiai_1991}
\bibfield  {author} {\bibinfo {author} {\bibfnamefont {F.}~\bibnamefont {Hiai}}\ and\ \bibinfo {author} {\bibfnamefont {D.}~\bibnamefont {Petz}},\ }\bibfield  {title} {\emph {\bibinfo {title} {The proper formula for relative entropy and its asymptotics in quantum probability},}\ }\href {http://dx.doi.org/10.1007/BF02100287} {\bibfield  {journal} {\bibinfo  {journal} {Commun. Math. Phys.}\ }\textbf {\bibinfo {volume} {143}},\ \bibinfo {pages} {99} (\bibinfo {year} {1991})}\BibitemShut {NoStop}%
\bibitem[HP12]{hiai_2012}
\bibfield  {author} {\bibinfo {author} {\bibfnamefont {F.}~\bibnamefont {Hiai}}\ and\ \bibinfo {author} {\bibfnamefont {D.}~\bibnamefont {Petz}},\ }\bibfield  {title} {\emph {\bibinfo {title} {From quasi-entropy to various quantum information quantities},}\ }\href {http://dx.doi.org/10.2977/prims/79} {\bibfield  {journal} {\bibinfo  {journal} {Publ. Res. Inst. Math. Sci.}\ }\textbf {\bibinfo {volume} {48}},\ \bibinfo {pages} {525} (\bibinfo {year} {2012})}\BibitemShut {NoStop}%
\bibitem[HRF23]{hirche_2023}
\bibfield  {author} {\bibinfo {author} {\bibfnamefont {C.}~\bibnamefont {Hirche}}, \bibinfo {author} {\bibfnamefont {C.}~\bibnamefont {Rouz{\'e}}}, \ and\ \bibinfo {author} {\bibfnamefont {D.~S.}\ \bibnamefont {Fran{\ccc c}a}},\ }\bibfield  {title} {\emph {\bibinfo {title} {Quantum {{Differential Privacy}}: {{An Information Theory Perspective}}},}\ }\href {http://dx.doi.org/10.1109/TIT.2023.3272904} {\bibfield  {journal} {\bibinfo  {journal} {IEEE Trans. Inf. Theory}\ }\textbf {\bibinfo {volume} {69}},\ \bibinfo {pages} {5771} (\bibinfo {year} {2023})}\BibitemShut {NoStop}%
\bibitem[HT16]{hayashi_2016-1}
\bibfield  {author} {\bibinfo {author} {\bibfnamefont {M.}~\bibnamefont {Hayashi}}\ and\ \bibinfo {author} {\bibfnamefont {M.}~\bibnamefont {Tomamichel}},\ }\bibfield  {title} {\emph {\bibinfo {title} {Correlation detection and an operational interpretation of the {{R{\'e}nyi}} mutual information},}\ }\href {http://dx.doi.org/10.1063/1.4964755} {\bibfield  {journal} {\bibinfo  {journal} {J. Math. Phys.}\ }\textbf {\bibinfo {volume} {57}},\ \bibinfo {pages} {102201} (\bibinfo {year} {2016})}\BibitemShut {NoStop}%
\bibitem[HW16]{hayashi_2016-3}
\bibfield  {author} {\bibinfo {author} {\bibfnamefont {M.}~\bibnamefont {Hayashi}}\ and\ \bibinfo {author} {\bibfnamefont {S.}~\bibnamefont {Watanabe}},\ }\bibfield  {title} {\emph {\bibinfo {title} {Uniform {{Random Number Generation From Markov Chains}}: {{Non-Asymptotic}} and {{Asymptotic Analyses}}},}\ }\href {http://dx.doi.org/10.1109/TIT.2016.2530084} {\bibfield  {journal} {\bibinfo  {journal} {IEEE Trans. Inf. Theory}\ }\textbf {\bibinfo {volume} {62}},\ \bibinfo {pages} {1795} (\bibinfo {year} {2016})}\BibitemShut {NoStop}%
\bibitem[HW25]{huang_2025}
\bibfield  {author} {\bibinfo {author} {\bibfnamefont {Z.}~\bibnamefont {Huang}}\ and\ \bibinfo {author} {\bibfnamefont {M.~M.}\ \bibnamefont {Wilde}},\ }\bibfield  {title} {\emph {\bibinfo {title} {Accelerated optimization of measured relative entropies},}\ }\href@noop {} {\Eprint {http://arxiv.org/abs/2511.17976} {arXiv:2511.17976}  (\bibinfo {year} {2025})}\BibitemShut {NoStop}%
\bibitem[ILL89]{impagliazzo_1989}
\bibfield  {author} {\bibinfo {author} {\bibfnamefont {R.}~\bibnamefont {Impagliazzo}}, \bibinfo {author} {\bibfnamefont {L.~A.}\ \bibnamefont {Levin}}, \ and\ \bibinfo {author} {\bibfnamefont {M.}~\bibnamefont {Luby}},\ }\bibfield  {title} {\emph {\bibinfo {title} {Pseudo-random generation from one-way functions},}\ }in\ \href {http://dx.doi.org/10.1145/73007.73009} {\emph {\bibinfo {booktitle} {Proceedings of the Twenty-First Annual {{ACM}} Symposium on {{Theory}} of Computing}}}\ (\bibinfo {year} {1989})\ pp.\ \bibinfo {pages} {12--24}\BibitemShut {NoStop}%
\bibitem[Iva87]{ivanovic_1987}
\bibfield  {author} {\bibinfo {author} {\bibfnamefont {I.~D.}\ \bibnamefont {Ivanovic}},\ }\bibfield  {title} {\emph {\bibinfo {title} {How to differentiate between non-orthogonal states},}\ }\href {http://dx.doi.org/10.1016/0375-9601(87)90222-2} {\bibfield  {journal} {\bibinfo  {journal} {Phys. Lett. A}\ }\textbf {\bibinfo {volume} {123}},\ \bibinfo {pages} {257} (\bibinfo {year} {1987})}\BibitemShut {NoStop}%
\bibitem[JR25]{ji_2025}
\bibfield  {author} {\bibinfo {author} {\bibfnamefont {K.}~\bibnamefont {Ji}}\ and\ \bibinfo {author} {\bibfnamefont {B.}~\bibnamefont {Regula}},\ }\bibfield  {title} {\emph {\bibinfo {title} {Beyond {{Hoeffding}} and {{Chernoff}}: {{Trading}} conclusiveness for advantages in quantum hypothesis testing},}\ }\href@noop {} {\Eprint {http://arxiv.org/abs/2510.07601} {arXiv:2510.07601}  (\bibinfo {year} {2025})}\BibitemShut {NoStop}%
\bibitem[Kra94]{krawczyk_1994}
\bibfield  {author} {\bibinfo {author} {\bibfnamefont {H.}~\bibnamefont {Krawczyk}},\ }\bibfield  {title} {\emph {\bibinfo {title} {{{LFSR-based Hashing}} and {{Authentication}}},}\ }in\ \href {http://dx.doi.org/10.1007/3-540-48658-5_15} {\emph {\bibinfo {booktitle} {Advances in {{Cryptology}} --- {{CRYPTO}} '94}}},\ \bibinfo {editor} {edited by\ \bibinfo {editor} {\bibfnamefont {Y.~G.}\ \bibnamefont {Desmedt}}}\ (\bibinfo  {publisher} {Springer},\ \bibinfo {address} {Berlin},\ \bibinfo {year} {1994})\ pp.\ \bibinfo {pages} {129--139}\BibitemShut {NoStop}%
\bibitem[KRBM07]{konig_2007}
\bibfield  {author} {\bibinfo {author} {\bibfnamefont {R.}~\bibnamefont {K{\"o}nig}}, \bibinfo {author} {\bibfnamefont {R.}~\bibnamefont {Renner}}, \bibinfo {author} {\bibfnamefont {A.}~\bibnamefont {Bariska}}, \ and\ \bibinfo {author} {\bibfnamefont {U.}~\bibnamefont {Maurer}},\ }\bibfield  {title} {\emph {\bibinfo {title} {Small {{Accessible Quantum Information Does Not Imply Security}}},}\ }\href {http://dx.doi.org/10.1103/PhysRevLett.98.140502} {\bibfield  {journal} {\bibinfo  {journal} {Phys. Rev. Lett.}\ }\textbf {\bibinfo {volume} {98}},\ \bibinfo {pages} {140502} (\bibinfo {year} {2007})}\BibitemShut {NoStop}%
\bibitem[KRS09]{konig_2009}
\bibfield  {author} {\bibinfo {author} {\bibfnamefont {R.}~\bibnamefont {Konig}}, \bibinfo {author} {\bibfnamefont {R.}~\bibnamefont {Renner}}, \ and\ \bibinfo {author} {\bibfnamefont {C.}~\bibnamefont {Schaffner}},\ }\bibfield  {title} {\emph {\bibinfo {title} {The {{Operational Meaning}} of {{Min-}} and {{Max-Entropy}}},}\ }\href {http://dx.doi.org/10.1109/TIT.2009.2025545} {\bibfield  {journal} {\bibinfo  {journal} {IEEE Trans. Inf. Theory}\ }\textbf {\bibinfo {volume} {55}},\ \bibinfo {pages} {4337} (\bibinfo {year} {2009})}\BibitemShut {NoStop}%
\bibitem[Li14]{li_2014}
\bibfield  {author} {\bibinfo {author} {\bibfnamefont {K.}~\bibnamefont {Li}},\ }\bibfield  {title} {\emph {\bibinfo {title} {Second-order asymptotics for quantum hypothesis testing},}\ }\href {http://dx.doi.org/10.1214/13-AOS1185} {\bibfield  {journal} {\bibinfo  {journal} {Ann. Stat.}\ }\textbf {\bibinfo {volume} {42}},\ \bibinfo {pages} {171} (\bibinfo {year} {2014})}\BibitemShut {NoStop}%
\bibitem[LLY26]{li_2025-1}
\bibfield  {author} {\bibinfo {author} {\bibfnamefont {S.-B.}\ \bibnamefont {Li}}, \bibinfo {author} {\bibfnamefont {K.}~\bibnamefont {Li}}, \ and\ \bibinfo {author} {\bibfnamefont {L.}~\bibnamefont {Yu}},\ }\bibfield  {title} {\emph {\bibinfo {title} {Two-{{Parameter R\'enyi Information Quantities With Applications}} to {{Privacy Amplification}} and {{Soft Covering}}},}\ }\href {http://dx.doi.org/10.1109/TIT.2026.3698938} {\bibfield  {journal} {\bibinfo  {journal} {IEEE Trans. Inf. Theory}\ }\textbf {\bibinfo {volume} {72}},\ \bibinfo {pages} {5370} (\bibinfo {year} {2026})}\BibitemShut {NoStop}%
\bibitem[LR99]{lesniewski_1999}
\bibfield  {author} {\bibinfo {author} {\bibfnamefont {A.}~\bibnamefont {Lesniewski}}\ and\ \bibinfo {author} {\bibfnamefont {M.~B.}\ \bibnamefont {Ruskai}},\ }\bibfield  {title} {\emph {\bibinfo {title} {Monotone {{Riemannian}} metrics and relative entropy on noncommutative probability spaces},}\ }\href {http://dx.doi.org/10.1063/1.533053} {\bibfield  {journal} {\bibinfo  {journal} {J. Math. Phys.}\ }\textbf {\bibinfo {volume} {40}},\ \bibinfo {pages} {5702} (\bibinfo {year} {1999})}\BibitemShut {NoStop}%
\bibitem[LYH23]{li_2023}
\bibfield  {author} {\bibinfo {author} {\bibfnamefont {K.}~\bibnamefont {Li}}, \bibinfo {author} {\bibfnamefont {Y.}~\bibnamefont {Yao}}, \ and\ \bibinfo {author} {\bibfnamefont {M.}~\bibnamefont {Hayashi}},\ }\bibfield  {title} {\emph {\bibinfo {title} {Tight {{Exponential Analysis}} for {{Smoothing}} the {{Max-Relative Entropy}} and for {{Quantum Privacy Amplification}}},}\ }\href {http://dx.doi.org/10.1109/TIT.2022.3217671} {\bibfield  {journal} {\bibinfo  {journal} {IEEE Trans. Inf. Theory}\ }\textbf {\bibinfo {volume} {69}},\ \bibinfo {pages} {1680} (\bibinfo {year} {2023})}\BibitemShut {NoStop}%
\bibitem[Mat14]{matsumoto_2014}
\bibfield  {author} {\bibinfo {author} {\bibfnamefont {K.}~\bibnamefont {Matsumoto}},\ }\bibfield  {title} {\emph {\bibinfo {title} {On maximization of measured {$f$}-divergence between a given pair of quantum states},}\ }\href@noop {} {\Eprint {http://arxiv.org/abs/1412.3676} {arXiv:1412.3676}  (\bibinfo {year} {2014})}\BibitemShut {NoStop}%
\bibitem[MDS{\etalchar{+}}13]{muller-lennert_2013}
\bibfield  {author} {\bibinfo {author} {\bibfnamefont {M.}~\bibnamefont {{M{\"u}ller-Lennert}}}, \bibinfo {author} {\bibfnamefont {F.}~\bibnamefont {Dupuis}}, \bibinfo {author} {\bibfnamefont {O.}~\bibnamefont {Szehr}}, \bibinfo {author} {\bibfnamefont {S.}~\bibnamefont {Fehr}}, \ and\ \bibinfo {author} {\bibfnamefont {M.}~\bibnamefont {Tomamichel}},\ }\bibfield  {title} {\emph {\bibinfo {title} {On quantum {{R{\'e}nyi}} entropies: {{A}} new generalization and some properties},}\ }\href {http://dx.doi.org/10.1063/1.4838856} {\bibfield  {journal} {\bibinfo  {journal} {J. Math. Phys.}\ }\textbf {\bibinfo {volume} {54}},\ \bibinfo {pages} {122203} (\bibinfo {year} {2013})}\BibitemShut {NoStop}%
\bibitem[MFSR22]{metger_2022}
\bibfield  {author} {\bibinfo {author} {\bibfnamefont {T.}~\bibnamefont {Metger}}, \bibinfo {author} {\bibfnamefont {O.}~\bibnamefont {Fawzi}}, \bibinfo {author} {\bibfnamefont {D.}~\bibnamefont {Sutter}}, \ and\ \bibinfo {author} {\bibfnamefont {R.}~\bibnamefont {Renner}},\ }\bibfield  {title} {\emph {\bibinfo {title} {Generalised entropy accumulation},}\ }in\ \href {http://dx.doi.org/10.1109/FOCS54457.2022.00085} {\emph {\bibinfo {booktitle} {2022 {{IEEE}} 63rd {{Annual Symposium}} on {{Foundations}} of {{Computer Science}} ({{FOCS}})}}}\ (\bibinfo {year} {2022})\ pp.\ \bibinfo {pages} {844--850}\BibitemShut {NoStop}%
\bibitem[MG99]{moreland_1999}
\bibfield  {author} {\bibinfo {author} {\bibfnamefont {T.}~\bibnamefont {Moreland}}\ and\ \bibinfo {author} {\bibfnamefont {S.}~\bibnamefont {Gudder}},\ }\bibfield  {title} {\emph {\bibinfo {title} {Infima of {{Hilbert}} space effects},}\ }\href {http://dx.doi.org/10.1016/S0024-3795(98)10119-2} {\bibfield  {journal} {\bibinfo  {journal} {Linear Algebra Its Appl.}\ }\textbf {\bibinfo {volume} {286}},\ \bibinfo {pages} {1} (\bibinfo {year} {1999})}\BibitemShut {NoStop}%
\bibitem[MO15]{mosonyi_2015}
\bibfield  {author} {\bibinfo {author} {\bibfnamefont {M.}~\bibnamefont {Mosonyi}}\ and\ \bibinfo {author} {\bibfnamefont {T.}~\bibnamefont {Ogawa}},\ }\bibfield  {title} {\emph {\bibinfo {title} {Quantum {{Hypothesis Testing}} and the {{Operational Interpretation}} of the {{Quantum R{\'e}nyi Relative Entropies}}},}\ }\href {http://dx.doi.org/10.1007/s00220-014-2248-x} {\bibfield  {journal} {\bibinfo  {journal} {Commun. Math. Phys.}\ }\textbf {\bibinfo {volume} {334}},\ \bibinfo {pages} {1617} (\bibinfo {year} {2015})}\BibitemShut {NoStop}%
\bibitem[MWW09]{matthews_2009}
\bibfield  {author} {\bibinfo {author} {\bibfnamefont {W.}~\bibnamefont {Matthews}}, \bibinfo {author} {\bibfnamefont {S.}~\bibnamefont {Wehner}}, \ and\ \bibinfo {author} {\bibfnamefont {A.}~\bibnamefont {Winter}},\ }\bibfield  {title} {\emph {\bibinfo {title} {Distinguishability of {{Quantum States Under Restricted Families}} of {{Measurements}} with an {{Application}} to {{Quantum Data Hiding}}},}\ }\href {http://dx.doi.org/10.1007/s00220-009-0890-5} {\bibfield  {journal} {\bibinfo  {journal} {Commun. Math. Phys.}\ }\textbf {\bibinfo {volume} {291}},\ \bibinfo {pages} {813} (\bibinfo {year} {2009})}\BibitemShut {NoStop}%
\bibitem[NGW24]{nuradha_2024}
\bibfield  {author} {\bibinfo {author} {\bibfnamefont {T.}~\bibnamefont {Nuradha}}, \bibinfo {author} {\bibfnamefont {Z.}~\bibnamefont {Goldfeld}}, \ and\ \bibinfo {author} {\bibfnamefont {M.~M.}\ \bibnamefont {Wilde}},\ }\bibfield  {title} {\emph {\bibinfo {title} {Quantum {{Pufferfish Privacy}}: {{A Flexible Privacy Framework}} for {{Quantum Systems}}},}\ }\href {http://dx.doi.org/10.1109/TIT.2024.3404927} {\bibfield  {journal} {\bibinfo  {journal} {IEEE Trans. Inf. Theory}\ }\textbf {\bibinfo {volume} {70}},\ \bibinfo {pages} {5731} (\bibinfo {year} {2024})}\BibitemShut {NoStop}%
\bibitem[NH07]{nagaoka_2007}
\bibfield  {author} {\bibinfo {author} {\bibfnamefont {H.}~\bibnamefont {Nagaoka}}\ and\ \bibinfo {author} {\bibfnamefont {M.}~\bibnamefont {Hayashi}},\ }\bibfield  {title} {\emph {\bibinfo {title} {An {{Information-Spectrum Approach}} to {{Classical}} and {{Quantum Hypothesis Testing}} for {{Simple Hypotheses}}},}\ }\href {http://dx.doi.org/10.1109/TIT.2006.889463} {\bibfield  {journal} {\bibinfo  {journal} {IEEE Trans. Inf. Theory}\ }\textbf {\bibinfo {volume} {53}},\ \bibinfo {pages} {534} (\bibinfo {year} {2007})}\BibitemShut {NoStop}%
\bibitem[Nis96]{nisan_1996}
\bibfield  {author} {\bibinfo {author} {\bibfnamefont {N.}~\bibnamefont {Nisan}},\ }\bibfield  {title} {\emph {\bibinfo {title} {Extracting randomness: How and why. {{A}} survey},}\ }in\ \href {http://dx.doi.org/10.1109/CCC.1996.507667} {\emph {\bibinfo {booktitle} {Proceedings of {{Computational Complexity}} ({{Formerly Structure}} in {{Complexity Theory}})}}}\ (\bibinfo {year} {1996})\ pp.\ \bibinfo {pages} {44--58}\BibitemShut {NoStop}%
\bibitem[PAB{\etalchar{+}}20]{pirandola_2020}
\bibfield  {author} {\bibinfo {author} {\bibfnamefont {S.}~\bibnamefont {Pirandola}}, \bibinfo {author} {\bibfnamefont {U.~L.}\ \bibnamefont {Andersen}}, \bibinfo {author} {\bibfnamefont {L.}~\bibnamefont {Banchi}}, \bibinfo {author} {\bibfnamefont {M.}~\bibnamefont {Berta}}, \bibinfo {author} {\bibfnamefont {D.}~\bibnamefont {Bunandar}}, \bibinfo {author} {\bibfnamefont {R.}~\bibnamefont {Colbeck}}, \bibinfo {author} {\bibfnamefont {D.}~\bibnamefont {Englund}}, \bibinfo {author} {\bibfnamefont {T.}~\bibnamefont {Gehring}}, \bibinfo {author} {\bibfnamefont {C.}~\bibnamefont {Lupo}}, \bibinfo {author} {\bibfnamefont {C.}~\bibnamefont {Ottaviani}}, \bibinfo {author} {\bibfnamefont {J.~L.}\ \bibnamefont {Pereira}}, \bibinfo {author} {\bibfnamefont {M.}~\bibnamefont {Razavi}}, \bibinfo {author} {\bibfnamefont {J.~S.}\ \bibnamefont {Shaari}}, \bibinfo {author} {\bibfnamefont {M.}~\bibnamefont {Tomamichel}}, \bibinfo {author} {\bibfnamefont {V.~C.}\ \bibnamefont {Usenko}}, \bibinfo {author} {\bibfnamefont
  {G.}~\bibnamefont {Vallone}}, \bibinfo {author} {\bibfnamefont {P.}~\bibnamefont {Villoresi}}, \ and\ \bibinfo {author} {\bibfnamefont {P.}~\bibnamefont {Wallden}},\ }\bibfield  {title} {\emph {\bibinfo {title} {Advances in quantum cryptography},}\ }\href {http://dx.doi.org/10.1364/AOP.361502} {\bibfield  {journal} {\bibinfo  {journal} {Adv. Opt. Photon., AOP}\ }\textbf {\bibinfo {volume} {12}},\ \bibinfo {pages} {1012} (\bibinfo {year} {2020})}\BibitemShut {NoStop}%
\bibitem[Per88]{peres_1988}
\bibfield  {author} {\bibinfo {author} {\bibfnamefont {A.}~\bibnamefont {Peres}},\ }\bibfield  {title} {\emph {\bibinfo {title} {How to differentiate between non-orthogonal states},}\ }\href {http://dx.doi.org/10.1016/0375-9601(88)91034-1} {\bibfield  {journal} {\bibinfo  {journal} {Phys. Lett. A}\ }\textbf {\bibinfo {volume} {128}},\ \bibinfo {pages} {19} (\bibinfo {year} {1988})}\BibitemShut {NoStop}%
\bibitem[Pet86]{petz_1986}
\bibfield  {author} {\bibinfo {author} {\bibfnamefont {D.}~\bibnamefont {Petz}},\ }\bibfield  {title} {\emph {\bibinfo {title} {Quasi-entropies for finite quantum systems},}\ }\href {http://dx.doi.org/10.1016/0034-4877(86)90067-4} {\bibfield  {journal} {\bibinfo  {journal} {Rep. Math. Phys.}\ }\textbf {\bibinfo {volume} {23}},\ \bibinfo {pages} {57} (\bibinfo {year} {1986})}\BibitemShut {NoStop}%
\bibitem[Pet96]{petz_1996}
\bibfield  {author} {\bibinfo {author} {\bibfnamefont {D.}~\bibnamefont {Petz}},\ }\bibfield  {title} {\emph {\bibinfo {title} {Monotone metrics on matrix spaces},}\ }\href {http://dx.doi.org/10.1016/0024-3795(94)00211-8} {\bibfield  {journal} {\bibinfo  {journal} {Lin. Alg. Appl.}\ }\textbf {\bibinfo {volume} {244}},\ \bibinfo {pages} {81} (\bibinfo {year} {1996})}\BibitemShut {NoStop}%
\bibitem[PG11]{petz_2011}
\bibfield  {author} {\bibinfo {author} {\bibfnamefont {D.}~\bibnamefont {Petz}}\ and\ \bibinfo {author} {\bibfnamefont {C.}~\bibnamefont {Ghinea}},\ }\bibfield  {title} {\emph {\bibinfo {title} {Introduction to quantum {{Fisher}} information},}\ }in\ \href {https://www.worldscientific.com/doi/abs/10.1142/9789814338745_0015} {\emph {\bibinfo {booktitle} {Quantum Probability and Related Topics}}}\ (\bibinfo {year} {2011})\ pp.\ \bibinfo {pages} {261--281},\ \Eprint {http://arxiv.org/abs/1008.2417} {arXiv:1008.2417} \BibitemShut {NoStop}%
\bibitem[PR22]{portmann_2022}
\bibfield  {author} {\bibinfo {author} {\bibfnamefont {C.}~\bibnamefont {Portmann}}\ and\ \bibinfo {author} {\bibfnamefont {R.}~\bibnamefont {Renner}},\ }\bibfield  {title} {\emph {\bibinfo {title} {Security in quantum cryptography},}\ }\href {http://dx.doi.org/10.1103/RevModPhys.94.025008} {\bibfield  {journal} {\bibinfo  {journal} {Rev. Mod. Phys.}\ }\textbf {\bibinfo {volume} {94}},\ \bibinfo {pages} {025008} (\bibinfo {year} {2022})}\BibitemShut {NoStop}%
\bibitem[Ras06]{rastegin_2006}
\bibfield  {author} {\bibinfo {author} {\bibfnamefont {A.~E.}\ \bibnamefont {Rastegin}},\ }\bibfield  {title} {\emph {\bibinfo {title} {Sine distance for quantum states},}\ }\href@noop {} {\Eprint {http://arxiv.org/abs/quant-ph/0602112} {arXiv:quant-ph/0602112}  (\bibinfo {year} {2006})}\BibitemShut {NoStop}%
\bibitem[Ren05]{renner_2005}
\bibfield  {author} {\bibinfo {author} {\bibfnamefont {R.}~\bibnamefont {Renner}},\ }\emph {\bibinfo {title} {Security of {{Quantum Key Distribution}}}},\ \href@noop {} {Ph.D. thesis},\ \bibinfo  {school} {ETH Zurich} (\bibinfo {year} {2005}),\ \Eprint {http://arxiv.org/abs/quant-ph/0512258} {arXiv:quant-ph/0512258} \BibitemShut {NoStop}%
\bibitem[Ren18]{renes_2018}
\bibfield  {author} {\bibinfo {author} {\bibfnamefont {J.~M.}\ \bibnamefont {Renes}},\ }\bibfield  {title} {\emph {\bibinfo {title} {On {{Privacy Amplification}}, {{Lossy Compression}}, and {{Their Duality}} to {{Channel Coding}}},}\ }\href {http://dx.doi.org/10.1109/TIT.2018.2865386} {\bibfield  {journal} {\bibinfo  {journal} {IEEE Trans. Inf. Theory}\ }\textbf {\bibinfo {volume} {64}},\ \bibinfo {pages} {7792} (\bibinfo {year} {2018})}\BibitemShut {NoStop}%
\bibitem[RK05]{renner_2005-1}
\bibfield  {author} {\bibinfo {author} {\bibfnamefont {R.}~\bibnamefont {Renner}}\ and\ \bibinfo {author} {\bibfnamefont {R.}~\bibnamefont {K{\"o}nig}},\ }\bibfield  {title} {\emph {\bibinfo {title} {Universally {{Composable Privacy Amplification Against Quantum Adversaries}}},}\ }in\ \href {http://dx.doi.org/10.1007/978-3-540-30576-7_22} {\emph {\bibinfo {booktitle} {Theory of {{Cryptography}}}}},\ \bibinfo {editor} {edited by\ \bibinfo {editor} {\bibfnamefont {J.}~\bibnamefont {Kilian}}}\ (\bibinfo  {publisher} {Springer},\ \bibinfo {address} {Berlin},\ \bibinfo {year} {2005})\ pp.\ \bibinfo {pages} {407--425}\BibitemShut {NoStop}%
\bibitem[RLD26]{regula_2025}
\bibfield  {author} {\bibinfo {author} {\bibfnamefont {B.}~\bibnamefont {Regula}}, \bibinfo {author} {\bibfnamefont {L.}~\bibnamefont {Lami}}, \ and\ \bibinfo {author} {\bibfnamefont {N.}~\bibnamefont {Datta}},\ }\bibfield  {title} {\emph {\bibinfo {title} {Tight {{Relations}} and {{Equivalences Between Smooth Relative Entropies}}},}\ }\href {http://dx.doi.org/10.1109/TIT.2026.3661711} {\bibfield  {journal} {\bibinfo  {journal} {IEEE Trans. Inf. Theory}\ }\textbf {\bibinfo {volume} {72}},\ \bibinfo {pages} {3051} (\bibinfo {year} {2026})}\BibitemShut {NoStop}%
\bibitem[RT26]{rubboli_2026}
\bibfield  {author} {\bibinfo {author} {\bibfnamefont {R.}~\bibnamefont {Rubboli}}\ and\ \bibinfo {author} {\bibfnamefont {M.}~\bibnamefont {Tomamichel}},\ }\bibfield  {title} {\emph {\bibinfo {title} {The strong converse exponent of composable randomness extraction against quantum side information},}\ }\href@noop {} {\Eprint {http://arxiv.org/abs/2601.19182} {arXiv:2601.19182}  (\bibinfo {year} {2026})}\BibitemShut {NoStop}%
\bibitem[RW04]{renner_2004}
\bibfield  {author} {\bibinfo {author} {\bibfnamefont {R.}~\bibnamefont {Renner}}\ and\ \bibinfo {author} {\bibfnamefont {S.}~\bibnamefont {Wolf}},\ }\bibfield  {title} {\emph {\bibinfo {title} {Smooth {{Renyi}} entropy and applications},}\ }in\ \href {http://dx.doi.org/10.1109/ISIT.2004.1365269} {\emph {\bibinfo {booktitle} {2004 {{IEEE International Symposium}} on {{Information Theory}} ({{ISIT}})}}}\ (\bibinfo {year} {2004})\ p.\ \bibinfo {pages} {233}\BibitemShut {NoStop}%
\bibitem[RW05]{renner_2005-2}
\bibfield  {author} {\bibinfo {author} {\bibfnamefont {R.}~\bibnamefont {Renner}}\ and\ \bibinfo {author} {\bibfnamefont {S.}~\bibnamefont {Wolf}},\ }\bibfield  {title} {\emph {\bibinfo {title} {Simple and {{Tight Bounds}} for {{Information Reconciliation}} and {{Privacy Amplification}}},}\ }in\ \href {http://dx.doi.org/10.1007/11593447_11} {\emph {\bibinfo {booktitle} {Advances in {{Cryptology}} - {{ASIACRYPT}} 2005}}},\ \bibinfo {editor} {edited by\ \bibinfo {editor} {\bibfnamefont {B.}~\bibnamefont {Roy}}}\ (\bibinfo  {publisher} {Springer},\ \bibinfo {address} {Berlin},\ \bibinfo {year} {2005})\ pp.\ \bibinfo {pages} {199--216}\BibitemShut {NoStop}%
\bibitem[SD22]{salzmann_2022}
\bibfield  {author} {\bibinfo {author} {\bibfnamefont {R.}~\bibnamefont {Salzmann}}\ and\ \bibinfo {author} {\bibfnamefont {N.}~\bibnamefont {Datta}},\ }\bibfield  {title} {\emph {\bibinfo {title} {Total insecurity of communication via strong converse for quantum privacy amplification},}\ }\href {http://arxiv.org/abs/2202.11090} {\Eprint {http://arxiv.org/abs/2202.11090} {arXiv:2202.11090}  (\bibinfo {year} {2022})}\BibitemShut {NoStop}%
\bibitem[SGC22]{shen_2022}
\bibfield  {author} {\bibinfo {author} {\bibfnamefont {Y.-C.}\ \bibnamefont {Shen}}, \bibinfo {author} {\bibfnamefont {L.}~\bibnamefont {Gao}}, \ and\ \bibinfo {author} {\bibfnamefont {H.-C.}\ \bibnamefont {Cheng}},\ }\bibfield  {title} {\emph {\bibinfo {title} {Strong {{Converse}} for {{Privacy Amplification}} against {{Quantum Side Information}}},}\ }\href@noop {} {\Eprint {http://arxiv.org/abs/2202.10263} {arXiv:2202.10263}  (\bibinfo {year} {2022})}\BibitemShut {NoStop}%
\bibitem[SGC23]{shen_2023}
\bibfield  {author} {\bibinfo {author} {\bibfnamefont {Y.-C.}\ \bibnamefont {Shen}}, \bibinfo {author} {\bibfnamefont {L.}~\bibnamefont {Gao}}, \ and\ \bibinfo {author} {\bibfnamefont {H.-C.}\ \bibnamefont {Cheng}},\ }\bibfield  {title} {\emph {\bibinfo {title} {Optimal {{Second-Order Rates}} for {{Quantum Information Decoupling}}},}\ }in\ \href {http://dx.doi.org/10.1109/ISIT54713.2023.10206502} {\emph {\bibinfo {booktitle} {2023 {{IEEE International Symposium}} on {{Information Theory}} ({{ISIT}})}}}\ (\bibinfo {year} {2023})\ pp.\ \bibinfo {pages} {991--996}\BibitemShut {NoStop}%
\bibitem[SGC24]{shen_2024}
\bibfield  {author} {\bibinfo {author} {\bibfnamefont {Y.-C.}\ \bibnamefont {Shen}}, \bibinfo {author} {\bibfnamefont {L.}~\bibnamefont {Gao}}, \ and\ \bibinfo {author} {\bibfnamefont {H.-C.}\ \bibnamefont {Cheng}},\ }\bibfield  {title} {\emph {\bibinfo {title} {Optimal {{Second-Order Rates}} for {{Quantum Soft Covering}} and {{Privacy Amplification}}},}\ }\href {http://dx.doi.org/10.1109/TIT.2024.3351963} {\bibfield  {journal} {\bibinfo  {journal} {IEEE Trans. Inf. Theory}\ }\textbf {\bibinfo {volume} {70}},\ \bibinfo {pages} {5077} (\bibinfo {year} {2024})}\BibitemShut {NoStop}%
\bibitem[Sio58]{sion_1958}
\bibfield  {author} {\bibinfo {author} {\bibfnamefont {M.}~\bibnamefont {Sion}},\ }\bibfield  {title} {\emph {\bibinfo {title} {On general minimax theorems},}\ }\href {https://msp.org/pjm/1958/8-1/p14.xhtml} {\bibfield  {journal} {\bibinfo  {journal} {Pac. J. Math.}\ }\textbf {\bibinfo {volume} {8}},\ \bibinfo {pages} {171} (\bibinfo {year} {1958})}\BibitemShut {NoStop}%
\bibitem[SK20]{sidhu_2020}
\bibfield  {author} {\bibinfo {author} {\bibfnamefont {J.~S.}\ \bibnamefont {Sidhu}}\ and\ \bibinfo {author} {\bibfnamefont {P.}~\bibnamefont {Kok}},\ }\bibfield  {title} {\emph {\bibinfo {title} {Geometric perspective on quantum parameter estimation},}\ }\href {http://dx.doi.org/10.1116/1.5119961} {\bibfield  {journal} {\bibinfo  {journal} {AVS Quantum Sci.}\ }\textbf {\bibinfo {volume} {2}},\ \bibinfo {pages} {014701} (\bibinfo {year} {2020})}\BibitemShut {NoStop}%
\bibitem[Sti94]{stinson_1994}
\bibfield  {author} {\bibinfo {author} {\bibfnamefont {D.~R.}\ \bibnamefont {Stinson}},\ }\bibfield  {title} {\emph {\bibinfo {title} {Universal hashing and authentication codes},}\ }\href {http://dx.doi.org/10.1007/BF01388651} {\bibfield  {journal} {\bibinfo  {journal} {Des. Codes Crypt.}\ }\textbf {\bibinfo {volume} {4}},\ \bibinfo {pages} {369} (\bibinfo {year} {1994})}\BibitemShut {NoStop}%
\bibitem[SW13]{sharma_2013}
\bibfield  {author} {\bibinfo {author} {\bibfnamefont {N.}~\bibnamefont {Sharma}}\ and\ \bibinfo {author} {\bibfnamefont {N.~A.}\ \bibnamefont {Warsi}},\ }\bibfield  {title} {\emph {\bibinfo {title} {Fundamental {{Bound}} on the {{Reliability}} of {{Quantum Information Transmission}}},}\ }\href {http://dx.doi.org/10.1103/PhysRevLett.110.080501} {\bibfield  {journal} {\bibinfo  {journal} {Phys. Rev. Lett.}\ }\textbf {\bibinfo {volume} {110}},\ \bibinfo {pages} {080501} (\bibinfo {year} {2013})}\BibitemShut {NoStop}%
\bibitem[TCR10]{tomamichel_2010}
\bibfield  {author} {\bibinfo {author} {\bibfnamefont {M.}~\bibnamefont {Tomamichel}}, \bibinfo {author} {\bibfnamefont {R.}~\bibnamefont {Colbeck}}, \ and\ \bibinfo {author} {\bibfnamefont {R.}~\bibnamefont {Renner}},\ }\bibfield  {title} {\emph {\bibinfo {title} {Duality {{Between Smooth Min-}} and {{Max-Entropies}}},}\ }\href {http://dx.doi.org/10.1109/TIT.2010.2054130} {\bibfield  {journal} {\bibinfo  {journal} {IEEE Trans. Inf. Theory}\ }\textbf {\bibinfo {volume} {56}},\ \bibinfo {pages} {4674} (\bibinfo {year} {2010})}\BibitemShut {NoStop}%
\bibitem[TH13]{tomamichel_2013}
\bibfield  {author} {\bibinfo {author} {\bibfnamefont {M.}~\bibnamefont {Tomamichel}}\ and\ \bibinfo {author} {\bibfnamefont {M.}~\bibnamefont {Hayashi}},\ }\bibfield  {title} {\emph {\bibinfo {title} {A {{Hierarchy}} of {{Information Quantities}} for {{Finite Block Length Analysis}} of {{Quantum Tasks}}},}\ }\href {http://dx.doi.org/10.1109/TIT.2013.2276628} {\bibfield  {journal} {\bibinfo  {journal} {IEEE Trans. Inf. Theory}\ }\textbf {\bibinfo {volume} {59}},\ \bibinfo {pages} {7693} (\bibinfo {year} {2013})}\BibitemShut {NoStop}%
\bibitem[TKR{\etalchar{+}}10]{temme_2010}
\bibfield  {author} {\bibinfo {author} {\bibfnamefont {K.}~\bibnamefont {Temme}}, \bibinfo {author} {\bibfnamefont {M.~J.}\ \bibnamefont {Kastoryano}}, \bibinfo {author} {\bibfnamefont {M.~B.}\ \bibnamefont {Ruskai}}, \bibinfo {author} {\bibfnamefont {M.~M.}\ \bibnamefont {Wolf}}, \ and\ \bibinfo {author} {\bibfnamefont {F.}~\bibnamefont {Verstraete}},\ }\bibfield  {title} {\emph {\bibinfo {title} {The {$\chi^2$}-divergence and mixing times of quantum {M}arkov processes},}\ }\href {http://dx.doi.org/10.1063/1.3511335} {\bibfield  {journal} {\bibinfo  {journal} {J. Math. Phys.}\ }\textbf {\bibinfo {volume} {51}},\ \bibinfo {pages} {122201} (\bibinfo {year} {2010})}\BibitemShut {NoStop}%
\bibitem[TL17]{tomamichel_2017-1}
\bibfield  {author} {\bibinfo {author} {\bibfnamefont {M.}~\bibnamefont {Tomamichel}}\ and\ \bibinfo {author} {\bibfnamefont {A.}~\bibnamefont {Leverrier}},\ }\bibfield  {title} {\emph {\bibinfo {title} {A largely self-contained and complete security proof for quantum key distribution},}\ }\href {http://dx.doi.org/10.22331/q-2017-07-14-14} {\bibfield  {journal} {\bibinfo  {journal} {Quantum}\ }\textbf {\bibinfo {volume} {1}},\ \bibinfo {pages} {14} (\bibinfo {year} {2017})}\BibitemShut {NoStop}%
\bibitem[Tom16]{tomamichel_2016}
\bibfield  {author} {\bibinfo {author} {\bibfnamefont {M.}~\bibnamefont {Tomamichel}},\ }\href@noop {} {\emph {\bibinfo {title} {Quantum {{Information Processing}} with {{Finite Resources}}}}}\ (\bibinfo  {publisher} {Springer},\ \bibinfo {year} {2016})\BibitemShut {NoStop}%
\bibitem[TRT26]{tan_2026}
\bibfield  {author} {\bibinfo {author} {\bibfnamefont {M.~S.}\ \bibnamefont {Tan}}, \bibinfo {author} {\bibfnamefont {B.}~\bibnamefont {Regula}}, \ and\ \bibinfo {author} {\bibfnamefont {M.}~\bibnamefont {Tomamichel}},\ }\bibfield  {title} {\emph {\bibinfo {title} {Making the most of leftovers: {{Improved}} privacy amplification for quantum key distribution},}\ }\href@noop {} {\Eprint {http://arxiv.org/abs/2609.38306} {arXiv:2609.38306}  (\bibinfo {year} {2026})}\BibitemShut {NoStop}%
\bibitem[TSSR11]{tomamichel_2011}
\bibfield  {author} {\bibinfo {author} {\bibfnamefont {M.}~\bibnamefont {Tomamichel}}, \bibinfo {author} {\bibfnamefont {C.}~\bibnamefont {Schaffner}}, \bibinfo {author} {\bibfnamefont {A.}~\bibnamefont {Smith}}, \ and\ \bibinfo {author} {\bibfnamefont {R.}~\bibnamefont {Renner}},\ }\bibfield  {title} {\emph {\bibinfo {title} {Leftover {{Hashing Against Quantum Side Information}}},}\ }\href {http://dx.doi.org/10.1109/TIT.2011.2158473} {\bibfield  {journal} {\bibinfo  {journal} {IEEE Trans. Inf. Theory}\ }\textbf {\bibinfo {volume} {57}},\ \bibinfo {pages} {5524} (\bibinfo {year} {2011})}\BibitemShut {NoStop}%
\bibitem[TTN{\etalchar{+}}26]{tupkary_2025}
\bibfield  {author} {\bibinfo {author} {\bibfnamefont {D.}~\bibnamefont {Tupkary}}, \bibinfo {author} {\bibfnamefont {E.~Y.-Z.}\ \bibnamefont {Tan}}, \bibinfo {author} {\bibfnamefont {S.}~\bibnamefont {Nahar}}, \bibinfo {author} {\bibfnamefont {L.}~\bibnamefont {Kamin}}, \ and\ \bibinfo {author} {\bibfnamefont {N.}~\bibnamefont {L{\"u}tkenhaus}},\ }\bibfield  {title} {\emph {\bibinfo {title} {Security proofs for practical {{QKD}}: {{Variations}}, techniques, gaps, and limitations},}\ }\href {http://dx.doi.org/10.1103/28rs-frmw} {\bibfield  {journal} {\bibinfo  {journal} {Rev. Mod. Phys.}\ }\textbf {\bibinfo {volume} {98}},\ \bibinfo {pages} {035003} (\bibinfo {year} {2026})}\BibitemShut {NoStop}%
\bibitem[TV15]{temme_2015}
\bibfield  {author} {\bibinfo {author} {\bibfnamefont {K.}~\bibnamefont {Temme}}\ and\ \bibinfo {author} {\bibfnamefont {F.}~\bibnamefont {Verstraete}},\ }\bibfield  {title} {\emph {\bibinfo {title} {Quantum chi-squared and goodness of fit testing},}\ }\href {http://dx.doi.org/10.1063/1.4905843} {\bibfield  {journal} {\bibinfo  {journal} {J. Math. Phys.}\ }\textbf {\bibinfo {volume} {56}},\ \bibinfo {pages} {012202} (\bibinfo {year} {2015})}\BibitemShut {NoStop}%
\bibitem[WC81]{wegman_1981}
\bibfield  {author} {\bibinfo {author} {\bibfnamefont {M.~N.}\ \bibnamefont {Wegman}}\ and\ \bibinfo {author} {\bibfnamefont {J.~L.}\ \bibnamefont {Carter}},\ }\bibfield  {title} {\emph {\bibinfo {title} {New hash functions and their use in authentication and set equality},}\ }\href {http://dx.doi.org/10.1016/0022-0000(81)90033-7} {\bibfield  {journal} {\bibinfo  {journal} {J. Comput. Syst. Sci.}\ }\textbf {\bibinfo {volume} {22}},\ \bibinfo {pages} {265} (\bibinfo {year} {1981})}\BibitemShut {NoStop}%
\bibitem[WH13]{watanabe_2013}
\bibfield  {author} {\bibinfo {author} {\bibfnamefont {S.}~\bibnamefont {Watanabe}}\ and\ \bibinfo {author} {\bibfnamefont {M.}~\bibnamefont {Hayashi}},\ }\bibfield  {title} {\emph {\bibinfo {title} {Non-asymptotic analysis of privacy amplification via {{R{\'e}nyi}} entropy and inf-spectral entropy},}\ }in\ \href {http://dx.doi.org/10.1109/ISIT.2013.6620720} {\emph {\bibinfo {booktitle} {2013 {{IEEE International Symposium}} on {{Information Theory}}}}}\ (\bibinfo {year} {2013})\ pp.\ \bibinfo {pages} {2715--2719}\BibitemShut {NoStop}%
\bibitem[WR12]{wang_2012}
\bibfield  {author} {\bibinfo {author} {\bibfnamefont {L.}~\bibnamefont {Wang}}\ and\ \bibinfo {author} {\bibfnamefont {R.}~\bibnamefont {Renner}},\ }\bibfield  {title} {\emph {\bibinfo {title} {One-{{Shot Classical-Quantum Capacity}} and {{Hypothesis Testing}}},}\ }\href {http://dx.doi.org/10.1103/PhysRevLett.108.200501} {\bibfield  {journal} {\bibinfo  {journal} {Phys. Rev. Lett.}\ }\textbf {\bibinfo {volume} {108}},\ \bibinfo {pages} {200501} (\bibinfo {year} {2012})}\BibitemShut {NoStop}%
\bibitem[WWY14]{wilde_2014}
\bibfield  {author} {\bibinfo {author} {\bibfnamefont {M.~M.}\ \bibnamefont {Wilde}}, \bibinfo {author} {\bibfnamefont {A.}~\bibnamefont {Winter}}, \ and\ \bibinfo {author} {\bibfnamefont {D.}~\bibnamefont {Yang}},\ }\bibfield  {title} {\emph {\bibinfo {title} {Strong {{Converse}} for the {{Classical Capacity}} of {{Entanglement-Breaking}} and {{Hadamard Channels}} via a {{Sandwiched R{\'e}nyi Relative Entropy}}},}\ }\href {http://dx.doi.org/10.1007/s00220-014-2122-x} {\bibfield  {journal} {\bibinfo  {journal} {Commun. Math. Phys.}\ }\textbf {\bibinfo {volume} {331}},\ \bibinfo {pages} {593} (\bibinfo {year} {2014})}\BibitemShut {NoStop}%
\bibitem[YSP19]{yang_2019}
\bibfield  {author} {\bibinfo {author} {\bibfnamefont {W.}~\bibnamefont {Yang}}, \bibinfo {author} {\bibfnamefont {R.~F.}\ \bibnamefont {Schaefer}}, \ and\ \bibinfo {author} {\bibfnamefont {H.~V.}\ \bibnamefont {Poor}},\ }\bibfield  {title} {\emph {\bibinfo {title} {Wiretap {{Channels}}: {{Nonasymptotic Fundamental Limits}}},}\ }\href {http://dx.doi.org/10.1109/TIT.2019.2904500} {\bibfield  {journal} {\bibinfo  {journal} {IEEE Trans. Inf. Theory}\ }\textbf {\bibinfo {volume} {65}},\ \bibinfo {pages} {4069} (\bibinfo {year} {2019})}\BibitemShut {NoStop}%
\end{thebibliography}%



\let\addcontentsline\oldaddcontentsline

 \appendix
 \counterwithin{theorem}{section}

\section{Hypothesis testing relative entropy as measured smooth divergence}
\label{app:DH}

{ \renewcommand{\thetheorem}{\ref{lem:DH_as_D0}}
\begin{boxed}
\begin{lemma}\label{lem:DH_app_lemma}
For all quantum states $\rho$, all $\sigma \geq 0$, and all $\ve \in [0,1)$, the hypothesis testing relative entropy equals the measured smooth R\'enyi divergence of order 0:
\begin{equation}\begin{aligned}
  D^\ve_H (\rho \| \sigma) = \Ds{0}{\ve}{M} (\rho \| \sigma) = \sup_{\M\in\MM} \Ds{0}{\ve}{T}(\M(\rho)\|\M(\sigma)).
\end{aligned}\end{equation}
Here we recall that
\begin{equation}\begin{aligned}
  D_0(p\|q) = - \log \sum_{x: p(x) > 0} q(x).
\end{aligned}\end{equation}
\end{lemma}
\end{boxed}
}
We remark here that, even for classical distributions, the optimisation over measurements is necessary to obtain this equivalence; that is, $\Ds{0}{\ve}{T}(p\|q) \neq D_H^\ve(p\|q)$ in general. This is because $\Ds{0}{\ve}{T}$ can be understood as the error in hypothesis testing using only \emph{deterministic} tests, while $D_H^\ve$ allows also for randomised strategies. Introducing an optimisation over measurements, as in $\Ds{0}{\ve\vphantom{t}}{M}(p\|q)$, is then equivalent to accounting for randomisation. 
\begin{proof}[Proof of Lemma~\ref{lem:DH_app_lemma}]
Fix any measurement channel $\M$ with output alphabet $\X$, let $p = \M(\rho)$ and $q = \M(\sigma)$ be the post-measurement probability distributions, and let $p' \geq 0$ be any distribution such that $\norm{p-p'}{+} \leq \ve$. Define the set $S \coloneqq \lset x \bar p'(x) > 0 \rset$ and use it to define the test operator $M \coloneqq \M^\dagger(\boldsymbol{1}_S)$, where $\boldsymbol{1}_S$ denotes the indicator function of the event $x \in S$ (projection onto $S$). Then 
\begin{equation}\begin{aligned}
  1 - \Tr M \rho &= 1 - \Tr \boldsymbol{1}_{S} \M(\rho)\\
  &= \sum_{x \notin S} p(x) \\
  &= \sum_{x \notin S} p(x) - p'(x)\\
  &\leq \!\!\!\!\! \sum_{x:\, p'(x)\leq p(x)} \!\!\!\!\! p(x) - p'(x) \\
  &= \sum_{x} \left( p(x) - p'(x) \right)_+ \\&= \norm{p - p'}{+} \\&\leq \ve,
\end{aligned}\end{equation}
and furthermore
\begin{equation}\begin{aligned}
  - \log \Tr M \sigma &= - \log \sum_{x \in S} q(x) = D_0(p' \| q).
\end{aligned}\end{equation}
This altogether means that $M$ is a feasible test in the definition of $D^\ve_H$ in~\eqref{eq:dh_def} with feasible value $D_0(p' \| q)$, yielding $D^\ve_H(\rho\|\sigma) \geq \Ds{0}{\ve}{T} (\M(\rho)\|\M(\sigma))$ upon maximising over $p'$. Taking a supremum over measurement channels $\M$ establishes one direction of the inequality.

For the other direction, let $M \in [0,\id]$ be any test feasible for $D^\ve_H$, i.e.\ one such that $\Tr M \rho \geq 1-\ve$. Define the measurement channel $\M$ with a binary output alphabet as
\begin{equation}\begin{aligned}
   \M(X) \coloneqq \Tr (MX) \proj{0} + \Tr\big([\id-M]\,X\big) \proj{1}.
\end{aligned}\end{equation} 
Let $p = \M(\rho)$, $q = \M(\sigma)$, and choose $p' \coloneqq \big(p(0), 0\big) \leq p$. Then clearly $\norm{p-p'}{+} = 1 - \Tr M\rho \leq \ve$, and so
\begin{equation}\begin{aligned}
  \Ds{0}{\ve}{T}\big(\M(\rho) \| \M(\sigma)\big) &\geq D_0(p'\| q)\\
  &= - \log q(0)\\
  &= - \log \Tr M \sigma.
\end{aligned}\end{equation}
Optimising over feasible tests $M$ gives $\Ds{0}{\ve}{T}(\M(\rho) \| \M(\sigma)) \geq D^\ve_H(\rho\|\sigma)$, entailing that the two quantities are in fact equal.
\end{proof}

\section{Semidefinite representation of measured smooth conditional entropies}
\label{app:sdp}

Here we discuss how the conditional entropies used in our study, and in particular the smooth variant of the measured collision entropy, can be expressed as semidefinite programs (SDP).

Recall the semidefinite representation of $\Ds{2}{\ve}{M}$ from Corollary~\ref{lem:D2_SDP},
\begin{align}
 &\Ds{2}{\ve}{M}(\rho\|\sigma) \nonumber\\
 &= \log \inf_{Z,T,R} \lset \Tr T \bar T \geq 0,\; Z \in \mathbb{C}^{d \times d},\; \Re(Z) = R,\; R \leq \rho,\; \Tr(\rho-R) \leq \ve, \,\begin{pmatrix}\sigma & Z \\ Z^\dagger & T \end{pmatrix} \geq 0 \rset\label{eq:D2_smooth_SDP_primal}\\
 &= \log \sup_{B,C,t} \lset 2 \Tr B \rho - 2 \ve t - \Tr C \sigma \bar B, C \geq 0,\; B \leq t \id,\; \begin{pmatrix}C & B \\ B & \id \end{pmatrix} \geq 0 \rset. \label{eq:D2_smooth_SDP_dual}
\end{align}
We recall also a semidefinite representation of $\Dmax{\ve}{M}$~\cite{nuradha_2024,regula_2025}
\begin{align}
 \Dmax{\ve}{M}(\rho\|\sigma)  &= \log \inf_{\lambda, R} \lset \lambda \bar R \leq \lambda\sigma,\; R \leq \rho,\; \Tr(\rho-R) \leq \ve \rset\label{eq:Dmax_smooth_SDP_primal}\\
 &= \log \sup_{B,t} \lset \Tr B \rho - \ve t \bar B \geq 0,\; B \leq t \id,\; \Tr B \sigma \leq 1 \rset. \label{eq:Dmax_smooth_SDP_dual}
\end{align}

The SDP expression for the quantities $\Hs{2}{\ve}{M}[down](X|E)_\rho$ and $\Hmin{\ve}{M}[down](X|E)_\rho$ is then immediate: simply pick $\rho = \rho_{XE}$ and $\sigma = \id_X \otimes \rho_E$ in any of the SDP formulations of~\eqref{eq:D2_smooth_SDP_primal}--\eqref{eq:Dmax_smooth_SDP_dual}. Let us momentarily leave aside the variational expressions for $\Hmin{\ve}{M}[up](X|E)_\rho$ as we will return to it in the next section, connecting it to guessing probability.

For $H^{\MM,\uparrow}_{2}$, one can follow one of two approaches. In the form of~\eqref{eq:D2_smooth_SDP_primal}, we observe that $\sigma$ does not appear in the objective function and the constraints are linear in $\sigma$; we can then simply introduce the optimisation over $\sigma_E$ into the SDP and write out full program as
\begin{equation}\begin{aligned}
  - \Hs{2}{\ve}{M}[up](X|E)_\rho = \log \inf_{Z,T,R,\sigma} &\lset \vphantom{\begin{pmatrix}\id_X \otimes \sigma_E & Z_{XE} \\ Z_{XE}^\dagger & T_{XE} \end{pmatrix}}\Tr T_{XE} \bar T_{XE} \geq 0,\; Z_{XE} \in \mathbb{C}^{d \times d},\; \Re(Z_{XE}) = R_{XE},\; R_{XE} \leq \rho_{XE},\right.\\
 &\;\; \left.  \Tr(\rho_{XE}-R_{XE}) \leq \ve,\; \begin{pmatrix}\id_X \otimes \sigma_E & Z_{XE} \\ Z_{XE}^\dagger & T_{XE} \end{pmatrix} \geq 0 ,\; \sigma_E \geq 0,\; \Tr \sigma_E = 1 \rset
\end{aligned}\end{equation}
where $d$ is the dimension of the space of the joint system $XE$.
Another formulation can be obtained from the dual form of~\eqref{eq:D2_smooth_SDP_dual}. Let us consider first the case $\ve = 0$:
\begin{align}
  - H^{\MM,\uparrow}_{2}(X|E)_\rho &= \log \inf_{\substack{\sigma_E \geq 0\\\Tr \sigma_E = 1}} \sup_{B,C} \lset 2 \Tr B_{XE} \rho_{XE} - \Tr \left[ C_{XE} \left(\id_X \otimes \sigma_E\right)\right] \bar B_{XE}, C_{XE} \geq 0,\; \begin{pmatrix}C_{XE} & B_{XE} \\ B_{XE} & \id_{XE} \end{pmatrix} \geq 0 \rset\nonumber\\
  &= \log \vphantom{{\sup_{\substack{\sigma_E > 0\\\Tr \sigma_E = 1}}}}\sup_{B,C} \lset 2 \Tr B_{XE} \rho_{XE} - \smash{\sup_{\substack{\sigma_E \geq 0\\\Tr \sigma_E = 1}}} \Tr C_E \sigma_E \bar B_{XE}, C_{XE} \geq 0,\; \begin{pmatrix}C_{XE} & B_{XE} \\ B_{XE} & \id_{XE} \end{pmatrix} \geq 0 \rset\\
  &= \log \sup_{B,C,k} \lset 2 \Tr B_{XE} \rho_{XE} - k \bar B_{XE}, C_{XE} \geq 0,\; C_E \leq k \id_E ,\; \begin{pmatrix}C_{XE} & B_{XE} \\ B_{XE} & \id_{XE} \end{pmatrix} \geq 0\rset,\nonumber
\end{align}
where the second line is by Sion's minimax theorem~\cite{sion_1958}, and the last uses the fact that $\sup_{\sigma_E \geq 0, \Tr \sigma_E = 1} \Tr X_E \sigma_E = \lambda_{\max}(X_E)$ for any Hermitian $X_E$. Extending this to all $\ve$ gives
\begin{align}
  - \Hs{2}{\ve}{M}[up](X|E)_\rho &= \log \sup_{B,C,t,k} \lset 2 \Tr B_{XE} \rho_{XE} - 2 \ve t - k \bar \vphantom{\begin{pmatrix}C_{XE} & B_{XE} \\ B_{XE} & \id_{XE} \end{pmatrix}} B_{XE}, C_{XE} \geq 0,\, B_{XE} \leq t \id_{XE},\, C_E \leq k \id_E,\, \right.\nonumber
  \\ & \hphantom{\log \sup_{B,C,t,k} \lset 2 \Tr B_{XE} \rho_{XE} - 2 \ve t - k \bar \right. \quad} \left. \begin{pmatrix}C_{XE} & B_{XE} \\ B_{XE} & \id_{XE} \end{pmatrix} \geq 0 \rset \\
 &= \log \sup_{B,C} \lset 2 \Tr B_{XE} \rho_{XE} - 2 \ve \norm{B_{XE}}{\infty} - \norm{C_E}{\infty} \bar B_{XE}, C_{XE} \geq 0,\, \begin{pmatrix}C_{XE} & B_{XE} \\ B_{XE} & \id_{XE} \end{pmatrix} \geq 0 \rset. \nonumber
\end{align}
For the purpose of implementing this SDP in practice, it is useful to notice that if $\rho_{XE} = \sum_x p_X(x) \proj{x} \otimes \rho_{E,x}$ is a CQ state, then both $\rho_{XE}$ and $\id_X \otimes \sigma_E$ are invariant under the application of the dephasing channel $\Delta_X (Z_{XE}) = \sum_{x} \proj{x} Z_{XE} \proj{x}$, and hence $B_{XE}$ and $C_{XE}$ can be without loss of generality assumed to likewise be invariant --- that is, classical--quantum. We then get the blockwise form
\begin{equation}\begin{aligned}
  - \Hs{2}{\ve}{M}[up](X|E)_\rho &= \log \smash{\sup_{\substack{\{B_{E,x},C_{E,x}\}_x,\\t,k}}} \lset 2 \sum_x p_X(x) \Tr B_{E,x} \rho_{E,x} - 2 \ve t - k \bar \vphantom{\begin{pmatrix}C_{XE} & B_{XE} \\ B_{XE} & \id_{XE} \end{pmatrix}} B_{E,x}, C_{E,x} \geq 0 \; \forall x,\right.
  \\ & \hphantom{\log \sup_{\substack{\{B_{E,x},C_{E,x}\}_x,\\t,k}} \lset 2 \sum_x B_{E,x} \rho_{E,x} - 2 t \ve - k \bar\right. \quad} B_{E,x} \leq t \id_{E} \; \forall x,\;\,  \sum_x C_{E,x} \leq k \id_E,\,\\
   &\hphantom{\log \sup_{\substack{\{B_{E,x},C_{E,x}\}_x,\\t,k}} \lset 2 \sum_x B_{E,x} \rho_{E,x} - 2 t \ve - k \bar\right. \quad} \left. \begin{pmatrix}C_{E,x} & B_{E,x} \\ B_{E,x} & \id_{E} \end{pmatrix} \geq 0 \; \forall x\rset,
\end{aligned}\end{equation}
the advantage of which is especially that the large Schur complement constraint is explicitly expressed as $|\X|$ less demanding constraints on the system $E$ only.

With regards to computability, we also make note of the recent work~\cite{huang_2025} which suggests that there may be numerical implementations of measured R\'enyi divergences that are even more efficient than SDP variational forms. It would be interesting to understand whether this can also provide speedups in the numerical evaluation of the smooth conditional entropies considered here.

\section{Connection with guessing probability}
\label{app:guess_prob}

A conceptually important interpretation of $H_{\min}(X|E)_{\rho}$ is as the guessing probability, that is, the average probability of guessing the random variable $X$ given $E$~\cite{konig_2009}:
\begin{equation}\begin{aligned}
  p_{\rm guess}(X|E)_\rho \coloneqq \sup \lset \sum_x p_X(x) \Tr M_{E,x} \rho_{E,x} \bar M_{E,x} \geq 0 \; \forall x,\; \sum_x M_{E,x} = \id_E \rset, 
\end{aligned}\end{equation}
where $\rho_{XE} = \sum_x p_X(x) \proj{x} \otimes \rho_{E,x}$.

Let us first rederive this known result in an attempt to understand how to generalise it. We have by definition that
\begin{equation}\begin{aligned}
  \exp\!\left(- H^{\uparrow}_{\min}(X|E)_{\rho}\right) = \inf_{\lambda,\sigma} \lset \lambda \bar \rho_{XE} \leq \id_X \otimes \lambda \sigma_E, \; \lambda \in \RR_+,\; \sigma_E \geq 0,\; \Tr \sigma_E = 1 \rset.
\end{aligned}\end{equation}
Notice that the constraint $\sigma_E \geq 0$ is superfluous: since $\rho_{XE} \geq 0$, any feasible $\sigma_E$ satisfying $\rho_{XE} \leq \id_X \otimes \lambda \sigma_E$ must be positive semidefinite anyway; we can thus remove this condition without loss of generality. Letting $S_E = \lambda \sigma_E$, we then have
\begin{equation}\begin{aligned}
 \exp\!\left( - H^{\uparrow}_{\min}(X|E)_{\rho}\right) &= \inf_{S} \lset \Tr S_E \bar \rho_{XE} \leq \id_X \otimes S_E \rset\\
  &= \sup_{W} \lset \Tr W_{XE} \rho_{XE} \bar 0 \leq W_{XE},\; W_E = \id_E \rset
\end{aligned}\end{equation}
by strong Lagrange duality. As $\rho_{XE}$ is a CQ state, we can also restrict ourselves to CQ operators $W_{XE}$, which corresponds to an optimisation over $W_{XE} = \sum_x \proj{x} \otimes W_{E,x}$ such that $W_{E,x} \geq 0$ and $\sum_x W_{E,x} = \id_E$. But this is precisely an optimisation over all POVMs on $E$, matching the definition of $p_{\rm guess}$.

Recall now that the smoothed variant $\Hmin{\ve}{M}(X|E)_\rho$ corresponds to the similar optimisation
\begin{equation}\begin{aligned}
 \exp\!\left( - \Hmin{\ve}{M}(X|E)_{\rho}\right) &= \inf_{\lambda, R, \sigma} \lset \lambda \bar R_{XE} \leq \lambda \left(\id_X \otimes \sigma_E\right),\; R_{XE} \leq \rho_{XE},\; \Tr(\rho_{XE}-R_{XE}) \leq \ve ,\right.\\
  &\hphantom{= \inf_{\lambda, R, \sigma} \big\{ \lambda \quad} \left. \sigma_E \geq 0,\; \Tr \sigma_E = 1 \rset.
\end{aligned}\end{equation}
When attempting to extend the derivation from the $\ve=0$ case, we immediately run into an issue, as it is no longer possible to remove the condition $\sigma_E \geq 0$ without affecting the optimal value of the program. To adapt to this, one can slightly modify the definition of the guessing probability to only require that $\sum_x M_{E,x} \leq \id_E$, that is, that the collection $\{M_{E,x}\}_{x}$ is a part of a POVM and not necessarily a complete measurement. This can be understood e.g.\ as allowing an additional measurement outcome that represents inconclusive discrimination, as is common in settings such as unambiguous state discrimination~\cite{ivanovic_1987,peres_1988,bae_2015} or hypothesis testing with rejection~\cite{gutman_1989,grigoryan_2011,ji_2025}.

To connect such a guessing probability notion with the measured smooth min-entropy, we begin by using Lagrange duality to write
\begin{equation}\begin{aligned}
 \exp\!\left( - \Hmin{\ve}{M}(X|E)_{\rho}\right) &= \inf \lset \Tr S_E \bar \rho_{XE} \leq \id_X \otimes S_E + Q_{XE}, \; S_E \geq 0,\;  \; Q_{XE} \geq 0, \Tr Q_{XE} \leq \ve \rset\\
  &= \sup \lset \Tr W_{XE} \rho_{XE} - \ve t \bar 0 \leq W_{XE},\; W_E \leq \id_E,\; W_{XE} \leq t \id_{XE} \rset\\
  &= \sup \lset \Tr W_{XE} \rho_{XE} - \ve \norm{W_{XE}}{\infty} \bar 0 \leq W_{XE},\; W_E \leq \id_E \rset.
\end{aligned}\end{equation}
Here again we may restrict without loss of generality to classical--quantum $W_{XE}$, where in particular $\norm{W_{XE}}{\infty} = \max_x \norm{W_{E,x}}{\infty}$. 
Define now a parametrised variant of guessing probability as
\begin{equation}\begin{aligned}
  p'_{\rm guess}(X|E, \redd{t})_\rho \coloneqq \sup \lset \sum_x p_X(x) \Tr M_{E,x} \rho_{E,x} \bar 0 \leq M_{E,x} \leq \redd{t} \id_E \; \forall x,\; \sum_x M_{E,x} \leq \id_E \rset.
\end{aligned}\end{equation}
The parameter $t$ here constrains how `peaked' or sharp the measurement can be, as no single measurement outcome can ever occur with probability larger than $t$. 
From the above expressions we then obtain the following interpretation.
\begin{proposition}
For any CQ state $\rho_{XE}$ and any $\ve \in [0,1)$, the measured smooth min-entropy can be expressed as
\begin{equation}\begin{aligned}
   \exp\!\left( - \Hmin{\ve}{M}(X|E)_{\rho}\right) = \sup_{t \in [0,1]} \, \left[ p'_{\rm guess}(X|E, t)_\rho - \ve t \right].
\end{aligned}\end{equation}
\end{proposition}
The smoothing parameter $\ve$ is now playing a role similar to a Lagrange multiplier, enforcing a trade-off between the different optimisation variables and penalising the sharpness of the chosen measurement. 
While in many cases the optimum may be achieved at $t=1$ (e.g.\ with a projective measurement), this is not always the case; consider for instance the trivial $\rho_{XE} \propto \id_{XE}$, where the choice of measurement does not matter, so one picks the smallest $t$ instead.

Yet another way to write the optimisation would be
\begin{equation}\begin{aligned}
  \exp\!\left( - \Hmin{\ve}{M}(X|E)_{\rho}\right) = \sup \lset \frac{\Tr W_{XE} \rho_{XE} - \ve }{ \norm{W_{E}}{\infty} } \bar 0 \leq W_{XE} \leq \id_{XE} \rset.
\end{aligned}\end{equation}
This does not constrain $\{W_{E,x}\}_x$ to be a valid (sub-)POVM, but merely that each $W_{E,x}$ is a POVM element; the normalisation factor $ \norm{W_{E}}{\infty}$ then accounts for that by renormalising $\{W_{E,x}\}_x$ into a sub-POVM.


\end{document}